\documentclass[10pt,a4]{amsart}
\usepackage{verbatim}
\usepackage{enumerate}
\usepackage{hyperref}
\usepackage{makecell} 

\makeatletter
\@namedef{subjclassname@2020}{%
  \textup{2020} Mathematics Subject Classification}
\makeatother

\mathchardef\ordinarycolon\mathcode`\:
\mathcode`\:=\string"8000
\begingroup \catcode`\:=\active
  \gdef:{\mathrel{\mathop\ordinarycolon}}
\endgroup

\usepackage[sep=3pt, offset=1.2em]{simpler-wick} 
\usepackage{float} 

\usepackage{subcaption}

\usepackage{tikz}
\usepackage{tkz-euclide}
\usepackage{animate}
\usetikzlibrary{patterns}
\usetikzlibrary{arrows.meta}
\usetikzlibrary{quotes}
\tikzset{
  fermion/.style={draw=black, postaction={decorate},decoration={markings,mark=at position .55 with {\arrow{>}}}},
    bdry/.style={draw,shape=circle,fill=black,minimum size=5pt,inner sep=0pt},
  b1/.style={draw,shape=circle,minimum size=5pt,inner sep=0pt},
  b2/.style={draw,shape=circle,fill=gray!40,minimum size=5pt,inner sep=0pt},
  b3/.style={draw,shape=circle,fill=gray!80,minimum size=5pt,inner sep=0pt}
  }
  
\newcommand{\dd}{\partial}
\newcommand{\ra}{\rightarrow}
\newcommand{\bt}{\bullet}
\newcommand{\mr}{\mathrm}
\newcommand{\ZZ}{\mathbb{Z}}
\newcommand{\CC}{\mathbb{C}}
\newcommand{\RR}{\mathbb{R}}
\newcommand{\g}{\mathfrak{g}}
\newcommand{\frh}{\mathfrak{h}}
\newcommand{\FF}{\mathcal{F}}
\newcommand{\DD}{\mathcal{D}}
\newcommand{\BB}{\mathcal{B}}
\newcommand{\YY}{\mathcal{Y}}
\newcommand{\VV}{\mathcal{V}}
\newcommand{\PP}{\mathcal{P}}
\newcommand{\LL}{\mathcal{L}}
\newcommand{\mc}{\mathcal}
\newcommand{\til}{\widetilde}
\newcommand{\fl}{\mathrm{fl}}
\newcommand{\res}{\mathrm{res}}
\newcommand{\ii}{\mathrm{in}}
\newcommand{\oo}{\mathrm{out}}
\newcommand{\gh}{\mathrm{gh}}
\newcommand{\ph}{\mathrm{ph}}
\newcommand{\eff}{\mathrm{eff}}
\newcommand{\A}{\mathcal{A}}
\newcommand{\As}{\mathsf{A}}
\newcommand{\as}{\mathsf{a}}
\newcommand{\AsI}{\mathsf{A}_I}
\newcommand{\Ires}{{I\,\mathrm{res}}}
\newcommand{\D}{\mathcal{D}}
\newcommand{\const}{\mr{c}}
\newcommand{\ad}{\mathrm{ad}}
\newcommand{\nl}{\mathrm{nl}}
\newcommand{\ol}{\overline}
\newcommand{\ggg}{\mathcal{G}}
\newcommand{\calS}{\mathcal{S}}
\newcommand{\tr}{\mathrm{tr}}
\newcommand{\pol}{f}
\newcommand{\dt}{d_I}  
\newcommand{\PhiPM}{\mathcal{F}^\partial}  
\newcommand{\bl}{\textcolor{blue}}

\theoremstyle{remark}
\newtheorem{remark}{Remark}[section]
\newtheorem*{Ack}{Acknowledgments}
\theoremstyle{plain}
\newtheorem{lemma}[remark]{Lemma}
\newtheorem{proposition}[remark]{Proposition}
\newtheorem{corollary}[remark]{Corollary}
\newtheorem{theorem}{Theorem}

\theoremstyle{definition}

\newtheorem{example}[remark]{Example}

\begin{document}
\title[Chern--Simons on cylinders]{
Quantum Chern--Simons theories on cylinders: 
BV-BFV partition functions
}

\begin{abstract}
We compute partition functions of Chern--Simons type theories for cylindrical spacetimes $I \times \Sigma$, with $I$ an interval and $\dim \Sigma = 4l+2$, in the 
BV-BFV formalism (a refinement of the Batalin--Vilkovisky formalism adapted to manifolds with boundary and cutting--gluing). The case $\dim \Sigma = 0$ is considered as a toy example.  We show that one can identify---for certain choices of residual fields---the ``physical part'' (restriction to degree zero fields) of the BV-BFV effective action with the Hamilton--Jacobi action computed in the companion paper \cite{HJ}, without any quantum corrections. This Hamilton--Jacobi action is the action functional of a conformal field theory on $\Sigma$. For $\dim \Sigma = 2$, this implies a version of the CS-WZW correspondence. For $\dim \Sigma = 6$, using a particular polarization on one end of the cylinder, the Chern--Simons partition function 
is related to
Kodaira--Spencer gravity (a.k.a.\ BCOV theory); this provides a BV-BFV quantum perspective on the semiclassical 
result by Gerasimov and Shatashvili. 
\end{abstract}

\author{Alberto S. Cattaneo}
\address{Institut f\"ur Mathematik, Universit\"at Z\"urich\\
Winterthurerstrasse 190, CH-8057 Z\"urich, Switzerland}  
\email{cattaneo@math.uzh.ch}

\author{Pavel Mnev}

\address{University of Notre Dame}

\address{St. Petersburg Department of V. A. Steklov Institute of Mathematics of the Russian Academy of Sciences}
\email{pmnev@nd.edu}

\author{Konstantin Wernli}
\address{Centre for Quantum Mathematics, IMADA, University of Southern Denmark}
\email{kwernli@nd.edu}
\thanks{This research was (partly) supported by the NCCR SwissMAP, funded by the Swiss National Science Foundation. A.S.C. and K.W. acknowledge partial support of SNF Grants No.\ 200020\_192080 and 200020\_172498/1.  K. W. also acknowledges support from a BMS Dirichlet postdoctoral fellowship and the SNF Postdoc.Mobility grant P2ZHP2\_184083, and would like to thank the Humboldt-Universit\"at Berlin, in particular the group of Dirk Kreimer, and the university of Notre Dame for their hospitality. }


\maketitle

\setcounter{tocdepth}{3}
\tableofcontents

\allowdisplaybreaks

\section{Introduction}

This paper is a 
sequel to
the paper ``Constrained systems, generalized Hamilton--Jacobi actions, and quantization'' \cite{HJ} by the same authors (but can be read independently).

 As announced in \cite{HJ}, the main result of this paper is the explicit computation of the perturbative partition functions of  Chern--Simons theories on cylinders $I \times \Sigma$, with respect to various boundary polarizations. 
Their restriction to degree zero fields turns out to be
the exponential of the corresponding Hamilton--Jacobi action, defined in \cite{HJ} and recalled in Section \ref{sec:HJ}, without any quantum corrections. 

Interestingly, the Hamilton--Jacobi actions of the theories we consider can be related to action functionals of conformal field theories on $\Sigma$. This means that the partition function of Chern--Simons theories (with certain boundary conditions) can be identified with the partition function of a conformal field theory (coupled to sources) --- a property that one might call ``holographic duality.''  In that terminology, among other results, we show the following: 
\begin{itemize}
\item The holographic dual theory of 3D abelian Chern--Simons theory is the 2D free boson CFT, see Section \ref{sss: intro 3d ab CS} (while for a different choice of boundary polarization, we obtain the beta-gamma system as the dual, see (\ref{ab CS hol-hol Z via S_HJ})).
\item The holographic dual of 3D nonabelian Chern--Simons theory is WZW theory (see Section \ref{sss: intro 3d non ab CS}). In particular, the bulk-boundary version of the Batalin--Vilkovisky master equation (referred to below as the modified quantum master equation) corresponds to the Polyakov--Wiegmann formula for the WZW action functional.
\item The holographic dual of 7D Chern--Simons theory is a free 2-form theory for the ``standard'' polarization and the Kodaira--Spencer gravity  for a particular nonlinear polarization (see Section \ref{sss: intro 7d CS vs KS}). 
\end{itemize}
A remark on the terminology: The term ``holographic duality'' is often used for the case where the bulk theory is a theory of quantum gravity, e.g., a string theory, such as in the celebrated AdS/CFT correspondence \cite{Maldacena}, \cite{Witten3} and its more general variant, the gauge/gravity correspondence (see \cite{Erdmenger} for a review). The bulk/boundary correspondences we discuss below were, in a different context, discovered earlier, see for instance \cite{Witten}, \cite{GS}. They can be interpreted as special cases of holography, thinking of Chern-Simons theory as a string theory  \cite{Witten4}. 

The first motivating point for this paper and its prequel \cite{HJ}, suggested to us by S. Shatashvili, concerned precisely the last item in the list above: namely, the systematical understanding of  the relation between 7D abelian Chern--Simons theory and 6D Kodaira--Spencer \cite{KS} gravity (otherwise known as BCOV theory \cite{BCOV}) from the BV-BFV perspective. 
At the semiclassical level, the relation is a result of Gerasimov--Shatashvili \cite{GS} (see also our review in \cite[Section~7.6]{HJ}). In this paper, we explore the perturbative BV-BFV quantization and show that, for an appropriate choice of gauge fixing, no quantum corrections are added to the semiclassical result. We thus prove the conjecture put forward by Gerasimov and Shatashvili in their original paper.

There are two other key points that motivated this paper and its companion paper \cite{HJ}. Firstly, we were interested in studying in detail the bulk-boundary or ``holographic'' correspondences mentioned above.  In this paper we prove that in these special cases the boundary theory is simply an effective theory of the bulk theory, in the sense that the bulk fields have been partially integrated out. Given the results in \cite{HJ}, we expect that the effective action viewpoint can explain more general bulk-boundary correspondences. To put it in clear words: holographic duality means that the boundary theory is the semiclassical limit of a certain effective action of the bulk theory. In the theories we consider in this paper, this semiclassical limit is exact,\footnote{In our framework, this is due to the fact that the constraints are affine functions of the variables describing the field configurations on either boundaries (such theories are called ``biaffine'' in \cite{HJ}).} but in general there is of course no reason to expect this. 

Secondly, partition functions on cylinders can be interpreted as kernels of generalized Segal--Bargmann transforms (see Appendix \ref{app:SB}). They are of interest because, in a $d$-dimensional theory, they describe how a state on a $(d-1)$-dimensional manifold $\Sigma$ depends on the choice of a polarization. One way to interpret our results is that in our examples those generalized Segal--Bargmann transforms (in general, it is only their semiclassical limit) can be described by another quantum field theory that lives on $\Sigma$. 

Our results show that in both those cases -- seemingly unrelated at first glance -- the corresponding boundary theory is given by a (generalized) Hamilton--Jacobi action.
The BV-BFV formalism turns out to be a clear conceptual framework in which one can state and prove those results from first principles: the only inputs required are those of a local field theory, namely a space of fields $F_M$ and an action functional $S_M \colon F_M \to \RR$. From a different perspective, this paper can also be viewed as an invitation to learn the formalism.

Before passing to a detailed description of our results, as a primer on the BV-BFV formalism we give a  brief recollection of abelian Chern--Simons theory in the BV-BFV formalism, which can be safely skipped by readers familiar with the subject.
\subsection{Chern--Simons theory in the BV-BFV formalism}\label{ss: intro CS BVBFV}
We consider abelian Chern--Simons theory 
 in dimensions $d = 4l + 3$ with $l$ a positive integer.  For a $d$-dimensional spacetime manifold $N$ (possibly with boundary), the space of fields is defined as
 $F_N = \Omega^{2l+1}(N)$ and the action functional is 
$$S_N[A] = \frac12 \int_N A \wedge d A.$$ 

In dimension $d=3$, we also consider nonabelian Chern--Simons theory. Here there is a structure Lie algebra $\ggg$ of coefficients endowed with a nondegenerate invariant pairing  $\langle \cdot, \cdot \rangle$. The space of fields on $N$ is then the space of $\ggg$-valued 1-forms $F_N = \Omega^{1}(N,\ggg)$ (thought of as the space of connections on a trivial principal $G$-bundle $N\times G$ with $G$ the connected and simply connected Lie group integrating $\ggg$). The action functional is 
$$S_N[A] = \int_N \frac12 \langle A, dA \rangle + \frac16 \langle A, [A,A]\rangle.$$ 

 Since these theories are gauge theories, to define the perturbative partition function we need a gauge fixing formalism.  In this paper, we will use the BV-BFV formalism, the modification of the Batalin--Vilkovisky (BV) formalism for manifolds with boundary introduced by two of the authors together with N. Reshetikhin in \cite{CMR14,CMR15}. Let us briefly explain this formalism by means of our main example. 
 
 The BV-BFV extension of abelian Chern--Simons theory has $\ZZ$-graded space of fields $\FF_N = \Omega^\bullet(N)[2l+1]$. This notation is shorthand for saying that a homogeneous form $\omega$ is assigned ghost number $\gh(\omega) = 2l+1 - \deg(\omega)$, so that all forms have \emph{total degree} $\gh + \deg = 2l+1$. In particular $\FF^0_N = F_N$. The space $\FF_N$ is an odd  symplectic vector space with odd symplectic form 
 \[
 \omega_N(\A,\A') = \int_N \A \wedge \A',
 \]
 where $\A,\A'$ are nonhomogeneous differential forms and only the top degree part contributes to the integral.\footnote{The symplectic form is odd because it pairs components of $\A$,$\A'$ of opposite parity, which in turn is due to the fact that $\dim N$ is odd.} The BV extended action functional of abelian Chern--Simons theory is 
 $$\calS_N[\A] = \frac12\int_N \A \wedge d\A.$$ 
In particular, restricting to forms of ghost number 0, we recover the classical action $S_N[A]$. 

If $\dd N = \emptyset$, then $(\calS_N,\calS_N) = 0$, where $(\cdot,\cdot)$ denotes the Poisson bracket induced by $\omega_N$. This equation is called classical master equation in the BV formalism, and it implies $Q_N^2 = 0$, where 
\[
Q_N = \int_N d\A \wedge \frac{\delta}{\delta \A}
\]
is the odd hamiltonian vector field of $\calS_N$. 

If $\dd N \neq \emptyset$, then we assign additional BFV\footnote{BFV is short for Batalin--Fradkin--Vilkovisky \cite{BF,FV}. } data to the boundary. The space of boundary fields is $\FF^\dd_{\dd N}= \Omega^\bullet(\dd N)[2l+1]$ with even symplectic form 
\[
\omega^\dd_{\dd N}(\A,\A') = \int_{\dd N} \A \wedge \A'.
\]
This symplectic form is the de~Rham differential (on $\FF^\dd_{\dd N}$) of the 1-form 
\[
\alpha^\dd_{\dd N} = \frac12 \int_{\dd N} \A \wedge \delta \A.
\]
Finally, using the surjective submersion $\pi\colon \FF_N \to \FF^\dd_{\dd N}$, given by pullback of differential forms from $N$ to $\dd N$, we can project the vector field\footnote{The vector field $Q_N$ is no longer the hamiltonian vector field of $\calS_N$. It is instead defined via the formula above.} $Q_N$ to $\FF^\dd_{\dd N}$. One can check that it is also hamiltonian. For degree reasons it then automatically has a unique odd hamiltonian function that we denote by $\calS^\dd_{\dd N}$. The important structural relation between the boundary BFV data $(\FF^\dd_{\dd N},\alpha^\dd_{\dd N}, \calS^\dd_{\dd N})$ and the bulk ``broken'' BV data $(\FF_N,\omega_N,\calS_N,Q_N,\pi)$ is 
 \begin{equation}
 \delta\calS_N =\iota_{Q_N}\omega_N +  \pi^*\alpha^\dd_{\dd N} \label{eq:mCME I}.
 \end{equation} 
 The 
 data, together with the structural relation
 \eqref{eq:mCME I}, 
 are the content of the classical BV-BFV formalism.  For more details we refer to \cite{CMR14}. 
 
 For $f$ a function on $\FF^\dd_{\dd N}$, there is a symmetry of the data
given by shifting $\calS_N \to \calS_N^f =\calS_N + \pi^*f$ and $\alpha^\dd_{\dd N} \to \alpha^{\dd, f}_{\dd N} = \alpha^\dd_{\dd N} + \delta f $. Clearly this is a symmetry of Equation \eqref{eq:mCME I}. 
 \begin{remark}
 The BV-BFV formulation of abelian Chern--Simons theory can be extended---as a $\ZZ_2$-graded theory---to dimension $d = 1$, see Section \ref{ss: 1d ab CS}. Instead of $\RR$-valued forms, there one has to consider forms with values in an odd vector space $\Pi\g$, with $\g$ an ordinary vector space equipped with an inner product.
This is the abelian version of the model
 studied in \cite{1dCS}.
 \end{remark} 
 
 Let us explain now how to define the BV-BFV partition function. We will be very brief here; for a detailed exposition we refer to \cite{CMR15}. We will require some additional pieces of data. The first one is a polarization $\PP$ (involutive lagrangian distribution) on $\FF^\dd_{\dd N}$. We say that the boundary 1-form $\alpha^\dd$ is compatible with $\PP$ if it vanishes on vectors belonging to $\PP$. Typically this is not the case, but it may  be achieved by means of the symmetry $\alpha^\dd \to \alpha^\dd + \delta f$ discussed above. Denote by $\BB$ the leaf space of the polarization. In the examples of this paper we actually have $\FF^\dd_{\dd N} \cong T^*\BB$. 
\begin{remark}
In the examples in this paper the graded manifold $\mathcal{F}^\partial_{\partial N} $ is actually a vector space, and the simplest polarizations are splittings into complex lagrangian subspaces $\mathcal{F}^\partial_{\partial N} \otimes \CC = \mathcal{B} \oplus \mathcal{B}'$, we call those \emph{linear polarizations}. However, it is interesting to consider more general polarizations. An example is the \emph{Hitchin polarization} on $\Omega^3(M,\CC)$ explained in subsection \ref{sss: Hitichin}. 
\end{remark} 
 Next, we require a splitting $\FF_N \cong \BB \times \YY$ where $\YY$ is also an odd symplectic vector space. Finally, we choose the data of a gauge fixing on $\YY$: another splitting $\YY \cong \VV \times \YY'$ into odd symplectic vector spaces and a lagrangian $\LL\subset \YY'$ such that  $0$ is an isolated critical point of  $\calS_N$ when restricted to $\BB \times \VV \times \LL \subset \BB \times \VV \times \YY' \cong \FF_M$, fiberwise over $\BB\times\VV$. The odd symplectic space $\VV$ is called the space of residual fields and $\LL$ is called the gauge-fixing lagrangian. 
 
 Given all these data, we can define the perturbative partition function as the 
  integral of the exponentiated BV action over $\LL$: 
 \begin{equation*}
 Z_N(\As,\as) = \int_{\alpha \in \LL\subset \YY'}\mathcal{D}\alpha \exp\left(\frac{i}{\hbar}\calS^f_N(\As,\as,\alpha)\right) = \exp\left(\frac{i}{\hbar}S_\eff(\As,\as) \right).
 \end{equation*}
 The partition function $Z$ and the effective action $S_\eff$ are both functions on $\BB \times \VV$. 
 
 The integral is defined as a sum over Feynman diagrams---i.e., modeled on finite-dimensional Gaussian integrals. As a consequence of the structural equation \eqref{eq:mCME I}, one expects $Z_N$ to satisfy the modified quantum master equation (mQME)
 \begin{equation}
 (\Omega_\BB - \hbar^2\Delta_\VV)Z_N = 0, \label{eq:mQME intro}
 \end{equation}
 where $\Delta_\VV$ is the BV operator acting on functions on the odd symplectic vector space $\VV$ of residual fields, given in Darboux coordinates $(q^i,p_i)$ by $\sum_i \pm \frac{\dd}{\dd q^i}\frac{\dd }{\dd p_i}$, and $\Omega_\BB$ is a quantization of the BFV action $\calS^\dd_{\dd N}$ acting on functions on $\BB$ as a differential operator. If we write $\FF^\dd_{\dd N} = T^*\BB \ni (b,b')$,
 then $\Omega_\BB$ is given by $\calS^\dd_{\dd N}(b,-i\hbar \frac{\dd}{\dd b})$, with all derivatives to the right. At lowest order in $\hbar$, we have $\Omega_\BB^2 = 0$ as a consequence of $(S,S)=0$. To ensure this to all orders, one might have to add higher order corrections (although there is no guarantee in general that the corrections exist). In all problems considered in this paper, $\Omega_\BB$ squares to zero without further corrections (see Theorem \ref{thm:SummaryI}).
 Since these operators anticommute with each other and square to zero, there is a double complex where $Z_N$ defines a cohomology class $[Z_N]$. 
This cohomology class is invariant under deformation of the choices made in the construction. For more details on the mQME \eqref{eq:mQME intro}, we refer to \cite{CMRcq},\cite{CMR15}. 
 \begin{remark}[Choice of residual fields]\label{rem: intro res fields}
 The choice of the space $\VV \subset \YY$ is not unique. In fact, there is a partially ordered set of such choices, with maximal element $\YY$ and a minimal element $\VV_\mr{min}$, and one can pass from a bigger to a smaller choice by a BV pushforward. A more detailed discussion can be found in \cite[Appendix F]{CMR15}. In this paper, when we deal with dimensions $d \neq 1$, we usually first have a ``big'' (infinite-dimensional) choice of $\VV$. In some cases we are able to compute the BV pushforward to $\VV_\mr{min}$.
 \end{remark}
 \subsection{Main results of the paper}
 We are now ready 
  to describe the main results of this paper. We consider only spacetime manifolds $N$ that are cylinders: $N = I \times \Sigma$. We think of the interval as $I = [0,1]$, so that $\dd N = \{0\} \times \Sigma \, \sqcup \, \{1\} \times \Sigma$, and we denote by $\Sigma_\ii,\Sigma_\oo$ the two components. The BFV space of boundary fields $\FF^\dd_{\dd N}$ then splits as $\FF^\dd_{\dd N} = \FF^\dd_\ii \times \FF^\dd_\oo$. 
 
 We will consider polarizations of the 
 space of boundary fields $\FF^\dd_{\dd N}
 $  that are products of two polarizations on the two factors. We will work mostly with linear polarizations, i.e., splittings  $\FF^\dd_{\Sigma}\otimes \CC  = \BB \oplus \BB'$ where $\BB,\BB'$ are complementary complex lagrangian subspaces of $\FF^\dd_{\Sigma}\otimes \CC$, so that we have an injection $\omega_\Sigma^\sharp\colon \BB' \to \BB^*$. We will then write  (suppressing the complexification) $\FF^\dd_\Sigma \cong T^*\BB$ and say that we are using the $\BB$-representation.\footnote{
A comment on complex vs.\ real spaces: by default, spaces of fields  and spaces of boundary fields are real vector spaces. In this paper, we 
 have to complexify them to impose convenient boundary conditions/polarizations. The splittings involved in the  gauge fixing (and thus in the corresponding spaces $\YY$,  $\VV$, $\YY'$, $\LL$) are only defined over $\CC$. The equations of motion and the evolution relation are defined over $\RR$ but need to be complexified to be described by a generating function (involving a complex polarization). When writing down path integrals an implied step is a choice of a real contour in the complexified space of fields, see Appendix \ref{App: real contour} for an illustration of the principle.
 } In ghost number 0 we also allow nonlinear polarizations with smooth leaf space $\BB$ such that $F^\dd_\Sigma \cong T^*B$.\footnote{In general, one might have to restrict to  neighborhoods in $\FF^\dd_{\dd N}$ and $\BB$ to achieve this isomorphism.} 
 
 Consider now a representation $\FF^\dd_{\dd N}\cong T^*\BB_\ii \times T^*\BB_\oo$. Denote by $\mathcal{EL}$ the zero locus of $Q$. It consists of (nonhomogeneous) closed forms in the abelian case and of ``flat'' nonhomogoneous forms in the nonabelian one. We call the projection $\LL:=\pi(\mathcal{EL})\subset \FF^\dd_{\dd N}$ the \emph{BV evolution relation}. 
 Denoting by $F^\dd_{\dd N}$ the ghost number 0 part, we get a product of two ordinary cotangent bundles $F^\dd_{\dd N} \cong T^*B_\ii \times T^*B_\oo$. We denote the restriction of the graded evolution relation by $L := \LL|_{\gh = 0}$ and call it simply the \emph{evolution relation}. One can show that it is a lagrangian subspace and that it consists of boundary fields that can be extended to solutions of the Euler--Lagrange equations. A generalized generating function for $L$ is given by the Hamilton--Jacobi action $S_{\mr{HJ}}[b_\ii,b_\oo,e] \in C^\infty(B_\ii \times B_\oo \times V_{\mr{aux}})$, where $V_\mr{aux}$ is a space of additional parameters. The requirement on $V_\mr{aux}$ is that in the fiber of $F_N$ over any
 triple $(b_\ii,b_\oo,e) \in B_\ii \times B_\oo \times V_{\mr{aux}}$ there exists a unique solution to the equations of motion. There is a 
 poset of choices for this space. This is discussed in detail in the companion paper \cite{HJ} and recalled in the Section \ref{sec:HJ} below. A first set of results can then be summarized as follows.
\begin{theorem}\label{thm:SummaryI}
Consider one of the following BV-BFV theories: 
\begin{enumerate}
\item 1D abelian Chern--Simons theory with linear or nonlinear polarizations, 
\item $d = 4l+3$-dimensional Chern--Simons theory with linear or nonlinear polarizations, 
\item 3-dimensional nonabelian Chern--Simons theory with linear polarizations.
\end{enumerate} Then there exists a space of residual fields $\VV$ and a gauge-fixing lagrangian $\LL$ such that the ghost number 0 component of $\VV$ coincides with $V_\mr{aux}$ and the ghost number 0 component of $S_\eff$ coincides with $S_{\mr{HJ}}$. 
\end{theorem}
Here the choice of space of residual fields $\mathcal{V}$ is determined in ghost number 0 by the requirement that it should be isomorphic to $V_\mr{aux}$. In all cases in this paper, the gauge-fixing Lagrangian is the space of 0-forms along the interval $I$ intersected with $\mathcal{Y}$.\footnote{Since we are inverting only the de Rham differential along the interval $d_I$, this is the unique possible choice of gauge-fixing lagrangian. Invertibility of $d_I$ restricted to $\mathcal{L}$ requires that $\mathcal{L}$ contains no 1-forms. The lagrangian condition then requires that all 
0-forms along $I$ in $\mathcal{Y}$ belong to $\mathcal{L}$.}  
In particular, there are no quantum corrections to the ghost number 0 part of the effective action (notice that the HJ action can be computed, as shown in \cite[Section~7]{HJ}, completely at the classical level). The leading term in the effective action was expected to be the Hamilton-Jacobi action from the finite-dimensional results in the companion paper \cite{HJ} (see in partiuclar Theorem 11.4 there). This theorem is nothing but an expression of the fact that to leading order the quantum theory is determined by the Euler-Lagrange locus. The fact that there are no quantum corrections is probably not surprising for the abelian Chern-Simons theory. For nonabelian Chern-Simons theory, the absence of quantum corrections is due to the fact that, given a complex structure on $\Sigma$, the interaction term is affine in both the holomorphic and antiholomorphic components of the connection along $\Sigma$, and the fact that our polarizataion and gauge fixing are compatible with this complex structure. 

Our second main result concerns the mQME.
\begin{theorem}\label{thm:SummaryII}
In all cases of Theorem \ref{thm:SummaryI} with linear polarizations, the BV-BFV partition functon $Z$ satisfies the modified quantum master equation $$
(\Omega - \hbar^2 \Delta)Z = 0$$ with $\Omega = \Omega_{\BB_\ii} + \Omega_{\BB_\oo}$ given by the standard quantization (i.e. with derivatives to the right of multiplication operators) of the boundary action at both endpoints. For nonlinear polarizations $F^\dd_{\Sigma} = T^*B \ni (b,b')$, the mQME is satisfied whenever the constraint $d_{\Sigma}A = 0$ is linear in the momenta $b'$. 
\end{theorem}
Again, in this case there are no quantum corrections to $\Omega$. These theorems summarize the results obtained in the various sections of this paper, where we discuss the different examples individually. We will outline the paper in slightly more detail in Section \ref{sec:Outline} below. Before that, let us comment on some of the more specific results in more detail. 

\subsubsection{Three-dimensional abelian Chern--Simons theory}\label{sss: intro 3d ab CS}
In three-dimensional Chern--Simons theory, in ghost number 0 we have the lagrangian splitting $$F^\dd_\Sigma \otimes \CC = \Omega^1(\Sigma,\CC) = \Omega^{1,0}(\Sigma) \oplus \Omega^{0,1}(\Sigma).$$ 
For instance, one can define $$\BB_\oo = \Omega^0(\Sigma,\CC) \oplus \Omega^{1,0}(\Sigma) \ni (\As^0_\oo,\As^{1,0}_\oo)$$ and $$\BB_\ii = \Omega^{0,1}(\Sigma) \oplus \Omega^2(\Sigma,\CC)\ni (\As^{0,1}_\ii,\As^2_\ii).$$ As discussed in Section \ref{ss: ahol-hol}, a possible choice for the space of residual fields is 
$$\VV_\mr{small} = \{dt\cdot(\As^0_\Ires+\As^2_\Ires)+(1-t)\cdot \As^0_\res + t\cdot \As^2_\res\} \subset \Omega^\bullet(I \times \Sigma %
,\CC),
$$
where $\As^k_\Ires, \As^k_\res$ are complex valued 
$k$-forms on $\Sigma$, $t$ is the coordinate on $I = [0,1]$ and the ghost numbers are $\gh (\As^k_\Ires) =-k,  \gh(\As^k_\res) = 1-k$. We will denote by $\sigma := A^0_\Ires$ the only ghost number 0 field in $\VV_\mr{small}$. The BV-BFV partition function is then computed as 
\begin{multline*}
Z_\mr{small}
=\exp\frac{i}{\hbar}
\Big(\underbrace{\int_\Sigma\Big(
\As^{1,0}_\oo \As^{0,1}_\ii
+(\dd \As^{0,1}_\ii+\bar\dd \As^{1,0}_\oo )\; \sigma
 +\frac12 \sigma \dd\bar\dd \sigma
\Big)}_{S_\text{HJ}}+\\+
\int_\Sigma\Big(-\As^0_\oo\As^2_\ii +\As^0_\oo \As^2_\res -\As^2_\ii \As^0_\res+ \frac12 \As^2_\res \As^0_\res  \Big)
\Big).
\end{multline*}
In particular, focusing on the summand of the effective action in the first line, we recognize the Hamilton--Jacobi action from Example \ref{ex:HJ ab CS}, as an instance of Theorem \ref{thm:SummaryI}. It is the action functional of a 2D free boson conformal field theory, coupled to the boundary fields $\As^{1,0}_\oo$ and $\As^{0,1}_\ii$. We arrive at the same result in ghost number 0 in Section \ref{sec: 3d par ghost ab CS}, using $\BB_\ii = \Omega^0(\Sigma,\CC) \oplus \Omega^{0,1}(\Sigma)$. One can integrate out the remaining residual fields to obtain then the fact that the partition function of three-dimensional Chern--Simons theory for the minimal space of residual fields (cf.\ Remark \ref{rem: intro res fields}) coincides with the partition function of the 2D free boson CFT. In particular, one can observe the Weyl anomaly in the 3D Chern--Simons partition function. See Remark \ref{rem: ahol-hol fully integrated}. 

\subsubsection{Three-dimensional nonabelian Chern--Simons theory and CS-WZW correspondence} \label{sss: intro 3d non ab CS}
The same lagrangian splitting as above can be used to study the 3D nonabelian Chern--Simons theory---see Section \ref{ss: non-ab CS parall ghost}. The representation we use in that section is $\FF^\dd_{\dd N} = T^*\BB_\ii \times T^*\BB_\oo$ with 
$$\BB_\oo = \Omega^0(\Sigma,\ggg_\CC) \oplus \Omega^{1,0}(\Sigma,\ggg) \ni (\As^0_\oo,\As^{1,0}_\oo)$$ and $$\BB_\ii =  \Omega^0(\Sigma,\ggg_\CC) \oplus \Omega^{0,1}(\Sigma,\ggg).$$ 
As a space of residual fields one can use 
$$
dt\cdot\Omega^0(\Sigma,\ggg_\CC) \oplus \Omega^2(\Sigma,\ggg_\CC)[-1] 
\ni (dt\cdot\sigma, \As^*_\res)$$
with $\gh(\sigma) = 0, \gh(\As^*_\res) = -1$. 
We compute the effective action in Lemma \ref{lemma: S_eff non-ab via sigma} and see that it has a tree part $S^{\eff(0)}$ and a 1-loop part $\mathbb{W}$: 
$$S^\eff = S^{\eff(0)} - i\hbar\mathbb{W} = S^{\eff(0)}_\ph + S^{\eff(0)}_\gh - i\hbar\mathbb{W}$$
(the subscript $\ph$ denotes the terms involving only fields of ghost number 0, the subscript $\gh$ denotes terms involving fields with nonzero ghost number). 
At first glance the explicit formula \eqref{S eff non-abelian} seems obscure, but we observe a number of interesting phenomena: 
\begin{enumerate}[i)]  
\item One has to restrict the residual field $\sigma$ to a certain ``Gribov region'' $B_0 \subset \ggg_\CC$---a region where the exponential map $\exp\colon \ggg_\CC \to G_\CC$ is injective---to make sure that certain power series appearing in $S^{\eff(0)}_\gh$ converge (Remark \ref{rem: B_0}).
\item As shown in Lemma \ref{lemma: S_eff non-ab rewriting via g}, when we restrict $\sigma$ to $B_0$, we can reparametrize by $g = e^{-\sigma}\colon \Sigma \to G_\CC$. In this reparametrization, we can rewrite $S^{\eff(0)}_\ph$ as 
\[ 
S^{\eff(0)}_\ph =  \int_\Sigma \Big( \langle \As^{1,0}_\oo, g\, \As^{0,1}_\ii g^{-1} \rangle -
\langle \As^{1,0}_\oo,\bar\dd g\cdot g^{-1} \rangle 
- 
\langle  \As^{0,1}_\ii, g^{-1}\dd g  \rangle \Big)
+\mr{WZW}(g)
\]
with the Wess--Zumino--Witten term
\[
\mr{WZW}(g)=- \frac12 \int_\Sigma  \langle \dd g\cdot g^{-1}, \bar\dd g\cdot  g^{-1} \rangle -\frac{1}{12}  \int_{\Sigma \times I}\langle d h\cdot h^{-1},[d h\cdot h^{-1},d h\cdot h^{-1}] \rangle
\]
and $h = e^{(t-1)\sigma}.$ This coincides with the Hamilton--Jacobi action of Chern--Simons theory, see Example \ref{ex: HJ nonab CS}. 
\item The term $-i\hbar\mathbb{W}$ in principle violates Theorem \ref{thm:SummaryI} and is divergent. However, it has a nice interpretation as 
a change of path integral measure from $\DD\sigma$ to $\DD g$, see Section \ref{sss: ghost wheels}. In particular, if one interprets $Z$ as a \emph{half-density} rather than a function on the space of residual fields, and thus $S^\eff$ as a log-half-density, the effective action has no quantum corrections in the $(g,g^*)$ coordinates on $\VV$ (here $g^*$ is the Darboux coordinate for $g$). It is in this sense that Theorem \ref{thm:SummaryI} holds.  
\item In Section \ref{sss: mQME nonab CS} we show that $Z$ satisfies the modified quantum master equation in the different interpretations of $Z$ (partition function vs.\ half-density).  Interestingly, in the $(g,g^*)$ representation the mQME implies the well-known Polyakov--Wiegmann identity for the WZW action. 
\end{enumerate}

We thus observe a strong version of the CS-WZW correspondence: Namely, the effective theory of nonabelian Chern--Simons theory on $I \times \Sigma$ is a ``gauged WZW theory,'' i.e., a WZW theory on $\Sigma$ coupled to chiral gauge fields $\As^{1,0}_\ii,\As^{0,1}_\oo$.  

We also compute expectation values of vertical Wilson lines (Section \ref{sec:WilsonLines}) and show that they are given by field insertions in this WZW theory. This extends the CS-WZW correspondence to the level of observables. See the discussion in Section \ref{sec:CSWZW}. 

Formally, after integrating over the residual group-valued field $g$, the Chern--Simons partition function agrees with the partition function of the gauged WZW theory. One can use this to heuristically show the holomorphic factorization of the WZW model, as argued in Section \ref{sec: hol fact WZW}. 

Different versions of the relation between nonabelian Chern--Simons theory and the WZW model were studied in the literature before. A connection somewhat close to the one we are discussing appeared in \cite[Section 4]{BT}; one important difference is that we are focusing on the homological (BV-BFV) aspects 
obtaining WZW as an effective BV theory. The other point is that the logic of our computation is different (it is a pure perturbative computation; it does not rely on quantum gauge invariance but has it as a result), see Remark \ref{rem: comaprison with Blau-Thompson}.

 \subsubsection{Seven-dimensional Chern--Simons theory and the CS-BCOV correspondence}\label{sss: intro 7d CS vs KS}
 Finally, let us consider seven-dimensional Chern--Simons theory on a cylinder $N = I \times M$ with $M$ a Calabi--Yau manifold. In particular, the complex structure on $M$ defines a lagrangian splitting of $F^\dd_{M} = \Omega^3(M,\CC)$: 
$$\Omega^3(M, \CC) = X^+ \oplus X^- , \qquad X^+ = \Omega^{3,0}(M) \oplus \Omega^{2,1}(M), \qquad X^- = \Omega^{1,2}(M) \oplus \Omega^{0,3}(M).$$ 
This lagrangian splitting determines a polarization of $F^\dd_M$. 

On a Calabi--Yau manifold, however, there is another polarization of $F^\dd_M$ due to Hitchin \cite{Hitchin}. Namely, a complex three-form $A$ on $M$ which is not itself decomposable, i.e., a wedge product of three 1-forms on $M$, has a decomposition $A = A^{+,\nl} + A^{-,\nl}$ where $A^{\pm,\nl}$ are decomposable three-forms uniquely defined up to exchange of $+$ and $-$. This polarization is discussed in \ref{sss: Hitichin}. 

We can compute the partition function $Z$ on the cylinder with 
$$\BB_\ii = \Omega^{\leq 2}(M,\CC) \oplus X^+ \ni (\mathsf{c}_\ii,\As_\ii^{+,\mr{l}}) $$ 
and 
$$\BB_\oo = \Omega^{\leq 2}(M,\CC) \times X^{-,\nl} \ni (\mathsf{c}_\oo,\As_\oo^{-,\mr{nl}}). $$
In this case, Theorem \ref{thm:SummaryI} holds---as shown in Section \ref{sec: 1d par ghosts nonlinear}---and Theorem \ref{thm:SummaryII} holds because the constraint $d_MA = 0$ is linear in the momentum $A^{+,\nl}$. Thus, the physical part of the effective action coincides with the Hamilton--Jacobi action computed in \cite[Section 7.6]{HJ} and is given by \begin{multline*} S_\ph[\As^{+,\mr{l}}_\ii,\As^{-,\nl}_\oo,\As^{2,0}_\Ires,\As^{1,1}_\Ires,\As^{0,2}_\Ires] \\ = \frac12 \int_M \dd \As^{1,1}_\Ires \bar\dd \As^{1,1}_\Ires  + \int_M \As^{+,\mr{l}}_\ii d\As^{0,2}_\Ires + \As^{2,1}_\ii\bar\dd \As^{1,1}_\Ires  
- G(\As^{+,\l}_\ii +  d\As^{2,0}_\Ires + \dd\As^{1,1}_\Ires;\As^{-,\nl}_\oo)
\end{multline*}
with no quantum corrections in our choice of gauge fixing. Here $\As^{p,q}_\Ires$ denote 2-forms of Hodge type $(p,q)$ which are the residual fields of ghost number 0, and $G(A^{+,\mr{l}},A^{-,\mr{nl}})$ is the generating function satisfying $\delta G = A^{-,\mr l}\delta A^{+,l} - A^{+,\mr{nl}}\delta A^{-,\nl}$. Since the partition function $Z$ satisfies, by Theorem \ref{thm:SummaryII}, the modified quantum master equation, when changing the gauge fixing the partition function changes by an $(\Omega - \hbar^2\Delta)$-exact term.

 The partition function $Z$ can be interpreted as the integral kernel of a generalized Segal--Bargmann transform, see Appendix \ref{app:SB}. We thus show that the approximation used by Gerasimov and Shatashvili in \cite{GS}---where they were only assuming this representation to be  true in the semiclassical limit---is exact.
  Following \cite{GS}, we can then relate the Chern--Simons partition function to the partition function of Kodaira--Spencer or BCOV  theory, defined in \cite{BCOV} and recalled in Appendix \ref{app:KS}, as follows. One can consider a certain ($\Omega$-closed) state $\psi(\As_\oo^{-,\nl},\mathsf{c}_\oo)$ in the $A^{-,\nl}$-representation. We then apply the operator $Z$ to $\psi$ --- by multiplying and formally integrating over $\BB_\oo$---  and show that the result $Z\cdot\psi$ is still $(\Omega - \hbar^2\Delta)$-closed.  Next we identify a gauge-fixing lagrangian $\LL  \subset \VV$ and compute $Z''[\As^{3,0}_\ii,\As^{2,1}_\ii,c_\ii] = \int_\LL Z\cdot \psi.$ One can then show that in ghost number 0 
$$Z''_\ph[\As^{3,0}_\ii=\omega_0,\As^{2,1}_\ii = x] \sim  Z_{KS}[x] ,$$
where $\omega_0$ is a normalized generator of $H_{\dd}^{3,0}(M)$, $x$ is a $\dd$-harmonic form, and $Z_{KS}[x]$ is the Kodaira-Spencer partition function with background $x$. For the precise statement see Section \ref{sss: CS vs KS}.
In particular, we see that this statement holds not only in the semiclassical approximation to $Z_{CS}$ as in \cite{GS}, but that it is exact.  For general boundary conditions $\As^{3,0}_\ii,\As^{2,1}_\ii$, the Chern--Simons partition function can be computed from the mQME.

\subsection{Structure of the paper}\label{sec:Outline}
We summarize the remaining results by outlining the structure of the paper. \\
In Section \ref{sec:HJ}, we recall the construction of the Hamilton--Jacobi action from \cite{HJ}, and the important examples (abelian and nonabelian Chern--Simons theory) from that paper.

In Section \ref{ss: 1d ab CS}, we consider as a warm-up the example of the abelian 1D CS theory. This is the 1D AKSZ theory with target a vector space $\g$ that we assume to have an inner product and a compatible complex structure $J$, so that $\g \otimes \CC = \g^+ \oplus \g^-$ splits into $\pm i$-eigenspaces of $J$. We then compute the partition function for both $\BB_\ii = \BB_\oo = \g^+$ in Section \ref{ss: 1d hol-hol}    and $\BB_\ii =\g^-, \BB_\oo = \g^+$ in Section \ref{ss: 1d ahol-hol} and comment briefly on the  Theorems \ref{thm:SummaryI} and \ref{thm:SummaryII} in this context (which are in this case rather trivial). 

In Section \ref{s: 3dCS}, we consider the 3D abelian Chern--Simons theory on $I\times \Sigma$ as a 1D theory with values in $\g = \Omega^\bullet(\Sigma)$. Choosing a complex structure on $\Sigma$, we split $\g = \g^+ \oplus \g^-$ and consider   $\BB_\ii = \BB_\oo = \g^+$ in Section \ref{ss: hol-hol} and $\BB_\ii =\g^-, \BB_\oo = \g^+$ in Section \ref{ss: ahol-hol}. In both cases, we comment on the HJ and mQME properties, and in the second case also the pushforward to the minimal space of residual fields and the relation to the 2D free boson CFT is discussed.

In Section \ref{s: par ghosts}, we consider the case where 
$\BB_\ii$ and $\BB_\oo$ both have components only in nonnegative ghost number, and agree in positive ghost number. We call these ``parallel ghost polarization''. In Section \ref{sec: 1d par ghosts complex}, we consider 1D Chern--Simons theory with values in a complex, with opposite linear polarizations in ghost number 0. In Section \ref{sec: 1d par ghosts nonlinear}, we consider the same theory with a possibly nonlinear polarization on the $\oo$-boundary. These subsections serve as a toy model for the higher-dimensional Chern--Simons theories considered later. In Section \ref{ss: non-ab CS parall ghost}, we return to the three-dimensional Chern--Simons theory, with opposite linear polarization in degree 0. After briefly studying again the abelian case in Section \ref{sec: 3d par ghost ab CS}, we discuss the nonabelian case in more detail, the results are summarized already in Section \ref{sss: intro 3d non ab CS} above. Finally in Section \ref{ss: non-ab hol-hol}, we consider the nonabelian theory with parallel polarizations both in the ghost and physical sectors.  

In Section \ref{sec:higherdimCS}, we turn to Chern--Simons theories of arbitrary dimension. We consider both linear polarizations that are transversal in the ghost sector at opposite ends (Section \ref{ss: higher dim CS transversal}) and parallel in the ghost sector (Section \ref{ss: higher dim CS parallel}).  Finally in Section \ref{s:7dCSKS} we turn our attention to nonlinear polarizations at one boundary, in particular the 7D case with Hitchin polarization, that was summarized in Section \ref{sss: intro 7d CS vs KS} above. 

The appendices contain some complementary material. In Appendix \ref{app:SB} we 
show how to recover the usual Segal--Bargmann transform as a BV-BFV partition function on an interval with a particular choice of boundary polarizations. This is an illustration of the maxim that topological partition functions on cylinders 
yield instances of generalized Segal--Bargmann transforms. We also comment on the
 contour integration in the complexified space of fields. In Appendix \ref{app:KS}, we recall very briefly the Kodaira--Spencer theory of deformations of complex structures and the BCOV action functional. 

\subsection{Outlook}

Finally, let us point out some interesting directions for further research. 
\begin{itemize}
\item All our partition functions depend nontrivially on the choice of complex structure on the boundary.\footnote{As a matter of fact, they even depend on the \emph{Riemanian metric} inducing the complex structure on the boundary, a phenomenon known as conformal anomaly. See Remark \ref{rem: ahol-hol fully integrated}.(b).} This dependence should be described by extending the partition function to a (projectively flat) section of a vector bundle over the moduli space of complex structures on the boundary, for instance the one constructed in \cite{ADW}.

\item Recently \cite{MW} it has been suggested that the partition function of a 3D $U(1)$ Chern--Simons theory can be computed by averaging over Narain moduli space of boundary CFT's. We believe our methods could be generalized to include nontrivial flat bundles and we plan to investigate this proposal. 

\item Our results on the CS-WZW correspondence strongly suggest that the space of $n$-point conformal blocks can be described as the $\Omega$-cohomology (see Section \ref{sec:CSWZW}; the genus-zero case of this statement was a result of \cite{ABM}). This would provide an interesting new description of the space of conformal blocks.  We also hope it would lead to a better understanding of the relationship between Chern--Simons theory and the KZ(B) connection. 

\item It would be highly interesting to compare our findings on the CS-BCOV correspondence to other approaches to the subject such as \cite{CL}.

\item Another proposal to compute holographic duals of action functionals from BV-BFV formalism on manifolds with boundary was made by the second and third authors together with M. Schiavina in \cite{Hol}. The point of view there was more focused on descent equations and extensions to higher codimension, while the present paper emphasizes the role of the BV effective action. The relationship between the two constructions needs to be explored. 

\item In \cite{GW}, the authors show that there exists a 1-loop exact quantization of Chern-Simons theory on $\mathbb{R}^3$, which is similar to the result that we obtain here (in our case, the wheels appear only in the ghost sector of the theory). The gauge fixing they use is different from ours, and the focus there is not on partition functions, but rather on the anomaly-freeness of the theory, a problem which does not appear in our gauge fixing. Nevertheless it would be interesting to investigate this gauge-fixing from the BV-BFV viewpoint and compare it with our current results.

\end{itemize}

\newcommand{\ddd}{
d}
\newcommand{\EE}{
e}
\newcommand{\II}{
i}
\newcommand{\EL}{\text{\sl EL}}
\subsection{Notations and Conventions}
\begin{quote}
\emph{This is a quantum paper and notations fluctuate. Fixing one makes a complementary one explode.
}
\end{quote}

In this paper we study field theories on cylinders $N = I \times \Sigma$ from different viewpoints, with $I=[0,1]$ the interval with its standard orientation, and a $\Sigma$ a $(d-1)$-dimensional closed oriented manifold. Notations are adapted to the individual sections. \\
We are considering 
Chern--Simons-type theories, in different dimensions and with different targets. The Chern--Simons superfield is denoted $\mathcal{A} \in \Omega^\bullet(I \times \Sigma, \Pi\g)$.

 When we are considering 1-dimensional theories (with a possibly infinite-dimensional target) as in Sections \ref{ss: 1d ab CS}, \ref{sec: 1d par ghosts complex}, \ref{sec: 1d par ghosts nonlinear}, we denote the components of the superfield $\mathcal{A} = \psi + A$, where $\psi\in \Omega^0(I,\g)$ and $A \in \Omega^1(I,\g)$. Decoration of $\psi, A$ with superscripts denotes components with respect to\ a splitting of $\g$. Decoration of $\psi,A$ with subscripts denotes components with respect to\ a splitting of $\Omega^\bullet(I)$. Typical subscripts are $\ii$ and $\oo$, denoting fields supported on the $\ii$ or $\oo$ boundary (elements of $\BB_\ii$ or $\BB_\oo$) respectively, $\res$ for residual fields (elements of $\VV$), and $\fl$ for fluctuations (elements of $\LL$).  \\
 When we are thinking about higher-dimensional theories (still on cylinders) as in Section \ref{s: 3dCS} , \ref{sec:higherdimCS}, we denote the components of $\mathcal{A} = \As + \ddd t \cdot \AsI$, with $\As,\AsI \in \Omega^{0,\bullet}(I \times \Sigma)$. Superscripts now denote components of homogeneous form degree in $\Sigma$. \\
In Sections \ref{ss: non-ab CS parall ghost}, \ref{ss: non-ab hol-hol}, it is convenient to revert to a more ``traditional'' notation $\mathcal{A} =  c + A + A^* + c^*$, here the nonhomogeneous differential form is split according to form degree. There we also denote the (finite-dimensional) coefficient Lie algebra by $\ggg$. 
\begin{Ack}
We thank Samson Shatashvili for suggesting the study of 7D abelian Chern--Simons theory in the quantum BV-BFV formalism, now in Section~\ref{s:7dCSKS}, which was the original motivation out of which this paper 
and \cite{HJ}
grew. 
We also thank Ivan Contreras, Philippe Mathieu, Nicolai Reshetikhin, Pavel Safronov, Michele Schiavina,
Stephan Stolz, Alan Weinstein, Ping Xu and Donald Youmans for useful discussions. We also thank an unknown referee for very helpful and inspiring comments and suggestions.
\end{Ack}
\section{Constrained systems and generalized Hamilton--Jacobi actions}\label{sec:HJ}
We start with a short review of the results of \cite{HJ} that are relevant for this paper. We focus on action functionals of the form\footnote{Such an action functional is the classical part of an AKSZ theory \cite{AKSZ} on an interval $I$. See also \cite{BDMGV} for the study of such a theory in the Dirac formalism.}
\[
S[p,q,e]=\int_I(pdq-\langle H(p,q),e\rangle),
\]
where $I$ is the interval $[0,1]$, $(p,q)$ are coordinates on a given cotangent bundle $T^*B$ (and, 
by abuse of notation, also stand for a map from $I$ to $T^*B$),
$e$ is a one-form on $I$ taking value in some vector space $\frh$, and $H$ is a given map $T^*B\to\frh^*$. The pairing between the $p$ and the $q$ coordinates is understood, whereas for the pairing between $\frh$ and its dual we use the notation $\langle\ ,\rangle$. In the applications of this paper the space $\frh$ and the manifold $B$ are infinite-dimensional (typically, Fr\'echet spaces).

To be more precise, $T^*B$ denotes some given vector bundle over $B$ with a nondegenerate pairing to $TB$; we denote by $\theta$ the canonical one-form on it (which we will also call the Noether 1-form in the following) and by $\omega=d\theta$ the canonical symplectic form; by $\frh^*$ we denote a given subspace of the dual of $\frh$ such that its pairing to $\frh$ is still nondegenerate.
The first term in the action can also be written in coordinate-free way as $x^*\theta$ in terms of a path $x\colon I\to T^*B$.
For the second term, we 
assume a given map $X$ from $\frh$ to the vector fields on $T^*B$ and define, up to carefully chosen constants, the map $H$ by $\iota_X\omega=d H$. (Note that $H$ is a map from $\frh$ to the functions on $T^*B$, and we assume that, dually, it belongs to the chosen subspace $\frh^*$.)

\begin{example}[3D Chern--Simons theory]\label{exa:3DCS}
Consider
 3D
 Chern--Simons theory for a quadratic Lie algebra $\ggg$ on $I\times\Sigma$, where $\Sigma$ is a closed oriented surface with a chosen complex structure.
  The complexified phase space 
 is $T^*B=\Omega^{1,0}(\Sigma)\otimes\ggg\oplus\Omega^{0,1}(\Sigma)\otimes\ggg$ with
 $B=\Omega^{0,1}(\Sigma)\otimes\ggg$. We then have $\frh=\Omega^0(\Sigma)\otimes\ggg$ and $\frh^*=\Omega^2(\Sigma)\otimes\ggg$. The pairings are induced by the given pairing on $\ggg$ and by integration on $\Sigma$. An element of $T^*B$ is a connection one-form, the map $X$ yields the gauge transformations, and $H$ is the curvature two-form.
\end{example}


We split the fields into two classes: the dynamical field (the map $x$ to $T^*B$) and the Lagrange multiplier (the $\frh$-valued one-form $e$). We accordingly split the Euler--Lagrange (EL) equations into the evolution equation, the variations with respect to the dynamical field,
\[
dx =  \langle X,e \rangle,
\]
and the constraints, the variations with respect to the Lagrange multiplier,
\[
H=0.
\]
Note that the constraints must be satisfied at every time.

We define the evolution relation $L$ as the possible boundary values (at $0$ and $1$ in $I$) that a solution to the EL equations can have. 
Assuming it to be a (possibly immersed) submanifold, $L$ turns out to be an isotropic submanifold of $\overline{T^*B}\times T^*B$ \cite{CMR14}, where the bar means that we use the opposite symplectic form. We assume it to be actually split lagrangian (i.e., for every point $v$ of $L$, its tangent space $T_vL$, which is isotropic in general, must have an isotropic complement).\footnote{In particular, this implies that $L$ is lagrangian, i.e., that the symplectic orthogonal of $T_vL$ is $T_vL$ itself  for every point $v$ of $L$.}
Thanks to the Hodge decomposition theorem, this assumption is satisfied in all the examples of this paper. 

We then denote by $C$ the projection of $L$ on either factor $T^*B$ and we assume it to be a submanifold. As observed in \cite{CMRcq}, if $L$ is lagrangian, then $C$ is coisotropic. In particular, at every point $c\in C$ and for every $\xi\in\frh$,
the  vector  $\langle X(c),\xi \rangle$ is tangent to $C$. Moreover,
the span of these vectors at each point defines an involutive distribution on $C$, called the characteristic distribution (the reduced phase space of the theory is then defined as the reduction of $C$ with respect to its characteristic distributon).\footnote{In the case of 3D Chern--Simons theory, Example~\ref{exa:3DCS}, $C$ is the space of flat connections on the surface $\Sigma$, and the reduced phase space is the space of flat connections modulo gauge trnasformations.}

In the case at hand, we have that $C$ is the zero locus of $H$. The evolution equation, for a given $e$, is then the hamiltonian evolution for the (time-dependent) hamiltonian $\langle H,e\rangle$. Since $C$ is coisotropic, this evolution does not leave $C$---so it is enough to implement the constraint $H=0$ at the initial, or final, endpoint---and lies along the characteristic distribution. It follows that the evolution relation $L$ consists of pairs of points on $C$ lying on
the same leaf of the characteristic distribution.

Next we are interested in solutions to the EL equations. For this we have to fix boundary conditions; namely, we have to choose lagrangian submanifolds $L_0$ and $L_1$ of $T^*B$ at the endpoints of $I$, and we assume that the intersection of $L_0\times L_1$ with the evolution relation $L$ is discrete.
\footnote{More precisely, one looks for solutions of the EL equations that are critical points for the action. This requires changing the boundary one-form by an exact term in such a way that it vanishes on $L_0$ and $L_1$. In particular, this can only happen if $L_0$ and $L_1$ are isotropic. Moreover, we want the intersection of $L_0\times L_1$ with $L$ to be discrete, so that locally the solution is unique. At each intersection point, the tangent spaces to $L_0\times L_1$ and to $L$ are then complementary, which implies that they are not only isotropic but split lagrangian. 
We want this to happen for generic boundary conditions. 
This is the reason why $L$ is required to be split lagrangian.} 
For simplicity, we will work with a unique solution.
We are also interested in letting boundary conditions vary, so we consider families of lagrangian submanifolds (polarizations). Concretely, at the initial endpoint we take the $L_0$s to be the fibers of $T^*B$, which we then parametrize by $B$, whereas at the final endpoint we realize $T^*B$ as $T^*B'$, with $B'$ a possibly different manifold, and take the $L_1$s to be the fibers of $T^*B'$, which we then parametrize by $B'$.\footnote{In the examples of this paper, $B$ is a vector space, so $T^*B$ is of the form $B^*\oplus B$. We are also interested in complex polarizations. In the case at hand, this simply means allowing $B$ to be a complex vector space. Then the complexified phase space is $B^*\oplus B$.
} 
We want the variations of the action with the given boundary conditions not
to have boundary terms. This is automatically satisfied at the initial point, where we take the polarization $T^*B$, but 
we have to adapt the action to the canonical one-form $\theta'$ of $T^*B'$ at the final endpoint. For this, we assume that there is a function $f$ on $B\times B'$ such that $\theta = \theta' + df$ and we modify the action to
\[
S^f[p,q,e]:=S[x,e]-f(q(1),Q(p(1),q(1))),
\]
where $Q$ is the base coordinate of $T^*B'$. 

We define the Hamilton--Jacobi (HJ) action $S^f_\text{HJ}$ of the theory with respect to the given polarizations as the evaluation of $S^f$ on a solution (which we assume to be unique) to the evolution equation for each choice of $e$. Note that $S^f_{\text{HJ}}$ is a function on $B\times B'\times\Omega^1(I,\frh)\ni(q_\ii,Q_\oo,e)$. Also note that we do not impose the constraints in the definition of $S^f_{\text{HJ}}$. It was proved in \cite{HJ}  
$i)$ that $S^f_{\text{HJ}}$ is invariant under certain equivalence transformations of $e$, and 
$ii)$ that it is a generalized generating function for the evolution relation $L$ with respect to the given polarizations.

Let us elaborate on this. As for 
$i)$, assume for simplicity that, as in every example of this paper, $\frh$ is actually a Lie algebra and $H$ is an equivariant momentum map (for the infinitesimal action $X$ of $\frh$ on $T^*B$). Then $e$ may be regarded as a connection one-form on $I$. The equivalence transformations are in this case gauge transformations that are trivial at the endpoints. As for 
$ii)$, the statement means that, upon setting to zero the variation of $S^f_{\text{HJ}}$ with respect to (the equivalence class of) $e$, 
we recover the final $P$ variables of a solution as the variation of $S^f_{\text{HJ}}$ with respect to $Q_\oo$ and
the initial $p$ variables of a solution as minus the variation of $S^f_{\text{HJ}}$ with respect to $q_\ii$.

Explicit examples, relevant for this paper, are discussed in \cite[Section~7]{HJ}. We recall the results.

\begin{example}[Abelian 3D Chern--Simons theory]\label{ex:HJ ab CS}
We use the notations of Example~\ref{exa:3DCS}, but now with $\g=\RR$.
We take the initial polarization as $T^*B$, with $B=\Omega^{0,1}(\Sigma)$, and the final polarization 
as $T^*B'$, with $B'=\Omega^{1,0}(\Sigma)$.\footnote{
This polarization is known in the literature on Chern--Simons theory. In a context close to the context of the present paper, it was discussed in \cite[Section 2.4.4]{CMR13} and in \cite{ABM}. 
} 
We denote by $\dd$ and $\bar\dd$ the Dolbeault operators.
The HJ action then reads
\[
 S^f_\text{HJ}=\int_\Sigma \left(\As^{1,0}_\oo \As^{0,1}_\ii + \sigma (\bar\dd  \As^{1,0}_\oo+\dd \As^{0,1}_\ii)+\frac12 \sigma \dd \bar\dd \sigma\right),
\]
with $\As^{0,1}_\ii\in B$, $\As^{1,0}_\oo\in B'$, and $\sigma\in\Omega^0(\Sigma)$.
\end{example}

\begin{example}[Nonabelian 3D Chern--Simons theory]\label{ex: HJ nonab CS}
Again we use the notations of Example~\ref{exa:3DCS}. The initial and final polarizations now are 
$T^*B$, with $B=\Omega^{0,1}(\Sigma)\otimes\ggg$, and $T^*B'$, with $B'=\Omega^{1,0}(\Sigma)\otimes\ggg$.
We assume the exponential map from $\ggg$ to the its simply connected Lie group $G$ to be surjective. In this case the gauge-invariant parameter $g\in\operatorname{Map}(\Sigma,G)$ 
is of the form 
$g=\EE^{-\sigma}$ with 
$\sigma\in\operatorname{Map}(\Sigma,\ggg)$. The HJ action then reads
\[
 S^f_\text{HJ} = \int_\Sigma \Big( \langle \As^{1,0}_\oo, g\, \As^{0,1}_\ii g^{-1} \rangle -
\langle \As^{1,0}_\oo,\bar\dd g\cdot g^{-1} \rangle 
- 
\langle  \As^{0,1}_\ii, g^{-1}\dd g  \rangle \Big)
+\mr{WZW}(g)
\]
with the Wess--Zumino--Witten term
\[
\mr{WZW}(g)=- \frac12 \int_\Sigma  \langle \dd g\cdot g^{-1}, \bar\dd g\cdot  g^{-1} \rangle -\frac{1}{12}  \int_{\Sigma \times I}\langle \ddd h\cdot h^{-1},[\ddd h\cdot h^{-1},\ddd h\cdot h^{-1}] \rangle,
\]
where $h=\EE^{(t-1)\sigma}$.\footnote{Here we are using the conventions of Section \ref{ss: non-ab CS parall ghost} (Lemma \ref{lemma: S_eff non-ab rewriting via g}) which are different from the conventions of \cite{HJ}.} 
Thus, the HJ action of Chern--Simons theory can be identified with a ``gauged WZW action'' (see for instance \cite{Gawedzki CFT 2}). This points at a deep relationship between these two theories.

\end{example}

\section{BV-BFV approach warm-up: 1D abelian Chern--Simons}\label{ss: 1d ab CS}
As a warm-up exercise before the BV-BFV treatment of 3D Chern--Simons, let us consider one-dimensional abelian Chern--Simons theory\footnote{
This is the abelian version of the theory considered in \cite{1dCS}.
} on an interval $I=[0,1]$ --- the AKSZ theory  with $\ZZ_2$-graded space of BV fields 
$$\FF=\mr{Map}(T[1]I,\Pi\g)=\Omega^\bt(I)\otimes \Pi\g.$$
Here $\g$ is a vector space of coefficients endowed with a nondegenerate inner product $(,)$ and $\Pi$ is the parity-reversal symbol. A vector in $\FF$ is the superfield $\psi+A$, with $\psi$ a $\Pi\g$-valued $0$-form and $A$ a $\g$-valued $1$-form, and the BV action is:
\begin{equation}\label{1d ab CS action}
S(\psi+A)=\int_I\frac12 (\psi, \dt \psi) 
\end{equation}
with $\dt=dt \frac{d}{dt}$ the de Rham differential on the interval $t\in[0,1]$. The odd symplectic form on $\FF$ is given by $\omega=-\int_I (\delta A, \delta \psi)$. The cohomological vector field   (BRST operator) $Q$ on $\FF$ is defined by $Q:\psi\mapsto 0,\; A\mapsto \dt\psi$.

The BFV phase space assigned to a point is $\PhiPM_{pt}=\Pi\g$, equipped with Noether $1$-form $\alpha_{pt^\pm}=\pm\frac12(\psi,\delta \psi)$ where $\pm$ corresponds to the orientation of the point; the BFV action is zero,\footnote{It is nonzero in nonabelian theory: there one has $S_{pt^\pm}=\pm\frac16 (\psi,[\psi,\psi])$.} $S_{pt}=0$. 
We are using the following sign convention for the BV-BFV structure relation: 
\begin{equation}\label{delta S = i_Q omega + alpha}
\delta S=\iota_Q \omega+\pi^* \alpha_\dd .
\end{equation}

Assume that $\g$ is equipped with a complex structure $J\in \mr{End}(\g)$, $J^2=-\mr{Id}$, compatible with the inner product.
We have a splitting of the complexification of $\g$ into $\pm i$-eigenspaces of $J$:
\begin{equation}\label{g=g+ + g-}
\g_\CC=\g^{+}\oplus \g^{-}
\end{equation}
--- the ``holomorphic'' and ``antiholomorphic'' subspaces of $\g_\CC=\g\otimes\CC$, which are  lagrangian due to compatibility between $J$ and $(,)$.

\subsection{Holomorphic-to-holomorphic boundary conditions}\label{ss: 1d hol-hol}
Consider the boundary polarization $\mr{Span}(\frac{\dd}{\dd\psi^{-}})$ imposed at both $t=0$ and $t=1$ (a.k.a. $\psi^{+}-\psi^{+}$ representation, as the partition function will depend on the boundary value $\psi^{+}_\mr{in}$ at $t=0$ and boundary value $\psi^{+}_\mr{out}$ at $t=1$).  For compatibility with this polarization, we need to modify the action (\ref{1d ab CS action}) by boundary terms:
\begin{equation} \label{1d ab CS Spol}
S\mapsto S^\pol=S+\frac12 (\psi^+,\psi^-)\big|_{t=1}-\frac12 (\psi^+,\psi^-)\big|_{t=0} .
\end{equation}
Then the corresponding boundary $1$-form is:
\begin{multline*} \alpha_{\dd I}^\pol=
\left.\left(\frac12(\psi,\delta\psi)+\delta\frac12(\psi^+,\psi^-)\right)\right|_{t=1}-
\left.\left(\frac12(\psi,\delta\psi)+\delta\frac12(\psi^+,\psi^-)\right)\right|_{t=0}\\
=
(\psi^-,\delta\psi^+)\big|_{t=1}- (\psi^-,\delta\psi^+)\big|_{t=0} 
\end{multline*}
--- the canonical $1$-form in the chosen representation,
as desired (cf. Section \ref{ss: intro CS BVBFV}; see \cite[Section 9]{HJ} and references therein for more details).
The space of fields $\FF$ is fibered over the base $\BB=\Pi\g^+\oplus\Pi\g^+ =\{(\psi^+_\mr{in},\psi^+_\mr{out})\}$ with the fiber
$$\YY=\Omega^\bt(I,\dd I;\Pi\g^+)\oplus \Omega^\bt(I;\Pi\g^-) .$$
Here the first summand on the r.h.s. is $\g^+$-valued forms vanishing on the boundary and the second summand is $\g^-$-valued forms with free boundary conditions. The cochain complex $\YY$ admits the following splitting (a Hodge decomposition):
\begin{multline}\label{1d ab CS hol-hol Hodge}
\YY=
\underbrace{\big(dt\cdot \g^+\oplus 1\cdot \Pi\g^-\big)}_{\VV}\bigoplus\\
\bigoplus
\underbrace{\big(\Omega^0(I,\dd I;\Pi\g^+)\oplus \Omega^0_{\int=0}(I;\Pi\g^-)\big)}_{\YY'_{K-ex}}\bigoplus \underbrace{\big(\Omega^1_{\int=0}(I;\g^+)\oplus \Omega^1(I;\g^-)\big)}_{\YY'_{d-ex}} .
\end{multline}
Here the first term (``residual fields'') is a deformation retract of $\YY$ (in this case, in fact, its cohomology). The subscript $\int=0$ means ``forms with vanishing total integral'' (against $dt$ in $0$-form case). The two last terms jointly form an acyclic subcomplex $\YY'$ of $\YY$, split into a $d$-exact part and its direct complement --- the $K$-exact part, where $K\colon \YY^\bt\ra \YY^{\bt-1}$ is the chain homotopy between identity and projection onto $\VV$. Explicitly, $K$ kills all $0$-forms and acts on $\g^+$- and $\g^-$-valued $1$-forms as follows:
\begin{equation}\label{1d ab CS hol-hol K}
K\colon  \begin{array}{lcl}
dt\,g^+(t) &\mapsto& \int_0^t dt' g^+(t')-t\int_0^1 dt' g^+(t') ,  \\
dt\, g^-(t) &\mapsto& -\int_t^1 dt' g^-(t')+ \int_0^1 dt'\, t'\, g^-(t') .
\end{array}
\end{equation}
The integral kernel of $K$ is the propagator:
\begin{equation}\label{1d ab CS hol-hol eta}
\eta(t,t')=\pi^+\otimes(\theta(t-t')-t)+\pi^-\otimes(t'-\theta(t'-t)) ,
\end{equation}
where $\pi^\pm$ are the projectors from $\g$ to $\g^\pm$ and $\theta$ is the step function.

The BV-BFV partition function is given by the following path integral (see \cite{CMR15} for the general construction):
\begin{multline}\label{1d ab CS Z PI}
Z(\psi^+_\mr{in},\psi^+_\mr{out}; \psi^-_\mr{res}, A^+_\mr{res})=\\
=
\int_{\YY'_{K-ex}\subset \YY'} \mc{D} \psi^+_\mr{fl} \mc{D} \psi^-_\mr{fl} \;
e^{\frac{i}{\hbar}S^\pol \left(
\til{\psi^+_\mr{in}}+\til{\psi^+_\mr{out}}+\psi^+_\mr{fl}+\psi^-_\mr{res}+\psi^-_\mr{fl}+dt\cdot A^+_\mr{res}
\right)} .
\end{multline}
Here the notations are:
\begin{itemize}
\item $\til{\psi^+_\mr{in}}$ is the discontinuous extension\footnote{
See \cite{CMR15} and \cite[Section 9]{HJ} for the details on discontinuous extension of boundary fields.
} of $\psi^+_\mr{in}$ at $t=0$ by zero at $t>0$; likewise, $\til{\psi^+_\mr{out}}$ is the discontinuous extension of $\psi^+_\mr{out}$ at $t=1$ by zero  at $t< 1$; 
\item the ``fluctuation'' $(\psi^+_\mr{fl}, \psi^-_\mr{fl})\in \YY'_{K-ex}$ is the field we integrate over (while setting to zero the component in $\YY'_{d-ex}$ is the gauge fixing);
\item $(\psi^-_\mr{res}, dt\cdot A^+_\mr{res})\in \VV$, with $\psi^-_\res\in \Pi\g^-$ and $A^+_\res\in \g^+$, is the residual field.
\end{itemize}
Continuing the computation (\ref{1d ab CS Z PI}), we have the Gaussian integral
\begin{multline}\label{1d ab CS hol-hol Z computation}
Z=
\int \D \psi^+_\fl \D \psi^-_\fl\; \exp\frac{i}{\hbar}\Big(
\underbrace{\frac12 \int_I  \big( \psi^-_\res+\psi^-_\fl, \dt(\til{\psi^+_\oo}+ \til{\psi^+_\ii}) \big)}_a
+\\
+\underbrace{\frac12\int_I \big(\til{\psi^+_\oo}+ \til{\psi^+_\ii}, \dt (\psi^-_\res+\psi^-_\fl) \big)}_b+
\underbrace{\frac12 \int_I \big( \psi^-_\res+\psi^-_\fl , \dt\psi^+_\fl \big) }_c+
\underbrace{\frac12 \int_I \big(\psi^+_\fl, \dt(\psi^-_\res+\psi^-_\fl)\big)}_d+
\\
+\underbrace{\frac12 (\psi^+_\oo,\psi^-_\res+\psi^-_\fl\big|_{t=1})}_e \underbrace{-
\frac12 (\psi^+_\ii,\psi^-_\res+\psi^-_\fl\big|_{t=0})}_f
\Big)
\\ =\int   \mc{D} \psi^+_\mr{fl} \mc{D} \psi^-_\mr{fl} \; e^{\frac{i}{\hbar}\left( \int_I (
\psi^-_\fl, \dt \psi^+_\fl
) +(\psi^+_\oo,\psi^-_\res+\psi^-_\fl(1))- (\psi^+_\ii,\psi^-_\res+\psi^-_\fl(0))\right)}\\
=e^{\frac{i}{\hbar} (\psi^+_\oo-\psi^+_\ii,\psi^-_\res)} .
\end{multline}
Here the terms in the first expression above are:
\begin{itemize}
\item Term $a$ is a pure boundary term, in fact $a=e+f$, which leads to $\frac12$ factors of the boundary terms $e,f$ being doubled and replaced by $1$ in the second equality in (\ref{1d ab CS hol-hol Z computation}).
\item $b=0$.
\item $c=d=\frac12 \int_I (\psi^-_\fl,\dt\psi^+_\fl)$.
\end{itemize}

\subsection{Antiholomorphic-to-holomorphic boundary conditions}\label{ss: 1d ahol-hol}
Next, consider imposing the polarization $\frac{\dd}{\dd \psi^+}$ at $t=0$ and $\frac{\dd}{\dd\psi^-}$ at $t=1$ (a.k.a. $\psi^--\psi^+$ representation: we are fixing the boundary value  $\psi^-_\ii$ at $t=0$ and $\psi^+_\oo$ at $t=1$). The ``polarized action'' (the counterpart of (\ref{1d ab CS Spol})) in this case is:
\begin{equation}\label{1d ab CS ahol-hol S_pol}
S^\pol=S+\frac12(\psi^+,\psi^-)\big|_{t=1}+\frac12(\psi^+,\psi^-)\big|_{t=0}
\end{equation}
and the corresponding boundary $1$-form is:
\begin{multline*}
\alpha^\pol_{\dd I}=
\left.\left(\frac12(\psi,\delta\psi)+\delta\frac12(\psi^+,\psi^-)\right)\right|_{t=1}-
\left.\left(\frac12(\psi,\delta\psi)-\delta\frac12(\psi^+,\psi^-)\right)\right|_{t=0}\\
=
(\psi^-,\delta\psi^+)\big|_{t=1}- (\psi^+,\delta\psi^-)\big|_{t=0} .
\end{multline*}
This $1$-form vanishes along the chosen polarization, as desired.

Next, the fiber of the space of fields $\FF$ over the base $\BB=\Pi\g^-\oplus \Pi\g^+=\{(\psi^-_\ii,\psi^+_\oo)\}$ is the complex
\begin{equation}\label{1d ab CS ahol-hol Y}
\YY=\Omega^\bt(I,\{1\};\Pi\g^+)\oplus \Omega^\bt(I,\{0\};\Pi\g^-) ,
\end{equation}
which admits the following decomposition:
\begin{multline}\label{1d ab CS ahol-hol Hodge}
\YY=\underbrace{\big(dt\cdot \g_{\CC} \oplus (1-t)\cdot \Pi\g^+\oplus t\cdot \Pi\g^-\big)}_{\VV}
\bigoplus \\
\bigoplus \underbrace{\big(\Omega^0_{\int=0}(I,\{1\};\Pi\g^+)\oplus \Omega^0_{\int=0}(I,\{0\};\Pi\g^-)\big)}_{\YY'_{K-ex}} \bigoplus\\
\bigoplus \underbrace{
\begin{array}{l}
\Big(\Big\{g^+(t)dt \in \Omega^1(I;\g^+)\;\Big|\; \int_I dt\,g^+(t)\cdot t=0\Big\}\oplus\\
 \oplus \Big\{g^-(t)dt \in \Omega^1(I;\g^-)\;\Big|\; \int_I dt\,g^-(t)\cdot (1-t)=0\Big\}
\Big)
\end{array}
}_{\YY'_{d-ex}}   
\end{multline}
Again, this is a splitting of $\YY$ into a deformation retract\footnote{
Note that here we have chosen $\VV$ to be larger than cohomology (which in fact vanishes in this case).
} and an acyclic subcomplex, with the latter split in turn into the $d$-exact part and a direct complement --- the $K$-exact part, with the chain homotopy $K$ taking the form
$$ 
K:\begin{array}{lcl}
dt\, g^+(t) &\mapsto & 
-\int_t^1 dt'\, g^+(t')+ 2(1-t) \int_0^1 dt' \, t'\,g^+(t') , \\
dt\, g^-(t) &\mapsto &  
\int_0^t dt'\, g^-(t')-2t \int_0^1 dt' \, (1-t')\, g^-(t') .
\end{array}
$$
Its integral kernel --- the propagator --- is
\begin{equation}\label{1d ab CS ahol-hol eta}
\eta(t,t')=\pi^+\otimes \big(-\theta(t'-t)+2(1-t)\,t'\big)+\pi^-\otimes \big(\theta(t-t') -2t\, (1-t')\big) .
\end{equation}
We write an element of the space of residual fields $\VV$ as $(1-t)\cdot \psi^+_\res+t\cdot \psi^-_\res+dt \cdot A_\res$, with $\psi^+_\res\in \Pi\g^+, \psi^-_\res\in \Pi\g^-,A_\res\in \g_{\CC}$.

The BV-BFV partition function is:
\begin{multline}\label{1d ab CS Z ahol-hol}
Z(\psi^-_\ii,\psi^+_\oo;\psi^+_\res,\psi^-_\res,A_\res)=\\
=
\int_{\YY'_{K-ex}\subset \YY'} \mc{D}\psi^+_\fl\, \mc{D}\psi^-_\fl\; e^{\frac{i}{\hbar}S^\pol\left(
\til{\psi^-_\ii}+\til{\psi^+_\oo}+ (1-t)\cdot\psi^+_\res +t\cdot \psi^-_\res+\psi^+_\fl+\psi^-_\fl+dt\cdot A_\res
\right)}\\
=\int  \mc{D}\psi^+_\fl\, \mc{D}\psi^-_\fl\; \exp\frac{i}{\hbar}\Big(
\int (\psi^-_\fl,\dt \psi^+_\fl)
+\frac12 (\psi^-_\res,\psi^+_\res)
+(\psi^+_\oo,\psi^-_\res+\psi^-_\fl(1)) - (\psi^-_\ii,\psi^+_\res+\psi^+_\fl(0))
\Big)\\
=\exp{\frac{i}{\hbar}\left(
\frac12 (\psi^-_\res,\psi^+_\res)+(\psi^+_\oo,\psi^-_\res)-(\psi^-_\ii,\psi^+_\res)+(\psi^-_\ii,\psi^+_\oo)
\right)} .
\end{multline}
Here the last term comes from the simple Feynman diagram with a single propagator connecting 
$\psi^+_\oo$ and $\psi^-_\ii$.


\begin{remark}\label{rem: 1d ab cs ahol-hol full integral} One can further integrate out $\psi^\pm_\res$ in (\ref{1d ab CS Z ahol-hol}) resulting in the partition function
\begin{equation}\label{1d ab CS hol-ahol Z min}
Z(\psi^-_\ii,\psi^+_\oo)=e^{\frac{i}{\hbar}(\psi^+_\oo,\psi^-_\ii)} .
\end{equation}
It corresponds to choosing the space of residual fields in (\ref{1d ab CS ahol-hol Y}) to be zero (which is possible since the full complex $\YY$ is acyclic). Thus, (\ref{1d ab CS hol-ahol Z min}) is the minimal realization of the partition function of the theory on the interval with prescribed boundary polarizations, and it is the BV pushforward of the nonminimal realization (\ref{1d ab CS Z ahol-hol}).
\end{remark}

\begin{remark}
The exponent $S_\mr{HJ}=(\psi^+_\oo,\psi^-_\ii)$ in (\ref{1d ab CS hol-ahol Z min}) is the Hamilton--Jacobi action for the theory: it is the action (\ref{1d ab CS ahol-hol S_pol}) evaluated on the (unique) solution of EL equation $\dot{\psi}=0$ satisfying the boundary conditions $\psi^-|_{t=0}=\psi^-_\ii$, $\psi^+|_{t=1}=\psi^+_\oo$. Also, $S_\mr{HJ}$ is the generating function for the evolution relation of the theory:
\begin{multline*}
 L_{S_\mr{HJ}}=\Big\{\psi|_{t=1}=\psi^+_\oo+\frac{\dd S_\mr{HJ}}{\dd \psi^+_\oo}=\psi^+_\oo+\psi^-_\ii\;\; ,\;\; 
\psi|_{t=0}=\psi^-_\ii - \frac{\dd S_\mr{HJ}}{\dd \psi^-_\ii} = \psi^-_\ii+\psi^+_\oo \Big\} \\
=\big\{ \psi|_{t=1}=\psi|_{t=0}\big\} \quad \subset\;\; \Pi \g \times \Pi\g .
\end{multline*}
This provides a simple example of Hamilton--Jacobi formalism, see 
\cite{HJ} and Section \ref{sec:HJ},
with the phase space being the symplectic \textit{super}manifold $\Pi\g$.

Moreover, the exponent in (\ref{1d ab CS Z ahol-hol}) is a \textit{generalized} generating function for the evolution relation, with $\psi^\pm_\res$ the auxiliary parameters. It can also be seen as the Hamilton--Jacobi action for the action $S^\pol+\int dt (\lambda, \psi-\frac12 \psi_\res )$ with $S^\pol$ as in (\ref{1d ab CS ahol-hol S_pol}) and where $\lambda\in \Pi \g$ (a constant along $I$) is a Lagrange multiplier.

Likewise, the exponent $(\psi^+_\oo-\psi^+_\ii,\psi^-_\res)$ in the right hand side of (\ref{1d ab CS hol-hol Z computation}) is the generalized generating function for the same evolution relation, with respect to $(\psi^+_\oo,\psi^+_\ii)$-polarization, with $\psi^-_\res$ the auxiliary parameter, cf. \cite[Section 6.1]{HJ}.

To summarize, in these three cases Threorem \ref{thm:SummaryI} holds:
\begin{itemize}
\item Case of (\ref{1d ab CS hol-ahol Z min}): 1D abelian Chern-Simons with $(\psi^+_\oo,\psi^-_\ii)$-polarization at the endpoints of the interval, with $\mathcal{V}=V_\mr{aux}=0$. 
\item Case of (\ref{1d ab CS Z ahol-hol}): 1D abelian Chern-Simons with $(\psi^+_\oo,\psi^-_\ii)$-polarization, with $\mc{V}$ parametrized  by $(\psi^+_\res,\psi^-_\res,A_\res)\in \Pi\g^+\oplus \Pi\g^-\oplus \g_\CC$ and with $V_\mr{aux}=\Pi\g^+\oplus \Pi \g^-$ parametrized by $(\psi^+_\res,\psi^-_\res)$. 
\item Case of  (\ref{1d ab CS hol-hol Z computation}): 1D abelian Chern-Simons with $(\psi^+_\oo,\psi^+_\ii)$-polarization, with 
$\mc{V}=\Pi\g^-\oplus \g^+$ parametrized by $(\psi^-_\res,A^+_\res)$ and $V_\mr{aux}=\Pi\g^-$ parametrized by $\psi^-_\res$.
\end{itemize}

Note that in these cases $V_\mr{aux}$ is a direct summand in $\mc{V}$ but it is not singled out by the condition of vanishing ghost number (rather, it is the \emph{odd} part of $\mc{V}$): space of fields of 1D abelian Chern-Simons as considered here does not admit a $\ZZ$-grading.

We also remark that in all these cases Theorem \ref{thm:SummaryII} works trivially: $\Omega=0$ in 1D abelian Chern-Simons and $\Delta$ contains a derivative in $A_\res$ on which $S_\eff$ does not depend.

\end{remark}

\section{BV-BFV approach to 3D abelian Chern--Simons on a cylinder}\label{s: 3dCS}
Consider the $3$-dimensional abelian Chern--Simons theory on a cylinder $I\times \Sigma$, with $\Sigma$ a closed oriented surface and $I=[0,1]$ the interval parametrized by the coordinate $t$. The space of BV fields, as given by the AKSZ construction, is the $\ZZ$-graded mapping space
$$\FF= \mr{Map}(T[1](I\times \Sigma), \RR[1]) =\Omega^\bt(I\times \Sigma)[1].  $$
Exploiting the fact that the source is a cylinder, we can also write it as a free 
(i.e., with a quadratic action)
 $1$-dimensional AKSZ theory with the target given by forms on $\Sigma$:
$$ \FF=\mr{Map}(T[1]I,\mr{Map}(T[1]\Sigma,\RR[1]))=\Omega^\bt(I,\Omega^\bt(\Sigma)[1]). $$
The BV action is:
\begin{equation}\label{ab CS S}
S=\int_{ I\times\Sigma} \frac12 \mc{A}\wedge d\mc{A}= \int_I  \frac 12 (\mc{A}, \dt \mc{A}) + \frac12 (\mc{A} ,d_\Sigma \mc{A}).
\end{equation}
Here $d=\dt+d_\Sigma$ is the de Rham operator on the cylinder splitting into the surface part and the interval part; the pairing is integration over the surface: $(u,v)=\int_\Sigma u\wedge v$. The field splits into $0$- and $1$-form components along $I$ as 
$$\mc{A}=\As+dt\cdot \AsI$$
with $\As,\AsI$  two $t$-dependent nonhomogeneous forms on $\Sigma$; 
their homogeneous components are prescribed internal $\ZZ$-grading (ghost number) as follows:
$\gh (\As^{(p)})=1-p$, $\gh (\AsI^{(p)})=-p$.

Comparing to the discussion of Section \ref{ss: 1d ab CS}, this theory can be seen as $1$-dimensional Chern--Simons on $I$ with coefficients in $\g=\Omega^\bt(\Sigma)$. Here the fact that $\g$ is itself a cochain complex with differential $d_\Sigma$ gives rise to an additional term in the action. Also, the fact that $\g$ has a degree $-2$ (rather than degree $0$) graded-symmetric pairing allows one to prescribe $\ZZ$-grading to fields (in such a way that the action has degree $0$ and the odd symplectic form has degree $-1$) rather than just $\ZZ_2$-grading.

The BFV phase space 
assigned to a boundary surface ($\{1\}\times\Sigma$ or $\{0\}\times\Sigma$) is 
$\PhiPM_\Sigma=\Omega^\bt(\Sigma)[1]$ which is $0$-symplectic, with the 
Noether $1$-form $\pm\int_\Sigma \frac12 \As\wedge \delta \As $ where the sign is $+$ for the out-boundary and $-$ for the in-boundary. The phase space carries a degree $-1$ BFV action 
\begin{equation}\label{ab CS S_BFV}
S_\Sigma = \pm\int_\Sigma  \frac12 \As\wedge d_\Sigma \As .
\end{equation}

Next, assume that $\Sigma$ is endowed with a complex structure, so that complex-valued $1$-forms split as $\Omega^1_\CC(\Sigma)=\Omega^{1,0}(\Sigma)\oplus \Omega^{0,1}(\Sigma)$.  Then, mimicking (\ref{g=g+ + g-}), we split the (complexified) space of all forms on $\Sigma$ as follows:
\begin{equation}\label{ab CS pol}
\underbrace{\Omega^\bt_\CC(\Sigma)}_{\g_\CC} = \underbrace{(\Omega^0_\CC(\Sigma) \oplus \Omega^{1,0}(\Sigma))}_{\g^+} \bigoplus
\underbrace{(\Omega^{0,1}(\Sigma)\oplus \Omega^2_\CC(\Sigma))}_{\g^-} .
\end{equation}
This is, clearly, a splitting into lagrangian subspaces.

\subsection{Holomorphic-to-holomorphic  boundary conditions 
}\label{ss: hol-hol}
Consider the polarization $\mr{Span}\{\frac{\delta}{\delta \As^-}\}$ on both boundary surfaces, at $t=0$ and $t=1$, i.e., the one where we prescribe boundary values  $\As^+_\ii$, $\As^+_\oo$.
The corresponding modification of the action by boundary terms adjusting for the polarization is:
$$S^\pol=S+\frac12\int_{\{1\}\times\Sigma } \As^+  \As^- 
-\frac12\int_{\{0\}\times\Sigma } \As^+  \As^-$$
and the corresponding Noether $1$-form is:
$$\alpha^\pol_{\Sigma\times\dd I}=\int_{\{1\}\times\Sigma } \As^- \delta\As^+ - 
\int_{\{0\}\times\Sigma } \As^- \delta\As^+ . $$

The fiber of the (complexified) space of fields over the space of boundary conditions $\BB=\g^+[1]\oplus \g^+[1] =\{(\As^+_\ii,\As^+_\oo)\}$ is:
$$ \YY=\Omega^\bt(I,\dd I;\g^+[1]) \oplus \Omega^\bt(I; \g^-[1]) .
$$
Hodge decomposition (\ref{1d ab CS hol-hol Hodge}) holds (where one should replace $\Pi$ with degree shift $[1]$) and the formula for the chain homotopy (\ref{1d ab CS hol-hol K}) and the propagator (\ref{1d ab CS hol-hol eta}) also. Writing out the projectors $\pi^\pm$ explicitly in our case, we obtain the following formula for the propagator: 
\begin{multline}\label{ab CS hol-hol eta}
\eta((z,\bar{z},t)\, ;\, (z',\bar{z}',t'))=\\
=\delta^{(2)}(z-z')\;\frac{i}{2}\big(-dz\wedge d\bar{z}'+dz'\wedge d\bar{z}'\big) \; \big(\theta(t-t')-t\big)+\\
+\delta^{(2)}(z-z')\; \frac{i}{2}\big(d\bar{z}\wedge dz' + dz\wedge d\bar{z} \big)\; \big(t'-\theta(t'-t)\big).
\end{multline}
--- This is a distributional $2$-form on 
$(I\times \Sigma)\times (I\times \Sigma)$.
Here $z$ is the local complex coordinate on $\Sigma$. Our convention for the normalization of the delta function is: $\int \frac{i}{2} dz\wedge d\bar{z}\; \delta^{(2)}(z-z')=1$.

Note that the propagator (\ref{ab CS hol-hol eta}) is for the $\dt$ term in the action (\ref{ab CS S}) only, whereas the $d_\Sigma$ term is treated as a perturbation. 


The space of residual fields is:
$$
\VV=dt\cdot \g^+\oplus 1\cdot \g^- = \{dt\cdot \As^0_{I\,\res} +dt\cdot\As^{1,0}_{I\,\res}+\As^{0,1}_\res+\As^2_\res \} ,
$$
where $\As^0_{I\,\res},\; \As^{1,0}_{I\,\res},\;  \As^{0,1}_\res,\; \As^2_\res$ are $t$-independent forms on $\Sigma$ of de Rham degree $0, (1,0), (0,1), 2$, respectively, with 
internal degree 
$0,-1,0,-1$, respectively.

\begin{remark}[Axial gauge] \label{rem: axial gauge}
We call the gauge fixing introduced here the \emph{axial gauge}: it sets the ``axial'' field fluctuations --- those which are 1-forms along $I$ and forms of any degree along $\Sigma$ --- to zero. 

On the level of homological algebra, for $M,N$ closed manifolds, one can construct a chain contraction $K$ from $\Omega^\bt(M\times N)$ to $H^\bt(M)\otimes \Omega^\bt(N)$ of the form $K=K_M\otimes \mr{id}_N$ with $K_M$ a chain contraction from forms on $M$ to its cohomology (cohomology can be swapped for any deformation retract of the de Rham complex in the construction). The integral kernel of $K$ --- the propagator --- is a distributional form on $(M\times N)^{\times 2}$ containing the delta form on $N\times N$. A version of the axial gauge for Chern--Simons theory was first employed in \cite{FK}. 
In our situation, $N=\Sigma$ and $M=I$ is not a closed manifold and hence the construction has to be adapted 
for boundary conditions --- which is exactly what we did above. The chain contraction, corresponding to (\ref{ab CS hol-hol eta}), has the form $K=K_{I,\dd I}\otimes \mr{id}_{\g^+}+K_I \otimes \mr{id}_{\g^-}$. We will encounter versions of this construction for different choices of boundary conditions 
further in this paper (e.g., in the case of Section \ref{ss: ahol-hol}, the chain contraction has the form $K_{I,\{1\}}\otimes \mr{id}_{\g^+}+K_{I,\{0\}}\otimes \mr{id}_{\g^-}$).\footnote{
Depending on the choice of boundary conditions (e.g., in the case of Section \ref{ss: ahol-hol}), the space of forms subject to boundary conditions 
 $\YY[-1]$ may fail to be a subcomplex of $\Omega^\bt(I\times \Sigma)$ with respect to the total de Rham differential $d_I+d_\Sigma$. However, the operator 
 $K$ we are constructing can be seen as a chain contraction for just the ``axial'' differential $d_I$.}
\end{remark}

The BV-BFV partition function is readily calculated:
\begin{multline}\label{ab CS hol-hol Z}
Z(\underbrace{\As^0_\ii,\As^{1,0}_\ii}_{\As^+_\ii},\underbrace{ \As^0_\oo,\As^{1,0}_\oo}_{\As^+_\oo};\underbrace{\As^0_{I\,\res},\As^{1,0}_{I\,\res}}_{\As^+_{I\,\res}},\underbrace{\As^{0,1}_\res,\As^2_\res}_{\As^-_\res})=\\
=
\int_{\YY'_{K-ex}\subset \YY'}\mc{D} \As^+_\fl \;\mc{D} \As^-_\fl\;e^{\frac{i}{\hbar}S^\pol\left(
\til{\As^+_\ii}+\til{\As^+_\oo}+ \As^+_\fl+\As^-_\res+\As^-_\fl+dt\cdot\As^+_{I\,\res}
\right)}\\
=\int \mc{D}\As^+_\fl \;\mc{D} \As^-_\fl
\;e^{\frac{i}{\hbar}\big(
\int_{ I\times\Sigma} \As^-_\fl  \dt \As^+_\fl+ \int_{\{1\}\times\Sigma } \As^+_\oo \As^- - \int_{\{0\}\times\Sigma } \As^+_\ii \As^- + \int_{ I\times\Sigma} \frac12\A\, d_\Sigma \A \big)}
\\
=\int \mc{D} \As^+_\fl \;\mc{D} \As^-_\fl\;
\exp\frac{i}{\hbar} \Big(\int_{ I\times\Sigma} \As^-_\fl \, \dt \As^+_\fl+ \int_{\Sigma} 
\As^+_\oo \, (\As^-_\res+\As^-_\fl\big|_{t=1}) -\\ -
 \int_{\Sigma} \As^+_\ii\, (\As^-_\res+\As^-_\fl\big|_{t=0})
 +
 \int_\Sigma \As^0_{I\,\res} \dd \As^{0,1}_\res  + \int_{ I\times\Sigma} dt\;
 \As^+_{I\,\res} \,\bar\dd \As^+_\fl
\Big) .
\end{multline}
Here 
we are using the splitting $d_\Sigma=\dd+\bar\dd$ of de Rham operator on $\Sigma$ into the holomorphic and antiholomorphic Dolbeault operators. Finally, computing this Gaussian integral, we obtain
\begin{equation}\label{ab CS hol-hol Z answer}
Z=\exp\frac{i}{\hbar}\int_\Sigma\Big(
(\As^+_\oo - \As^+_\ii)\; \As^-_\res+\As^0_{I\,\res}\,\dd \As^{0,1}_\res+\frac12 (\As^+_\oo+\As^+_\ii)\;\bar\dd \As^+_{I\,\res}
\Big) .
\end{equation}
Here the last term arises from the 
 Wick contraction 
$$\wick{ \big\langle \int_\Sigma \As^+_\oo\,\c1{\As^-_\fl}\big|_{t=1}\;\;\int_{ I\times\Sigma} dt' \c1{\As^+_\fl} \bar\dd \As^+_{I\,\res}\, \big\rangle}=\underbrace{\int_I dt'\, (t'-\theta(t'-t))\big|_{t=1}}_{1/2} \int_\Sigma \As^+_\oo\,\bar\dd\As^+_{I\,\res} , $$ 
and a similar one with $\As^+_\ii$ talking  to $\bar\dd \As^+_{I\,\res}$. 

Graphically, the diagrams contributing to (\ref{ab CS hol-hol Z answer}) are:
\begin{figure}[H]
 \includegraphics[scale=1]{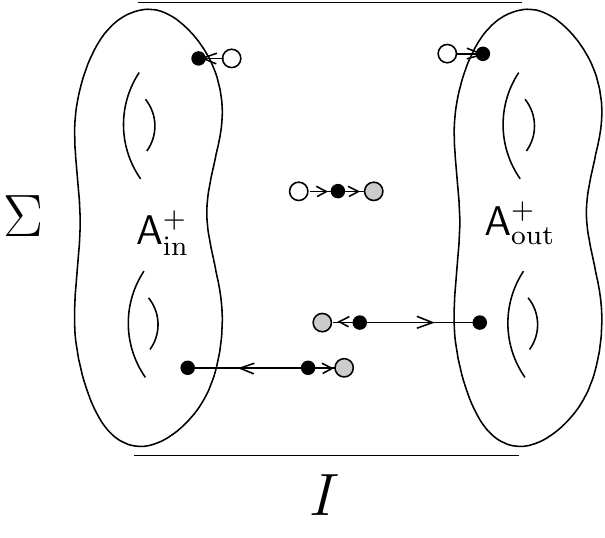} 
\caption{Feynman diagrams for the abelian theory on a cylinder in holomorphic-to-holomorphic polarization.}
\label{fig1}
\end{figure}
Here the conventions (Feynman rules) are:
\begin{itemize}
\item Black dots are vertices, which can be on in- or out-boundary (then they are univalent, with single incoming half-edge), or in the bulk (then they are bivalent -- with one incoming and one outgoing, or with two outgoing half-edges).
\item 
Half-edges can be internal (joined into pairs forming an edge, depicted as a long edge above) or external -- depicted as a short edge ending with a white or a gray blob, depending on orientation.
\item Boundary vertices are decorated by $\As^+_\oo$ on the out-boundary and by $\As^+_\ii$ on the in-boundary.
\item White blobs are decorated by $\As^-_\res$, gray blobs are decorated by $\As^+_{I\,\res}$.
\item (Long) edges are decorated by the propagator $\eta$.
\item Bulk vertices with one incoming and one outgoing half-edge  carry $\dd$;  bulk vertices with two outgoing half-edges carry $\bar\dd$. 
(Equivalently, one may say that bulk vertices are decorated by $d_\Sigma$ independently of orientation.)
\item For each connected graph $\Gamma$ in Figure \ref{fig1}, we take the product of decorations obtaining a differential form on $\mr{Conf}_\Gamma=\Sigma^{\#\mr{in-vertices}}\times (I\times \Sigma)^{\#\mr{bulk\,vertices}}\times \Sigma^{\#\mr{out-vertices}}$ depending on the residual fields. Then we take integral over $\mr{Conf}_\Gamma$, obtaining the value of the diagram.
\end{itemize}

We will return to the version of the result (\ref{ab CS hol-hol Z answer}) in the context of \textit{nonabelian} Chern--Simons theory  in Section \ref{ss: non-ab hol-hol}.

\subsubsection{Comparison with Hamilton--Jacobi action} \label{sss: hol-hol HJ}
We can write the result (\ref{ab CS hol-hol Z answer}) in the form
\begin{multline}\label{ab CS hol-hol Z via S_HJ}
Z=\exp\frac{i}{\hbar}\Big(
\underbrace{\int_\Sigma \Big( (\As^{1,0}_\oo-\As^{1,0}_\ii)\lambda +\lambda \dd\sigma + \frac12 (\As^{1,0}_\oo+\As^{1,0}_\ii)\bar\dd\sigma \Big) }_{S_\text{HJ}} +\\
+
\int_\Sigma\Big( (\As^0_\oo-\As^0_\ii)\As^2_\res +\frac12(\As^0_\oo+\As^0_\ii)\bar\dd\As^{1,0}_\Ires \Big)
\Big) ,
\end{multline}
where we introduced the alternative notation for degree zero residual fields
$$ \lambda:= \As^{0,1}_\res,\;\; \sigma:=\As^0_\Ires . $$
In the first integral in (\ref{ab CS hol-hol Z via S_HJ}) we recognize the Hamilton--Jacobi action 
\cite[Eq. (48)]{HJ}, 
which can be seen as the conformal $\beta\gamma$-system coupled to the boundary fields, while in the second integral we collected the contribution of nonzero-degree fields.

\subsubsection{Quantum master equation}\label{sss: hol-hol QME}
The space of states on a surface with $\As^+$-fixed polarization is the space of functions of $\As^+$ of the form 
\begin{equation}\label{states in A^+ pol}
\Psi(\As^+)=\sum_{n\geq 0} \int_{\Sigma^n}\gamma_n \pi_1^* \As^+\cdots \pi_n^* \As^+
\end{equation}
where $\{\gamma_n\}$ are $\hbar$-dependent 
distributional forms on $\Sigma^n$ and $\pi_i\colon \Sigma^n\ra\Sigma$ is the projection to the $i$-th copy of $\Sigma$ (we refer to \cite[Section 3.5.1]{CMR15} for details).
The space of states is
equipped with the differential (the quantum BFV operator)
\begin{equation}\label{Omega_Sigma^+}
\Omega_\Sigma^+=\int_\Sigma \left(-i\hbar\, \dd \As^0\,\frac{\delta}{\delta \As^{1,0}}+\epsilon \As^{1,0}\,\bar\dd \As^0 \right)
\end{equation}
with the sign $\epsilon=+1$ for the out-boundary and $\epsilon=-1$ for the in-boundary;\footnote{The integral over $\Sigma$ in (\ref{Omega_Sigma^+}) is understood as being with respect to the ``standard'' orientation, which coincides with the induced one from the cylinder on the out-boundary and is opposite to the induced one on the in-boundary.}
the superscript in $\Omega^+_\Sigma$ is a reminder of the choice of polarization. This operator is the canonical quantization of the boundary BFV action (\ref{ab CS S_BFV}), 
\begin{equation}\label{S BFV via hol splitting}
S_\Sigma=\epsilon\int_\Sigma \As^{1,0}\,\bar\dd \As^0+ \As^{0,1}\,\dd\As^0 .
\end{equation}
In the quantization, $\As^0, \As^{1,0}$ become multiplication operators and $\As^{0,1}\mapsto -\epsilon\, i\hbar \frac{\delta}{\delta \As^{1,0}}$, $\As^2\mapsto -\epsilon\, i\hbar \frac{\delta}{\delta \As^0}$ become derivations.

\begin{lemma}
The partition function (\ref{ab CS hol-hol Z answer}) satisfies the BV quantum master equation modified by the boundary terms (see \cite{CMR15}):
\begin{multline}\label{ab CS hol-hol mQME}
\Big(
\underbrace{
\int_\Sigma \big(-i\hbar\, \dd \As^0_\oo\,\frac{\delta}{\delta \As^{1,0}_\oo}+\As^{1,0}_\oo\,\bar\dd \As^0_\oo \big)
}_{\Omega_\oo^+}+
\underbrace{
\int_\Sigma \big(-i\hbar\, \dd \As^0_\ii\,\frac{\delta}{\delta \As^{1,0}_\ii}-\As^{1,0}_\ii\,\bar\dd \As^0_\ii \big)
}_{\Omega_\ii^+}
 - \\ -\hbar^2 \underbrace{\int_\Sigma \frac{\delta}{\delta \As^-_\res}\;\frac{\delta}{\delta \As^+_{I\,\res }} }_{\Delta_\res} \Big)\; Z = 0  .
\end{multline}
\end{lemma}

\begin{proof}
One checks this by a direct computation:
\begin{multline}\label{hol-hol mQME1}
(\Omega_\oo^+ +\Omega_\ii^+)Z =\\
= Z\cdot \int_\Sigma\Big((\dd \As^0_\oo - \dd \As^0_\ii)\,\As^{0,1}_\res+\frac12 (\dd \As^0_\oo+\dd\As^0_\ii)\,\bar\dd\As^0_\Ires+ \As^{1,0}_\oo\,\bar\dd\As^0_\oo-\As^{1,0}_\ii\,\bar\dd\As^0_\ii \Big) .
\end{multline}
On the other hand,
\begin{equation}\label{hol-hol mQME2}
\hbar^2\Delta_\res Z= Z\cdot \int_\Sigma \big(\As^+_\oo-\As^+_\ii-\dd\As^+_\Ires\big)\,\big(\dd\As^-_\res+\frac12 \bar\dd(\As^+_\oo+\As^+_\ii)\big) .
\end{equation}
Inspecting this expression, we see that it coincides with (\ref{hol-hol mQME1}), which proves (\ref{ab CS hol-hol mQME}).
\end{proof}

Following the terminology of \cite{CMR15}, we call the equation $(\Omega_\dd-\hbar^2\Delta_\res) Z=0$ the \textit{modified} (by the boundary term) quantum master equation (mQME).

\subsection{Antiholomorphic-to-holomorphic boundary conditions}\label{ss: ahol-hol}
Consider the polarization $\mr{Span}\{\frac{\delta}{\delta \As^+}\}$ at $t=0$ and $\mr{Span}\{\frac{\delta}{\delta \As^-}\}$ at $t=1$. I.e., we prescribe boundary values $\As^-_\ii, \As^+_\oo$. The corresponding modification of the action by boundary terms adjusting for the polarization is:
$$ S^\pol=S+\frac12\int_{\{1\}\times\Sigma}\As^+\As^-  + \frac12 \int_{ \{0\}\times\Sigma}\As^+\As^- $$
and the modified boundary Noether $1$-form is:
$$ \alpha^\pol_{\Sigma\times \dd I} = \int_{\{1\}\times\Sigma } \As^- \delta \As^+ -\int_{\{0\}\times\Sigma } \As^+\delta \As^- . $$

The fiber of the (complexified) space of fields over the space of boundary conditions $\BB=\g^-[1]\oplus \g^+[1]=\{(\As^-_\ii,\As^+_\oo)\}$ is the complex
\begin{equation*}
\YY= \Omega^\bt(I,\{1\};\g^+[1])\oplus \Omega^\bt(I,\{0\};\g^-[1]) .
\end{equation*}
Hodge decomposition (\ref{1d ab CS ahol-hol Hodge}) holds (where one replaces $\Pi\ra [1]$) and the propagator is given by (\ref{1d ab CS ahol-hol eta}) or, more explicitly,
\begin{multline*}
\eta((z,\bar{z},t)\, ; \, (z',\bar{z}',t'))=\\
=\delta^{(2)}(z-z')\frac{i}{2}(-dz\wedge d\bar{z}'+dz'\wedge d\bar{z}')\;\big(-\theta(t'-t)+2(1-t)\, t'\big)+\\
+\delta^{(2)}(z-z')\frac{i}{2}(d\bar{z}\wedge dz'+ dz\wedge d\bar{z})\;\big(\theta(t-t')-2t\,(1-t')\big) .
\end{multline*}

The space of residual fields is:
\begin{equation}\label{ahol-hol V}
\VV=dt\cdot \g_{\CC}\oplus (1-t)\cdot \g^+[1] \oplus t\cdot \g^-[1] = \{dt\cdot \As_\Ires + (1-t)\cdot \As^+_\res + t\cdot \As^-_\res \} ,
\end{equation}
where $\As_\Ires$, $\As^+_\res$, $\As^-_\res$ are $t$-independent forms on $\Sigma$. The homogeneous components of these residual fields
and their internal degrees (ghost numbers) are as follows:\\
\begin{tabular}{lllll|lll|lll}
 $\As_{\Ires}=$ & $\As_\Ires^0+$ & $\As_\Ires^{1,0}+$ & $\As_\Ires^{0,1}+$ & $\As_\Ires^2$ & $\As^+_\res=$ & $\As^0_\res+$ & $\As^{1,0}_\res$ & $\As^-_\res=$ & $\As^{0,1}_\res+$ & $\As^2_\res$  \\
& $0$ & $-1$ & $-1$ & $-2$ & & $1$ & $0$ & & $0$ & $-1$
\end{tabular}

The BV-BFV partition function is:
\begin{multline}\label{ab CS ahol-hol Z}
Z(\As^-_\ii,\As^+_\oo;\As_\Ires, \As^+_\res,\As^-_\res)=\\
=\int_{\YY'_{K-ex}\subset \YY'} \mc{D}\As^+_\fl\,\mc{D}\As^-_\fl\; e^{\frac{i}{\hbar}S^\pol\Big(\As^-_\ii+\As^+_\oo+(1-t)\cdot \As^+_\res+t\cdot \As^-_\res+\As^+_\fl+\As^-_\fl+dt\cdot \As_\Ires\Big)}\\
=\int \mc{D}\As^+_\fl\,\mc{D}\As^-_\fl\; e^{\frac{i}{\hbar}\Big(\int_{ I\times\Sigma} \big(\As^-_\fl+t \As^-_\res\big)\,\dt\big(\As^+_\fl+(1-t)\As^+_\res\big) +\int_{\{1\}\times\Sigma } \As^+_\oo \As^- - \int_{\{0\}\times\Sigma } \As^-_\ii\As^+ +\int_{ I\times\Sigma} \frac12 \A\,d_\Sigma \A  \Big)}
\\
=\int \mc{D}\As^+_\fl\,\mc{D}\As^-_\fl\; \exp \frac{i}{\hbar} \Big( \int_{ I\times\Sigma} \As^-_\fl \, \dt \As^+_\fl + \frac12 \int_\Sigma \As^-_\res\As^+_\res+\int_\Sigma \As^+_\oo(\As^-_\res+\As^-_\fl\big|_{t=1}) - \\
- \int_\Sigma \As^-_\ii (\As^+_\res+\As^+_\fl\big|_{t=0}) +\frac12 \int_\Sigma 
\As_\Ires \; d_\Sigma (\As^+_\res+\As^-_\res)
\Big) \\
=\exp\frac{i}{\hbar}\int_\Sigma \Big( -\As^+_\oo\As^-_\ii+\As^+_\oo \As^-_\res -\As^-_\ii\As^+_\res +\frac12 \As^-_\res \As^+_\res +\\
+\frac12 \big(
\As^{0,1}_\Ires \dd\As^0_\res+  \As^{1,0}_\Ires\bar\dd \As^0_\res + \As^0_\Ires (\dd \As^{0,1}_\res+ \bar\dd \As^{1,0}_\res)
\big)
\Big) .
\end{multline}
Here the first term in the final result is a contribution of the diagram where $\As^+_\oo$ is contracted by a propagator with $\As^-_\ii$.

\subsubsection{Partial integral over residual fields and comparison with Hamilton--Jacobi action}
Motivated by comparison with the Hamilton--Jacobi formalism, we consider the BV pushforward of the partition function (\ref{ab CS ahol-hol Z}) along the odd symplectic fibration
$$ p:\VV \ra \VV_\mr{small}=\{dt\cdot(\As^0_\Ires+\As^2_\Ires)+(1-t)\cdot \As^0_\res + t\cdot \As^2_\res\} . $$
In its kernel, we choose the gauge-fixing lagrangian subspace $\mc{L}$ cut out by equations $\As^{1,0}_\Ires=\As^{0,1}_\Ires=0$ and parametrized by $\As^{1,0}_\res, \As^{0,1}_\res$. The corresponding BV pushforward is:
\begin{multline}\label{ab CS hol-hol Z_small}
Z_\mr{small}=\int \D\As^{1,0}_\res\,\D \As^{0,1}_\res \; Z = \\
=\exp\frac{i}{\hbar}
\Big(\underbrace{\int_\Sigma\Big(
\As^{1,0}_\oo \As^{0,1}_\ii
+(\dd \As^{0,1}_\ii+\bar\dd \As^{1,0}_\oo )\; \sigma
 +\frac12 \sigma \dd\bar\dd \sigma
\Big)}_{S_\text{HJ}}+\\+
\int_\Sigma\Big(-\As^0_\oo\As^2_\ii +\As^0_\oo \As^2_\res -\As^2_\ii \As^0_\res+ \frac12 \As^2_\res \As^0_\res  \Big)
\Big) .
\end{multline} 
Here we denoted the degree zero scalar residual field by $$\sigma:=\As^0_\Ires\;\; \in \Omega^0_{\CC}(\Sigma) .$$ 
In the first bracket in (\ref{ab CS hol-hol Z_small}) we recognize the Hamilton--Jacobi action 
\cite[Eq. (47)]{HJ} (see also Example \ref{ex:HJ ab CS}) --- the action of a free (conformal) massless boson interacting with the boundary fields,\footnote{\label{fn:abWZW}
One can see $S_\text{HJ}$ as the abelian version of the gauged Wess--Zumino--Witten theory, see e.g (2.7) in \cite{Gawedzki_coset}.
} while in the second bracket we collected the contributions of nonzero-degree fields.

\subsubsection{Full integral over residual fields}\label{sec: full int ahol-hol}
If we wish to integrate out the remaining residual fields completely, we construct the gauge-fixing lagrangian $\mc{L}_\mr{small}\subset \VV_\mr{small}$ as follows. Choose an area form $\mu$ on $\Sigma$. 
Consider the splitting of $0$-forms into constants and forms with vanishing integral against $\mu$: $\As^0=\As^0_\const+\underline{\As^0}$.
Also, consider the splitting of $2$-forms into constant multiples of $\mu$ and forms of vanishing total integral: $\As^2=\mu\cdot \As^2_\const+\underline{\As^2}$. Then, we define the lagrangian $\mc{L}_\mr{small}\subset \VV_\mr{small}$ by equations $\As^2_\Ires = \sigma_{\const}=\underline{\As^2_\res}=0$. Thus, the lagrangian is parametrized by $\As^0_\res$, $\underline{\sigma}$, $\As^2_{\res,\const}$.\footnote{We have to split off the constants from $\sigma$, because they are in the kernel of the Laplacian $\dd\bar\dd$ and thus would obstruct the evaluation of the integral of (\ref{ab CS hol-hol Z_small}) over $\sigma$.}

The resulting full BV integral is:
\begin{multline*}
Z_*= \int_{
\mc{L}_\mr{small}\subset \VV_\mr{small}
} \D \As^0_\res \; \D \underline{\sigma}\;\D \As^2_{\res,\const}\;\;  Z_\mr{small}=\\
= \delta(\underline{\As^2_\ii})e^{-\frac{i}{\hbar}\int_\Sigma \mu\cdot \As^2_{\ii,\const}\As^0_{\oo,\const} }
\int \D \underline{\sigma}\; \exp\frac{i}{\hbar}\int_\Sigma 
\Big(
\As^{1,0}_\oo \As^{0,1}_\ii
+(\dd \As^{0,1}_\ii+\bar\dd \As^{1,0}_\oo )\; \underline{\sigma} +\frac12 \underline{\sigma} \dd\bar\dd \underline{\sigma}
\Big) .
\end{multline*}
Further, assume that the area form $\mu = \sqrt{\det{g}}\; d^2x $ is the Riemannian area form associated to a certain metric $g$ on $\Sigma$ inducing simultaneously the complex structure we use in our polarization. Then the integral over $\underline{\sigma}$ evaluates finally to
\begin{equation}\label{ab CS ahol-hol Z fully integrated}
Z_*=
\delta(\underline{\As^2_\ii})\; e^{-\frac{i}{\hbar}\int_\Sigma \mu\cdot  \As^2_{\ii,\const}\As^0_{\oo,\const} }
\; \cdot  \big({\det}'_{\Omega^0(\Sigma)}\Delta_g\big)^{-\frac12}\cdot e^{\frac{i}{\hbar} \mathbb{I}(\As^{1,0}_\oo,\As^{0,1}_\ii) }
\end{equation}
where 
\begin{itemize}
\item 
$\Delta_g$ is the metric Laplace operator acting on $0$-forms, ${\det}'$ means the zeta-regularized product of \textit{nonzero} eigenvalues.
\item The exponent in (\ref{ab CS ahol-hol Z fully integrated}) is
\begin{equation}
\mathbb{I}= \int_\Sigma \As^{1,0}_\oo P_\mr{harm}( \As^{0,1}_\ii )-i\int_{\Sigma\times\Sigma\;\ni (z,z')} 
\bar\dd \As^{1,0}_\oo\big|_z G(z,z')\, \bar\dd \As^{1,0}_\oo\big|_{z'} +  \dd \As^{0,1}_\ii\big|_z G(z,z')\, \dd\As^{0,1}_\ii \big|_{z'} . \label{eq:def I}
\end{equation}
Here $G
$ is the Green's function for $\Delta_g$, viewed as a function on $\Sigma\times\Sigma$ with a logarithmic singularity at the diagonal.\footnote{Recall that, in terms of Dolbeault operators and the area form, the metric Laplace operator is: $\Delta_g=\frac{2i}{\mu}\dd\bar\dd$.} 
The operator $P_\mr{harm}: \As^{0,1}\mapsto \As^{0,1}- 2i\int_{\Sigma\ni z'} \bar\dd G(z,z') \dd \As^{0,1}\big|_{z'}$ is the projector onto harmonic $(0,1)$-forms in the Hodge decomposition.

\end{itemize}

Written in different notations, the exponent in (\ref{ab CS ahol-hol Z fully integrated}) is:
\begin{equation}\label{ab CS ahol-hol I via d dbar}
\mathbb{I}=\int_\Sigma \As^{1,0}_\oo\,  (1-\bar\dd (\dd\bar\dd)^{-1}\dd) \,\As^{0,1}_\ii -\frac12 \As^{1,0}_\oo \, \bar\dd (\dd\bar\dd)^{-1} \bar\dd \, \As^{1,0}_\oo-\frac12 \As^{0,1}_\ii \, \dd (\dd\bar\dd)^{-1} \dd\, \As^{0,1}_\ii  .
\end{equation}


\begin{remark}\label{rem: ahol-hol fully integrated}
\begin{enumerate}[(a)]
\item The exponent $\mathbb{I}$ in (\ref{ab CS ahol-hol Z fully integrated}) depends only on the complex structure on $\Sigma$, not on the particular compatible metric $g$. In other words, it is invariant under Weyl transformations of the metric $g\mapsto e^{\phi}\, g$. Weyl-invariance of $\mathbb{I}$ is manifest in the form (\ref{ab CS ahol-hol I via d dbar}).  
\item \label{rem: ahol-hol fully integrated (b)} 
Unlike $\mathbb{I}$, the full quantum answer (\ref{ab CS ahol-hol Z fully integrated}) is not Weyl-invariant, since the determinant of the Laplacian is not invariant (a phenomenon known as the ``conformal anomaly" or ``trace anomaly'' of the free scalar field as a conformal field theory).
In addition to that quantum effect, the dependence of $Z_*$ on boundary $\gh\neq 0 $ fields $\As_\ii^2, \As_\oo^0$ involves the metric area form $\mu$.
\item
The lagrangian generated by $\mathbb{I}$ is
\begin{multline*} 
L_\mathbb{I} =\left\{
\begin{array}{c}
\As_\oo = \As_\oo^{1,0} + \frac{\delta\mathbb{I}}{\delta \As^{1,0}_\oo},\\
\As_\ii = \As_\ii^{0,1} - \frac{\delta\mathbb{I}}{\delta \As^{0,1}_\ii}
\end{array}
\right\} 
= \left\{  
\begin{array}{c}
\As_\oo=(1-\bar\dd (\dd\bar\dd)^{-1}\bar\dd) \As^{1,0}_\oo + 
P_\mr{harm} 
\As^{0,1}_\ii ,\\
\As_\ii=(1+\dd (\dd\bar\dd)^{-1}\dd) \As^{0,1}_\ii + P_\mr{harm} 
\As^{1,0}_\oo 
\end{array}
 \right\} .
\end{multline*}
It is easy to see that 
this lagrangian coincides with
the evolution relation of abelian Chern--Simons theory on the cylinder $ I\times\Sigma$,
$$L_{CS}=\big\{(\As_\oo,\As_\ii)\; \in \Omega^1(\Sigma)\times \Omega^1(\Sigma)\;\;\big| \;\; d\As_\oo=0,\; d\As_\ii=0,\; \As_\oo-\As_\ii=d(\cdots)\big\} .$$
Thus, $\mathbb{I}$ is a (nongeneralized\footnote{I.e., with no auxiliary fields.}) Hamilton--Jacobi action for the abelian theory on the cylinder.
\item
Classically, one can obtain $\mathbb{I}$ from the generalized Hamilton--Jacobi action (Example \ref{ex:HJ ab CS})
 as the conditional extremum of $S_\text{HJ}$ in $\sigma$, with $\As^{1,0}_\oo$ and $\As^{0,1}_\ii$ fixed. 
\end{enumerate}
\end{remark}

\begin{remark}\label{rem: Liouville}
To make (\ref{rem: ahol-hol fully integrated (b)}) of Remark \ref{rem: ahol-hol fully integrated} above more explicit: if $g_\tau=e^{\phi_\tau}g_0 $ is a $\tau$-dependent family of metrics compatible with the given complex structure on $\Sigma$, one has
\begin{equation}\label{trace anomaly}
\frac{d}{d\tau} Z_*^{g_\tau} = (\Omega^+_\oo+\Omega^-_\ii)(\xi Z_*^{g_\tau})+Z_*^{g_\tau}\cdot \frac{1}{48\pi}\int_\Sigma \mu_{g_\tau} R_{g_\tau}\dot{\phi}_\tau  
\end{equation}
with $R$ the scalar curvature of the metric, $\mu_{g_\tau}$ the Riemannian area form of $g_\tau$, the $\Omega$ operators given by (\ref{Omega hol out}), (\ref{Omega ahol in}) below and\footnote{This is the quantization of the hamiltonian $H = \As^2_{\res,c}\int_\Sigma\underline{\sigma}\dot{\mu}$ generating the family of lagrangians given by $\mu_{g^\tau}$, see \cite[Section 2]{CMR15}. } $\xi =\As_{\ii,c} \int_{\Sigma\times\Sigma}(\bar\dd A^{1,0}_\oo + \dd A^{0,1}_\ii)_zG(z,z')\dot{\mu}_{z'}$. The second term in (\ref{trace anomaly}) is the trace anomaly.
Furthermore, one can compensate the anomaly term by including the Liouville action as a counterterm,\footnote{We are referring to the fact that in a conformal field theory with central charge $c$, the partition function has the following behavior under Weyl transformations of metric:
$Z_\mr{CFT}^{e^{\phi}g} =Z_\mr{CFT}^{g} \cdot  e^{\frac{c}{48\pi}\int_\Sigma \frac12 d\phi\wedge * d\phi +R\phi \mu_g} $, see, e.g., \cite{Gawedzki CFT}.
} i.e.,  by introducing
\begin{equation*}
\widehat{Z}^{g}=Z_*^g\cdot  e^{-\frac{1}{48\pi}\int_\Sigma \frac12 d\phi \wedge  * d\phi +R_g\phi \mu_g },
\end{equation*}
where $\phi$ is defined by $g=e^{\phi}g_0$ with $g_0$ some ``reference'' metric in the same conformal class   (e.g., one can take $g_0$ to be the uniformization metric on $\Sigma$ --- 
spherical, flat or hyperbolic metric for  $\Sigma$  of genus $0$, $1$ or $\geq 2$, respectively). Then, for a conformal variation of metric we have $\delta_{\varphi} \widehat{Z}^{e^\varphi g} =(\Omega^+_\oo+\Omega^-_\ii)(\xi Z_*^g)$.

As an aside, it is tempting to compare the two phenomena:
\begin{enumerate}[(i)]
\item \label{rem: Liouville (i)} The anomalous metric dependence (under a Weyl transformation $g_\Sigma \ra e^\phi g_\Sigma$) of the partition function on the cylinder and the cancellation of that dependence by a Liouville action counterterm.
\item \label{rem: Liouville (ii)}The anomalous metric dependence (under $g_M\ra g_M+\delta g_M$) of the perturbative Chern--Simons partition function on a closed 3-manifold $M$ and the cancellation of that dependence by the gravitational Chern--Simons counterterm introducing the dependence on framing $M$, see \cite{Witten,AS}.
\end{enumerate}
But in fact, these effects seem different. In particular, the dependence on Weyl transformations in (\ref{rem: Liouville (i)}) rescales $Z$ by a real factor, whereas the anomalous metric dependence in (\ref{rem: Liouville (ii)}) affects only the phase of the partition function.

\end{remark}

\begin{remark}
As implied by (\ref{ab CS hol-hol Z_small}),
one can view the ``physical part'' $Z_*^\mr{ph}=e^{\frac{i}{\hbar}\mathbb{I}(\As^{1,0}_\oo,\As^{0,1}_\ii) }$ of (\ref{ab CS ahol-hol Z fully integrated}) as a generating function for the correlators of chiral currents $j=i\dd \phi$, $\bar{j}=i\bar\dd \phi$ in massless scalar theory (viewed as abelian WZW model): 
\begin{multline}
\langle j(z_1)\cdots j(z_n)\bar{j}(w_1)\cdots \bar{j}(w_m) \rangle =\\
=\hbar^{n+m} \frac{1}{ Z^\mr{ph}_*}\left.\frac{\delta}{\delta \As^{0,1}_\ii(z_1)}\cdots \frac{\delta}{\delta \As^{0,1}_\ii(z_n)} \frac{\delta}{\delta \As^{1,0}_\oo(w_1)}\cdots \frac{\delta}{\delta \As^{1,0}_\oo(w_m)} \right|_{\As^{0,1}_\ii=\As^{1,0}_\oo=0} Z^\mr{ph}_*
\end{multline}
Here we are assuming that all points $z_1,\ldots,z_n,w_1,\ldots,w_m$ in $\Sigma$ are pairwise distinct (so that we can ignore the term $\int_\Sigma\As^{1,0}_\oo\As^{0,1}_\ii$ in (\ref{ab CS hol-hol Z_small})).
We note that (\ref{ab CS ahol-hol I via d dbar}) implies a short-distance behavior of these correlators consistent with the OPEs (operator product expansions) 
\begin{equation}
j(z) j(w)\sim \frac{1}{(z-w)^2}+\mr{reg},\quad \bar{j}(z)\bar{j}(w)\sim \frac{1}{(\bar{z}-\bar{w})^2}+\mr{reg},\quad j(z)\bar{j}(w)\sim \mr{reg}
\end{equation}
as  $z\ra w$ (``reg'' stands for the regular part) -- the standard fundamental OPEs of abelian WZW model.
\end{remark}

\subsubsection{Quantum master equation}
The space of states on the out-surface with $\As^+$-fixed polarization was discussed in Section \ref{sss: hol-hol QME}: it is the space of functions of $\As^+_\oo$ of the form (\ref{states in A^+ pol}) with the BFV operator
\begin{equation} \label{Omega hol out}
\Omega_\oo^+ = \int_\Sigma \left( -i\hbar\; \dd \As^0_\oo\frac{\delta}{\delta \As^{1,0}_\oo}+\As^{1,0}_\oo \,\bar\dd \As^0_\oo\right) .
\end{equation}
The space of states on the in-surface with $\As^-$-fixed polarization is the space of functions of $\As^-_\ii$ (defined similarly to (\ref{states in A^+ pol})) with the BFV operator
\begin{equation}\label{Omega ahol in}
\Omega_\ii^- =\int_\Sigma \Big( -i\hbar\; \dd\As^{0,1}_\ii\frac{\delta}{\delta \As^2_\ii}+\hbar^2\frac{\delta}{\delta \As^{0,1}_\ii}\,\bar\dd \frac{\delta}{\delta \As^2_\ii} \Big) .
\end{equation}
This is the quantization of the BFV action (\ref{S BFV via hol splitting}) where $\As^{0,1},\As^2$ become multiplication operators and $\As^{1,0}\mapsto -\epsilon\, i\hbar \frac{\delta}{\delta \As^{0,1}}$, $\As^0\mapsto -\epsilon\, i\hbar \frac{\delta}{\delta \As^2}$ become derivations, where $\epsilon=-1$ for the in-boundary, as in Section \ref{sss: hol-hol QME}.

The BV Laplacian on residual fields (\ref{ahol-hol V}) is:\footnote{Here the factor $2$ comes from the fact that the odd symplectic form on $\VV$, induced from the standard BV 2-form on the space of fields, $\omega_\FF=\int_{ I\times\Sigma}\frac12 \delta\A\,\delta \A$, is $\omega_\VV=\frac12 \int_\Sigma \delta \As^+_\res\, \delta\As^-_\Ires + \delta \As^-_\res\, \delta\As^+_\Ires$. The factor $\frac12$ in the latter expression comes from $\int_0^1 dt (1-t)=\int_0^1 dt\, t=\frac12$.}
\begin{equation*}
\Delta_\res=2\int_\Sigma\frac{\delta}{\delta \As^-_\res} \frac{\delta}{\delta \As^+_\Ires}+
\frac{\delta}{\delta \As^+_\res} \frac{\delta}{\delta \As^-_\Ires} .
\end{equation*}

\begin{lemma}
Partition function (\ref{ab CS ahol-hol Z}) satisfies the modified quantum master equation
\begin{equation*}
\left( \Omega^+_\oo + \Omega^-_\ii-\hbar^2 \Delta_\res \right)\;Z =0 .
\end{equation*}
\end{lemma}

\begin{proof}
Indeed, a straightforward computation gives:
\begin{align*}
\Omega^+_\oo Z&=Z\cdot\int_\Sigma\Big( -\dd \As^0_\oo\,(\As^{0,1}_\ii-\As^{0,1}_\res)+ \As^{1,0}_\oo\;\bar\dd\As^0_\oo \Big), \\
\Omega^-_\ii Z&=Z\cdot \int_\Sigma\Big( 
\dd \As^{0,1}_\ii\, (\As^0_\oo-\As^0_\res)-(\As^{1,0}_\oo-\As^{1,0}_\res)\;\bar\dd (\As^0_\oo-\As^0_\res)
\Big), \\
-\hbar^2 \Delta_\res Z & = Z\cdot \int_\Sigma \Big(
(\dd \As^{0,1}_\res+\bar\dd \As^{1,0}_\res)\, (-\As^0_\oo+\frac12 \As^0_\res)+\\ \nonumber
 +\bar\dd\As^0_\res\, & (-\As^{1,0}_\oo+\frac12 \As^{1,0}_\res+\frac12 \dd\As^0_\Ires)
+\dd \As^0_\res\, (\As^{0,1}_\ii-\frac12 \As^{0,1}_\res+\frac12 \bar\dd \As^0_\Ires)
 \Big ) .
\end{align*}
The sum of these three expressions is zero.
\end{proof}

Similarly, one can check the quantum master equation for the ``small'' realization (\ref{ab CS hol-hol Z_small}):
\begin{equation*}
\left( \Omega^+_\oo + \Omega^-_\ii-\hbar^2 \Delta_\mr{small} \right)\;Z_\mr{small} =0,
\end{equation*}
where $$\Delta_\mr{small}=2\int_\Sigma  \frac{\delta}{\delta \As^2_\res}\frac{\delta}{\delta\sigma}
+\frac{\delta}{\delta \As^0_\res}\frac{\delta}{\delta \As^2_\Ires}$$
is the BV Laplacian on $\VV_\mr{small}$.

Finally, the result of the full integration over residual fields (\ref{ab CS ahol-hol Z fully integrated}) satisfies the BFV cocycle (gauge-invariance) condition
\begin{equation*}
\left( \Omega^+_\oo + \Omega^-_\ii\right)\; Z_*=0.
\end{equation*}

\section{Chern--Simons theory in ``parallel ghost polarization''}\label{s: par ghosts}
In three-dimensional Chern--Simons theory there is another way of picking a pair of polarizations on the opposite sides of a cylinder: we can use the $(\As^0_\oo,\As^{0,1}_\oo)$ representation on the out-boundary surface and the $(\As^0_\ii,\As^{1,0}_\ii)$ representation on the in-surface. Thus, the corresponding polarizations are transversal in ghost number 0 and parallel in ghost number $\neq 0$. See also the discussion of quantization of 1D systems with this class of polarizations in \cite[Section 11]{HJ}.
\subsection{One-dimensional Chern--Simons theory with values in a cochain complex
}\label{sec: 1d par ghosts complex}
As a warm-up, we consider again the one-dimensional theory, with a slightly different setup. Fix an odd integer $k$. Let $$\g = \bigoplus \g^i$$ 
be a graded vector space with a differential $d_\g$ and a compatible graded symmetric pairing $(\cdot,\cdot)$ of degree $-2k$.\footnote{The prime example being $\g = \Omega^\bullet(M)$ with $M$ a $2k$-manifold, $d_\g$ the de Rham differential and $(\alpha,\beta) = \int_M \alpha \wedge \beta$.} Now, we let $X=\g[k]$ - this is a 0-shifted graded symplectic vector space. We call the induced grading on $C^\infty(X)$ the ghost number. It is convenient to express elements of $C^\infty(X)$ in terms of the shifted identity map $\psi \in \mr{Hom}(X,\g)$ which has total degree (ghost number + degree) $k$.\footnote{
If $\{t_a\}$ is a basis in $\g$ and $\{\psi^a\}$ is the shifted dual basis in $X^*$, then we have $\psi=\sum_a t_a \psi^a \in \mr{Hom}(X,\g)$. If the degree of $t_a$ is $|a|$, then the ghost number of $\psi^a$ is $k-|a|$, and thus $\psi$ is indeed an object of total degree $k$.
}
 We denote the ghost number $l$ component of a field $\varphi$ by $\varphi^{[l]}$. In particular, the ghost number $l$ component $\psi^{[l]}$ of $\psi$ has degree $k - l$. For instance, the function 
 $$\Theta(\psi) = \frac12(\psi,d_\g\psi)$$ 
 has ghost number $+1$. Its hamiltonian vector field $Q$ has ghost number $+1$ and satisfies $Q^2 = 0$, thus $(X,\omega,Q)$ is a BFV vector space. 
 
We split the complexification of the ghost number $0$ component of $X$ 
as
 $X^{[0]}_\CC = X^+ \oplus X^-$, with $X^\pm$ the degree 0 $\pm i$-eigenspaces of a complex structure $J$ on $X^{[0]}$ compatible with the pairing. Thus, $X_\CC$ admits the total decomposition 
\begin{equation}
X_\CC = X^{[>0]}_\CC \oplus X^+ \oplus X^- \oplus X^{[<0]}_\CC\label{eq: par ghost split g}
\end{equation}
where $X^{[>0]},X^{[<0]}$ are the components of positive (resp.\ negative) ghost number.\footnote{I.e., coordinates, e.g., on $X^{[>0]}$ have positive ghost number so $X^{[>0]} = \bigoplus_{i<k}\g^i[k]$.} We also introduce the notations $d_\g^+,d_\g^-$ for the composition of the differential $X^{[1]} \to X^+ \oplus X^-$ with projections and similarly for the restriction of the differential ${ X^- \oplus X^+ \to X^{[-1]} }$ (so that $d_\g^-=d_\g\big|_{X^+}$, $d_\g^+=d_\g\big|_{X^-}$). We automatically have  
$(d_\g^+)^2 = (d_\g^-)^2 = 0$ and ${d_\g^+d_\g^- = - d_\g^-d_\g^+}$.\footnote{A particularly important example is the case $\g = \Omega^\bullet(M)$ where $M$ is a $2k$-manifold, with $k=2l+1$ odd, endowed with a complex structure that allows us to decompose $\Omega^k_\CC(M) = \Omega^+(M) \oplus \Omega^-(M)$ into lagrangian subspaces with respect to $(\alpha,\beta)=\int_M \alpha\wedge \beta$, with the splitting given by (\ref{Omega^k splitting into + and -}) below. 
In the case $k=1$, we have $d^+_\g = \dd,d^-_\g = \bar{\dd}$.}
\subsubsection{Setup}
We now consider the 1-dimensional AKSZ theory with target the symplectic graded vector space $(X,(\cdot,\cdot))$ and hamiltonian $\Theta(\psi)$. 
The space of fields is 
$$\FF = \Omega^\bullet(I;X).$$
It is parametrized by the superfield $\mathcal{A}$ valued in $\Omega^\bullet(I;\g)$. We denote the 0- and 1-form components of $\mathcal{A}$  by $\psi$ and $A$, respectively.
The total degrees of $\A,\psi,A$ are all odd. 
The action is 
$$S[\psi + A]= \frac12\int_I({\A},d_I\A) +\frac12 \int_I(\A,d_\g \A) = \frac{1}{2}\int_I (\psi,d_I\psi) +\int_I (A,d_\g\psi).$$
The 
space of boundary fields is 
$$\FF^\dd = X_{\ii} \times X_{\oo} \ni (\psi_\ii,\psi_\oo).$$
The boundary 1-form is 
$$\alpha_{\dd I} = \alpha^{\dd}_\oo + \alpha^{\dd}_\ii =  \frac12(\psi_\oo,\delta \psi_\oo) - \frac12(\psi_\ii,\delta \psi_\ii) = \frac12(\psi,\delta \psi)\bigg|_{t=0}^{t=1}.$$
Splitting elements of $X_\CC$ 
according to \eqref{eq: par ghost split g}, 
$\psi = \psi^{[<0]} + \psi^+ + \psi^- + \psi^{[>0]}$, 
 the boundary 1-form splits similarly: 
$$\alpha^{\dd}_\oo =\frac12\Big[ (\psi_\oo^{[<0]},\delta \psi_\oo^{[>0]}) +  (\psi_\oo^+,\delta \psi_\oo^-)+ (\psi_\oo^-,\delta \psi_\oo^+)+  (\psi_\oo^{[>0]},\delta \psi_\oo^{[<0]})\Big]$$
and similarly for $\alpha^\dd_\ii$. 
\subsubsection{Parallel ghost polarization}
 Let us now consider the case where the polarizations are parallel in the ghost sector (of the target) and transversal in the physical sector: 
\begin{align*}\PP &= \PP_\ii \times \PP_\oo, \\ 
\PP_\ii &= \left\lbrace\frac{\delta}{\delta \psi^{[<0]}_\ii},\frac{\delta}{\delta \psi_\ii^{+}}\right\rbrace ,\\
\PP_\oo &= \left\lbrace\frac{\delta}{\delta \psi^{[<0]}_\oo},\frac{\delta}{\delta \psi_\oo^{-}}\right\rbrace ,
\end{align*}
so that we are using the $(\psi^{[>0]}_\ii,\psi^{-}_\ii)$ representation at $t=0$ and the $(\psi^{[>0]}_\oo,\psi^{+}_\oo)$ representation at $t=1$, i.e.:
$$\BB = \BB_\ii \times \BB_\oo, \qquad \BB_\ii = X^{[ > 0]}_\CC \oplus X^-, \qquad \BB_\oo = X^{ [> 0]}_\CC \oplus X^+.$$
The polarized 1-form is 
$\alpha^\pol_{\dd I} =\alpha^{\dd,\PP_\oo} + \alpha^{\dd,\PP_\ii}$ with 
\begin{align*}
\alpha^{\dd,\PP_\ii} &= -(\psi^{[<0]}_\ii,\delta \psi^{[>0]}_\ii) - (\psi_\ii^{+},\delta \psi_\ii^-) = \alpha^{\dd}_\ii - \delta f_\ii, \\
\alpha^{\dd,\PP_\oo} & =(\psi^{[<0]}_\oo,\delta \psi^{[>0]}_\oo) + (\psi_\oo^-,\delta \psi_\oo^{+}) = \alpha^{\dd}_\oo + \delta f_\oo,
\end{align*}
where\footnote{Since $\psi$ is odd, we have, e.g., $\delta(\psi^+,\psi^-) = (\psi^-,\delta\psi^+) - (\psi^+,\delta\psi^-)$.}
 $f_\ii =\frac12 (\psi^{[>0]}_\ii,\psi^{[<0]}_\ii) - \frac12(\psi^+_\ii,\psi^-_\ii)$, $f_\oo = \frac12 (\psi^{[>0]}_\oo,\psi^{[<0]}_\oo) + \frac12(\psi_\oo^+,\psi_\oo^-)$,
so that 
$\alpha^\pol_{\dd I} = \alpha_{\dd I} + \delta f$
 with 
\begin{multline*}
f (\psi_\oo,\psi_\ii) = f_\oo(\psi_\oo) - f_\ii(\psi_\ii) \\
= \frac12 (\psi^{[>0]}_\oo,\psi^{[<0]}_\oo) + \frac12(\psi_\oo^+,\psi_\oo^{-}) -\frac12 (\psi^{[>0]}_\ii,\psi^{[<0]}_\ii) + \frac12(\psi_\ii^+,\psi_\ii^-) .
\end{multline*}
The polarized action is 
$$S^{\pol}[\mathcal{A}]=  \frac12\int_I({\A},d_I\A) +\frac12 \int_I(\A,d_\g \A) + f(\mathcal{A}).$$
\subsubsection{Gauge fixing}
The kernel $\YY$ of the map $\mathcal{F} \to \mathcal{B}$ is 

 \begin{equation}\YY = \Omega^\bullet(I,\partial I; X^{[> 0]}_\CC) \oplus \Omega^\bullet(I,\{0\};X^-) \oplus \Omega^\bullet(I,\{1\};X^+) \oplus \Omega^\bullet(I;X^{[<0]}_\CC).\label{eq: fluc 1d par gh}
 \end{equation}
 
Choosing the minimal space of residual fields 
$$\VV = \langle dt \rangle \otimes X^{ [> 0]}_\CC \oplus \langle 1 \rangle \otimes X^{ [<0]}_\CC \ni {\A}_\res = dt \cdot A_\res + 1\cdot \psi_{\res}, $$
we obtain 
$$\YY = \VV \times \YY' \ni (\A_\res,\A_\fl)$$
with 
\begin{equation*}
\YY' = \Omega^\bullet(I,\partial I; X^{[> 0]}_\CC)_{\int= 0} \oplus \Omega^\bullet(I,\{0\};X^-) \oplus \Omega^\bullet(I,\{1\};X^+) \oplus \Omega^\bullet(I;X^{[<0]}_\CC)_{\int \cdot \wedge dt=0}.
 \end{equation*}
Here the notation $\int \cdot = 0$ (resp.\ $\int \cdot dt = 0$) denotes acylic subcomplexes of forms with vanishing integral (resp.\ forms whose product with $dt$ has vanishing integral).
Choosing an extension 
$${\widetilde{\psi}} = \widetilde{\psi^{[>0]}_\ii} + \widetilde{\psi_\ii^-} + \widetilde{\psi^{[>0]}_\oo} + \widetilde{\psi_\oo^{+}}$$ 
of boundary fields into the bulk, we obtain a splitting 
of
$\A = \psi
 +A $
into $$\psi + A = \widetilde{\psi} + \psi_\res +  \psi_\fl + dt \cdot A_\res + A_\fl.$$ 
Inside $\YY'$, we have the gauge-fixing lagrangian $\LL \subset \YY'$ given by forms of degree 0 in $I$  --- i.e., $\LL$ is given by $A_\fl = 0$ --- and we write for $\psi_\fl \in \LL$ 
$$\psi_\fl = \psi_\fl^{[>0]} + \psi^-_\fl + \psi_\fl^{+} + \psi_\fl^{[<0]} .$$ 
Recollecting, for a field $\psi+A \in \BB \times \VV \times \LL$ we obtain 
\begin{align}
\psi &= \widetilde{\psi} + \psi_\res + \psi_\fl \label{eq:psi split I}\\
&= \widetilde{\psi^{[>0]}_\ii} + \widetilde{\psi_\ii^-} + \widetilde{\psi^{[>0]}_\oo} + \widetilde{\psi_\oo^{+}} + \psi_\res + \psi_\fl^{[>0]} + \psi^-_\fl + \psi_\fl^{+} + \psi_\fl^{[<0]},  \notag \\
A &= dt \cdot A_{\res}. \notag
\end{align}

The gauge-fixed polarized action is then computed as follows:
\begin{lemma}\label{lem: S GF par ghosts}
Restricted to the gauge-fixing lagrangian, the polarized action can be written as 
\begin{equation} \label{1d CS parall ghost S^f on L}
 S^\pol[\As] = S_{\mr{source}}[\psi_\ii,\psi_\oo,\psi_\fl] + S_0[\psi_\fl] + S_{\mr{int}}[\psi_\res,\psi_\fl,A_\res] +S_{\mr{back}}[\psi_\ii,\psi_\oo,\psi_\res,A_\res] ,
\end{equation}
where 
 \begin{align*}
S_{\mr{source}}[\psi_\ii,\psi_\oo,\psi_\res,\psi_\fl]  &=  (\psi^+_\oo,\psi_\fl^-(1)) - (\psi^-_\ii,\psi_\fl^+(0)) + (\psi^{[>0]}_\oo, \psi^{[<0]}_\fl(1)) - (\psi^{[>0]}_\ii,\psi_\fl^{[<0]}(0)) ,\\
 S_{0}[\psi_\fl] &= \int_I(\psi_\fl^+,d_I\psi_\fl^-) + \int_I(\psi^{[<0]}_\fl,d_I\psi_\fl^{[>0]})  , \\
 S_{\mr{int}}[\psi_\fl,A_\res] &=  -\int_I\,dt(d_\g^+ A^{[1]}_{\res},\psi_\fl^-) - \int_I\,dt(d_\g^-A^{[1]}_\res,\psi^+_\fl) + \int_Idt (A_{\res}^{[>1]},d_\g\psi_\fl^{[<0]})  , \\
 S_{\mr{back}}[\psi_\ii,\psi_\oo,\psi_\res,A_\res] &=(\psi^{[>0]}_\oo-\psi^{[>0]}_\ii,\psi_{\res}) + (A_{\res}^{[>1]},d_\g\psi_{\res}) .
 \end{align*}

\end{lemma}
 Here we have introduced the notation $A^{[1]}_\res,A^{[>1]}_\res$ for the components of $A_{\res}$
valued in $X^{[1]}$ and in $X^{[>1]}$, respectively.\footnote{
We understand $A^{[k]}_\res$ as valued in $X^{[k]}$ but with ghost number $k-1$; the shift is due to the fact that $A^{[k]}_\res$ is a coefficient of a $1$-form on the source. This shift is a standard feature of the AKSZ construction. In particular, $A^{[1]}_\res$ is an object of ghost number \textit{zero}.
}
\begin{proof}
The polarized action is 
$$S^\pol[\psi+A] = \frac{1}{2}\int_I (\psi,d_I\psi) + \int_I(A,d_\g\psi) + f(\psi) ,
$$
where 
\begin{align*}
f(\psi) &= 
 f_\oo(\psi(1)) - f_\ii(\psi(0))  
 \\
&= \frac12( \psi^{[>0]}_\oo,\psi^{[<0]}_\fl(1)) + \frac12(\psi^+_\oo,\psi^-_\fl(1)) -\frac12(\psi^{[>0]}_\ii,\psi^{[<0]}_\fl(0)) + \frac12 (\psi^+_\fl(0),\psi^-_\ii) .
\end{align*}
Splitting $ \psi$ as in \eqref{eq:psi split I} and letting the support of $\widetilde{\psi}$ go towards $\dd I$ we obtain that 
\begin{align*}
\frac12 \int_I (\psi,d_I\psi) &= \frac12\int_I (\widetilde{\psi} + \psi_\res + \psi_\fl,d_I(\widetilde{\psi} + \psi_\res + \psi_\fl)) \\
&=\frac12(\widetilde{\psi},\psi_\res)\bigg|^{t=1}_{t=0} + \frac12(\widetilde{\psi},\psi_\fl)\bigg|^{t=1}_{t=0} +  \frac{1}{2}\int_I(\psi_\fl,d_I\psi_\fl) \\
&= \frac12(\psi^{[>0]}_\oo-\psi^{[>0]}_\ii,\psi_{\res})  + \frac12(\psi^{[>0]}_\oo,\psi_\fl^{[<0]}(1)) + \frac12(\psi_\oo^{+},\psi_\fl^-(1))\\
&-\frac12(\psi^{[>0]}_\ii,\psi^{[<0]}_\fl(0)) - \frac12 (\psi_\ii^-,\psi_\fl^+(0)) 
+ \int_I(\psi_\fl^+,d_I\psi_\fl^-) + \int_I(\psi_\fl^{[<0]},d_I\psi_\fl^{[>0]})\end{align*}
and 
\begin{align*}
\int_I (dt \cdot A_\res,d_\g\psi) 
=&-\int_I\,dt(d_\g^+A^{[1]}_{\res},\psi_\fl^-) - \int_I\,dt(d_\g^-A^{[1]}_\res,\psi^+_\fl) \\
&+ \int_Idt (A_{\res}^{[>1]},d_\g\psi_\fl^{[<0]}) + (A_{\res}^{[>1]},d_\g\psi_{\res}).
\end{align*}
Collecting the various terms, we obtain (\ref{1d CS parall ghost S^f on L}). 
\end{proof}
Notice that adding $f$ has the effect of doubling the boundary source terms. 

 \subsubsection{Effective action} 

The effective action is defined by 
\begin{align*}
Z &= e^{\frac{i}{\hbar}S_\eff[\psi_\ii,\psi_\oo,\psi_\res,A_\res]} = \int\DD\psi_\fl\, e^{\frac{i}{\hbar}S^f[\psi_\ii,\psi_\oo,\psi_\res,\psi_\fl,A_\res]} \\
&=e^{\frac{i}{\hbar}S_{\mr{back}}}\int\DD\psi_\fl\,  e^{\frac{i}{\hbar}(S_{\mr{source}}+S_{0}+S_{\mr{int}})}
\end{align*} 
where the integral is defined in terms of Feynman diagrams. 
\begin{proposition}\label{prop: 1d par ghosts Seff}
The effective action is given by 
\begin{align*}
S_\eff &=   S_\ph+S_\gh\quad \mbox{with} \\
S_\ph &= (\psi^+_\oo,\psi^-_\ii) + (\psi^-_\ii,d_\g^+A^{[1]}_\res) + (\psi_\oo^+,d_\g^-A^{[1]}_\res) + \frac12 (d_\g^- A^{[1]}_\res,d_\g^+ A^{[1]}_\res), \\
S_\gh &= (\psi^{[>0]}_\oo - \psi^{[>0]}_\ii,\psi_\res) + (A^{[>1]}_\res,d_\g\psi_\res) .
\end{align*}
\end{proposition}
\begin{proof}
In terms of Feynman diagrams, $S_{\mr{source}}$ generates boundary vertices (figure \ref{fig:bdryvertices}), $S_{\mr{int}}$ generates bulk vertices (Figure \ref{fig:bulkvertices}).\footnote{We are ignoring here the ghost vertices that in this case do not contribute to the effective action since no vertex carries as $\psi^{[>0]}$.} \\
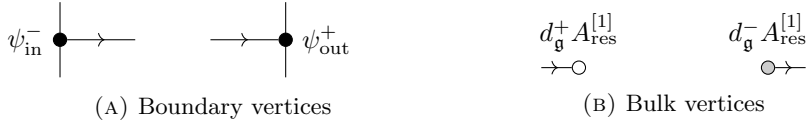
\begin{figure}[h]
\begin{subfigure}{0.45\linewidth}
\begin{tikzpicture}
\draw (0,.5) -- (0,-.5);
\node[coordinate] (b1) at (1,0) {};
\node[bdry,label=left:{$\psi_\ii^-$}] at (0,0) {} edge[fermion] (b1); 
\begin{scope}[shift={(3,0)}]
\draw (0,.5) -- (0,-.5);
\node[bdry,label=right:{$\psi_\oo^+$}] (bdry1) at (0,0) {} ; 
\node[coordinate] at (-1,0) {} edge[fermion] (bdry1);
\end{scope}
\end{tikzpicture}
\caption{Boundary vertices} \label{fig:bdryvertices}

\end{subfigure}
\begin{subfigure}{0.5\linewidth}
\centering
\begin{tikzpicture}
\node[b1, label=above:{$d_\g^+ A^{[1]}_\res$}] (bulk1) at (1,0) {};\node[coordinate] at (.5,0) {}edge[fermion] (bulk1);
\end{tikzpicture}
\hspace{1cm}
\begin{tikzpicture}
\node[coordinate] (bulk1) at (1.5,0) {};
\node[b2, label=above:{$d_\g^- A^{[1]}_\res$}]  at (1,0) {}edge[fermion] (bulk1);
\end{tikzpicture}
\caption{Bulk vertices}\label{fig:bulkvertices}
\end{subfigure}
\caption{Vertices in 1D AKSZ theories with linear polarizations}
\end{figure}  

The term $S_0$ generates the propagators
\begin{align*}
\eta_{\ph}(t,t')=\frac{i}{\hbar} 
\langle\psi^+_\fl(t) \psi^-_\fl(t')\rangle &= (\theta(t'-t))\cdot \omega^{-1} \\
\eta_{\gh}(t,t')=\frac{i}{\hbar} 
\langle\psi_\fl^{[>0]}(t) \psi_\fl^{[<0]}(t')\rangle &= (\theta(t-t')-t)\cdot\omega^{-1}
\end{align*}
with $\theta(t)$ the Heaviside function and $\omega^{-1}$ the inverse of the pairing on $\g$. There are three types of connected Feynman diagrams contributing to $S_\eff$ (see Figure \ref{fig:FeynmanGraphsLin}): 
\begin{enumerate}
\item A single edge connecting the two boundary vertices (Figure \ref{fig:FeynmanGraphsLinA}). That diagram evaluates to $(\psi_\ii^-,\psi_\oo^+)$.
\item A single edge connecting a boundary vertex to a bulk vertex (Figure \ref{fig:FeynmanGraphsLinB})
.  Those diagrams yield $(\psi^-_\ii,d_\g^+ A^{[1]}_\res) + (\psi^+_\oo,d_\g^-  A^{[1]}_\res).$
\item A single edge connecting the two bulk vertices (Figure \ref{fig:FeynmanGraphsLinC})
. This diagram gives, using $\int_{I \times I}\eta_{\ph}(t,t') = \frac12$, the contribution $\frac12 (d_\g^- A^{[1]}_\res,d_\g^+ A^{[1]}_\res)$.
\end{enumerate} 
\begin{figure}[h]
\begin{subfigure}[b]{0.4\linewidth}
\centering
\begin{tikzpicture}[scale=1]
\draw (0,0.5) -- (0,-.5);
\node[bdry,label=left:{$\psi^-_\ii$}] (bdry1) at (0,0) {};
\draw (2,0.5) -- (2,-0.5);
\node[bdry,label=right:{$\psi^+_\oo$}] (bdry2) at (2,0) {};
\draw[fermion] (bdry1) -- (bdry2);
\end{tikzpicture}
\caption{Single edge connecting two boundaries}\label{fig:FeynmanGraphsLinA}

\end{subfigure}
\begin{subfigure}[b]{0.55\linewidth}
\centering
\begin{tikzpicture}[scale=1]
\draw (0,0.5) -- (0,-.5);
\node[bdry,label=left:{$\psi^-_\ii$}] (bdry1) at (0,0) {};
\draw (4,0.5) -- (4,-0.5);
\node[bdry,label=right:{$\psi^+_\oo$}] (bdry2) at (4,0) {};
\node[b1, label=above:{$d_\g^+ A^{[1]}_\res$}] (bulk1) at (1,0) {};
\node[b2,label=above:{$d_\g^- A^{[1]}_\res$}](bulk2) at (3,0) {};
\draw[fermion] (bdry1) -- (bulk1);
\draw[fermion] (bulk2) -- (bdry2);
\end{tikzpicture}
\caption{Single edge connecting bulk and boundary}\label{fig:FeynmanGraphsLinB}
\end{subfigure}
\begin{subfigure}{0.5\linewidth}
\begin{tikzpicture}[scale=1]
\draw (0,0.5) -- (0,-.5);
\draw (4,0.5) -- (4,-0.5);
\node[b1, label=below:{$d_\g^+ A^{[1]}_\res$}] (bulk1) at (3,0) {};
\node[b2,label=above:{$d_\g^- A^{[1]}_\res$}](bulk2) at (1,0) {};
\draw[fermion] (bulk2) -- (bulk1);
\end{tikzpicture}
\caption{Single edge connecting two bulk vertices}\label{fig:FeynmanGraphsLinC}
\end{subfigure}
\caption{Connected Feynman diagrams in effective action }\label{fig:FeynmanGraphsLin}
\end{figure}
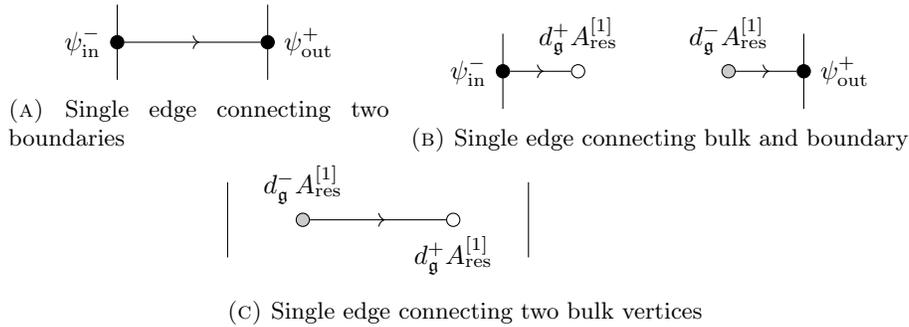
In total we obtain the effective action 
\begin{multline}
S_{\eff}[\psi_\ii,\psi_\oo,\psi_\res,A_\res] =  (\psi^{[>0]}_\oo-\psi^{[>0]}_\ii,\psi_\res) - (\psi^-_\ii,\psi^+_\oo) \\
+ (\psi^-_\ii,d_\g^+ A^{[1]}_\res) + (\psi_\oo^+,d_\g^-A^{[1]}_\res) + \frac12 (d_\g^- A^{[1]}_\res,d_\g^+ A^{[1]}_\res) + (A^{[>1]}_\res,d_\g \psi_{\res}). \label{eq: S eff 1d CS par}
\end{multline}
Separating the term depending only on ghost number 0 fields from the rest, we obtain the proof. 
\end{proof}
\begin{proposition}\label{prop: 1d par ghosts HJ}
The lagrangian generated by the $\gh=0$ part of the action is the  evolution relation in $X^{[0]}_\ii \times X^{[0]}_\oo$. 
\end{proposition}
\begin{proof}
The Euler-Lagrange equations of the theory in ghost number 0 are
\begin{align*}
d_I\psi^{[0]} &= -d_\g A^{[1]},  \\
d_\g\psi^{[0]} &= 0. 
\end{align*}
Projecting to boundary values $(\psi_\ii,\psi_\oo)$ we obtain 
the equations 
\begin{align}
d^-_\g\psi^+_\ii + d_\g^+\psi^-_\ii &= d^+_\g\psi^-_\oo+ d_\g^-\psi^+_\oo = 0, \label{eq: proof HJ par gh 1}\\
\psi^+_\oo - \psi^+_\ii &= d^+_\g a,  \label{eq: proof HJ par gh 2}\\
\psi^-_\oo - \psi^-_\ii &= d^-_\g b,  \label{eq: proof HJ par gh 3}
\end{align}
for some $a,b \in X^{[1]}$, and the first equation forces $a=b$ (up to a $d_\g^+$ and $d_\g^-$-closed term).
On the other hand, the lagrangian generated by $S_\ph$ is given by 
\begin{align*}
\psi^+_\ii &= -\frac{\dd S_\ph}{\dd \psi^-_\ii} = \psi^+_\oo - d_\g^+A^{[1]}_\res, \\
\psi^-_\oo &= \frac{\dd S_\ph}{\dd \psi^+_\oo} = \psi^-_\ii + d_\g^- A^{[1]}_\res,  \\
0 &= \frac{\dd S_\ph}{\dd A^{[1]}_\res} =d_\g^+\psi_\ii^- + d_\g^-\psi^+_\oo + d_\g^+d_\g^- A^{[1]}_\res.
\end{align*}
The first two equations  are equivalent to equations  \eqref{eq: proof HJ par gh 2},\eqref{eq: proof HJ par gh 3}, while the last equation enforces the constraint \eqref{eq: proof HJ par gh 1}.
\end{proof}

\subsubsection{Quantum master equation}
The modified 
quantum master equation 
$$(-\hbar^2\Delta_\res + \Omega)e^{\frac{i}{\hbar}S_\eff}$$
is equivalent to 
\begin{equation}
\left(\frac12\left\lbrace S_\eff,S_\eff \right\rbrace_\res - i\hbar\Delta_\res S_\eff+\Omega\right)Z =0 
\label{eq: mQME parallel ghosts}.
\end{equation}
Here we denote by $\{\cdot,\cdot\}_\res$ the BV ($+1$-shifted Poisson) bracket on $\VV$.
\begin{proposition}\label{prop: 1d par ghosts mQME}
The effective action $S_\eff$ given by \eqref{eq: S eff 1d CS par} satisfies the mQME \eqref{eq: mQME parallel ghosts} with boundary BFV operator $\Omega$ given by the standard quantization of $$\Theta(\psi) = \frac12(\psi,d_\g\psi) = (\psi^+,d_\g^-\psi^{[1]}) + (\psi^-,d_\g^+\psi^{[1]}) + (\psi^{[<0]},d_\g\psi^{[>1]}).$$ 
\end{proposition}
\begin{proof}
Expanding degree-wise as a differential operator, we obtain $\Omega = \Omega^{(1)} + \Omega^{(0)} =  \Omega^{(1)}_\oo +  \Omega^{(0)}_\oo +  \Omega^{(1)}_\ii + \Omega^{(0)}_\ii$ with 
\begin{align*}
\Omega^{(0)}_\oo &= \left(d_\g^-\psi^{[1]}_\oo,\psi^+_\oo\right), \\
\Omega^{(0)}_\ii &= -\left(d_\g^+\psi^{[1]}_\ii,\psi^-_\ii\right), \\
\frac{i}{\hbar}\Omega^{(1)}_\oo &=\left(d_\g^+\psi^{[1]}_\oo,\frac{\delta}{\delta \psi_\oo^+}\right) + \left(d_\g\psi_\oo^{[>1]},\frac{\delta}{\delta \psi_\oo^{[>0]}}\right),\\
\frac{i}{\hbar}\Omega^{(1)}_\ii &= \left(d_\g^-\psi^{[1]}_\ii,\frac{\delta}{\delta \psi_\ii^-}\right) + \left(d_\g\psi^{[>1]}_\ii,\frac{\delta}{\delta \psi^{[>0]}_\ii}\right).
\end{align*}
First of all, notice that $\Delta_\VV S_\eff = 0$ since in the only possibly nonvanishing term $(A^{[1]}_\res,d_\g\psi_\res)$ fields are not paired with their antifields because of the degree shift by the differential. Computing the BV bracket we obtain 
\begin{multline}
\frac12 \{ S_\eff,S_\eff \}_\res =\\
=- (\psi^{[1]}_\oo -\psi^{[1]}_\ii,d_\g^-\psi^+_\oo + d_\g^+\psi^-_\ii + d_\g^+d_\g^- A^{[1]}_\res) - (\psi^{[>1]}_\oo - \psi^{[>1]}_\ii,d_\g\psi_\res)\label{eq: 1d parallel ghosts mQME}
\end{multline}
(only terms of opposite ghost number survive in the pairing). 
On the other hand, since $\Omega^{(0)}$ is a multiplication operator and $\Omega^{(1)}$ contains only derivatives of first order, we have $\Omega Z = \Omega^{(0)} Z +  \frac{i}{\hbar}\Omega^{(1)} (S_\eff) Z$ and 
\begin{align*}
\frac{i}{\hbar}\Omega^{(1)}_\oo S_\eff &=   (\psi^-_\ii,d_\g^+\psi^{[1]}_\oo) + (d_\g^+\psi^{[1]}_\oo,d_\g^-A^{[1]}_\res)+ (d_\g\psi^{[>1]}_\oo, \psi_\res), \\
\frac{i}{\hbar}\Omega^{(1)}_\ii S_\eff &=  -(d_\g^-\psi^{[1]}_\ii,\psi^+_\oo) + (d_\g^-\psi^{[1]}_\ii,d_\g^+A^{[1]}_\res) - (d_\g\psi^{[>1]}_\ii, \psi_\res).
\end{align*}
 A straightforward computation  shows that $\Omega^{(0)} +\frac{i}{\hbar}\Omega^{(1)} S_\eff $ coincides with \eqref{eq: 1d parallel ghosts mQME}, thus completing the proof. 
\end{proof}

\subsection{General polarizations}\label{sec: 1d par ghosts nonlinear}
Next, we will consider the case where $X^{[0]}_\CC$ is equipped with another polarization $\mathcal{P}_\mr{nl}$ which is not necessarily linear (see \cite[Section 12]{HJ} for the corresponding toy model).  Let the base $\BB^{\PP_\mr{nl}}$ be locally parametrized by a coordinate $\psi^Q$, and the fibers by a coordinate $\psi^P$. Suppose $G(\psi^-,\psi^Q)$ is a generating function of the canonical
 transformation\footnote{I.e., the function $G(\psi^-,\psi^Q)$ satisfying $\delta G = \psi^+ \delta \psi^- - \psi^P \delta \psi^Q$.}
 $(\psi^-,\psi^+) \to (\psi^Q,\psi^P)$. Then we have that $\psi^+ = F(\psi^-,\psi^Q) = \frac{\dd G}{\dd \psi^-}$. We assume that $G$ is analytic in $\psi^-$ in a neighborhood $U$ of $\{0\} \times \BB^{\PP_\nl}$. 

We now consider again the 1D AKSZ theory on the interval, where we choose the polarizations at the two endpoints to be parallel in the ghost sector. In the physical sector we choose the $\psi^-$-representation on the in-boundary and the $\psi^Q$-representation on the out-boundary:
\begin{align*}\PP &= \PP_\ii \times \PP_\oo \quad \mbox{with} \\ 
\PP_\ii &= \left\lbrace\frac{\delta}{\delta \psi^{[<0]}_\ii},\frac{\delta}{\delta \psi_\ii^{+}}\right\rbrace, \\
\PP_\oo &= \left\lbrace\frac{\delta}{\delta \psi^{[<0]}_\oo},\frac{\delta}{\delta \psi_\oo^{P}}\right\rbrace.
\end{align*}
The base is 
$$\BB = \BB_\ii \times \BB_\oo,\;\; \BB_\ii = X^{[ > 0]}_\CC \oplus X^-,\;\; \BB_\oo = X^{ [> 0]}_\CC \times \BB^{\PP_{\mr{nl}}}.$$
The polarized 1-form is 
$\alpha^\pol_{\dd I} =\alpha^{\dd,\PP_\oo} + \alpha^{\dd,\PP_\ii}$ with 
\begin{align*}
\alpha^{\dd,\PP_\ii} &= -(\psi^{[<0]}_\ii,\delta \psi^{[>0]}_\ii) - (\psi_\ii^{+},\delta \psi_\ii^-) = \alpha^{\dd}_\ii -
 \delta f_\ii ,\\
\alpha^{\dd,\PP_\oo} & =(\psi^{[<0]}_\oo,\delta \psi^{[>0]}_\oo) + (\psi_\oo^P,\delta \psi_\oo^{Q}) = \alpha^{\dd}_\oo + \delta f_\oo ,
\end{align*}
where \begin{align*} 
f_\ii &=\frac12( \psi^{[>0]}_\ii,\psi^{[<0]}_\ii) - \frac12(\psi^
+_\ii,\psi^-_\ii), \\
 f_\oo &= \frac12 (\psi^{[>0]}_\oo,\psi^{[<0]}_\oo) - \frac12(\psi_\oo^+,\psi_\oo^-) - G(\psi^-_\oo,\psi^Q_\oo).
 \end{align*}
Thus, 
$\alpha^\pol_{\dd I} = \alpha_{\dd I} + \delta f$
 with 
\begin{multline*}
f (\psi_\oo,\psi_\ii) = f_\oo(\psi_\oo) - f_\ii(\psi_\ii) \\
= \frac12 (\psi^{[>0]}_\oo,\psi^{[<0]}_\oo) - \frac12(\psi_\oo^+,\psi_\oo^-) - G(\psi^-_\oo,\psi^Q_\oo) - \frac12 (\psi^{[>0]}_\ii,\psi^{[<0]}_\ii) + \frac12(\psi_\ii^+,\psi_\ii^-).
\end{multline*}
\subsubsection{Splitting the fields}
The goal is to find again a symplectomorphism $\Phi\colon\BB\times \VV \times \YY' \to \FF$.\footnote{Strictly speaking, the range of $\Phi$ is not $\FF$, but rather a certain regularization of $\FF$ more suitable for quantization. See the discussion in \cite[Section 9.2.3]{HJ}.} Here the trick is that we keep the space of fluctuations $\YY$ as above in Equation \eqref{eq: fluc 1d par gh}. In ghost number zero, the map $\Phi$ is defined as follows. For boundary values $\psi^-_\ii \in \BB_\ii^{[0]}$ and $\psi^Q_\oo \in \BB_\oo^{[0]}$ and fluctuations $\psi^-_\fl, \psi^+_\fl \in \YY$ (recall that $\psi^-_\fl(0) = \psi^+_\fl(1)= 0$), we let 
\begin{equation*}
\psi^-(t) = \begin{cases} \psi^-_\fl(t) &\mr{for} \;\; t > 0, \\ \psi^-_\ii &\mr{for} \;\; t = 0, \end{cases}
\end{equation*}
and 
\begin{equation*}
\psi^+(t) = \begin{cases} \psi^+_\fl(t) &\mr{for} \;\; t < 1, \\ F(\psi^-_\fl(1),\psi^Q_\oo) & \mr{for}\;\; t = 1. \end{cases} 
\end{equation*}
The map $\Phi$ is given by
\begin{multline}\label{Phi symplectomorphism}
\Phi(\psi^{[>0]}_\ii,\psi^-_\ii,\psi^{[>0]}_\oo,\psi^{Q}_\oo,\psi_\res,\psi_\fl,A) = \psi^-(\psi^-_\ii,\psi^-_\fl) + \psi^+(\psi^+_\fl,\psi^-_\fl, \psi^Q_\oo) \\ 
+{\psi_\res}+ \widetilde{\psi^{[>0]}_\ii} + \widetilde{\psi^{[>0]}_\oo} + \psi_\fl^{[\neq 0]} + A.
\end{multline}
In nonzero ghost number, this coincides with the splitting considered in the previous section. In what follows, we will discuss only the physical sector, i.e., the part in ghost number 0. The analysis in the ghost sector proceeds exactly as in Section \ref{sec: 1d par ghosts complex} and results in the ghost effective action $S_\gh$ described in Proposition \ref{prop: 1d par ghosts Seff}.
\subsubsection{Effective action}
Again, we can use the gauge-fixing lagrangian $\LL \subset \YY'$ given by zero forms. Restricted to $\BB \times \VV \times \LL$ and fields of ghost number 0, we have 
\begin{multline*}
S^f[\psi^-_\ii,\psi^Q_\oo,\psi^-_\fl,\psi^+_\fl,dt \cdot A_\res] = -(\psi^-_\ii,\psi^+_\fl(0)) - G(\psi^-_\fl(1),\psi^Q_\oo) \\+ \int_I (\psi^+_\fl,d_I \psi^-_\fl) - dt(d^+_\g A^{[1]}_\res,\psi^-_\fl) - dt(d^-_\g A^{[1]}_\res,\psi^+_\fl)
\end{multline*}
where the computation is very similar to the one in the proof of Lemma \ref{lem: S GF par ghosts}. 
The BV-BFV effective action is defined by 
\begin{align*}
Z &= e^{\frac{i}{\hbar}S_\eff[\psi_\ii,\psi_\oo,\psi_\res,A_\res]} = \int\DD\psi_\fl\; e^{\frac{i}{\hbar}S^f[\psi_\ii,\psi_\oo,\psi_\res,\psi_\fl,A_\res]} \\
&=e^{\frac{i}{\hbar}S_{\mr{back}}}\int\DD\psi_\fl\; e^{\frac{i}{\hbar}(S_{\mr{source}}+S_{0}+S_{\mr{int}})}
\end{align*} 
where the integral is defined in terms of Feynman diagrams. 
\begin{proposition}
The effective action (in ghost number 0) is 
\begin{equation}
S_\eff^\ph[\psi_\ii^-,\psi_\oo^Q,A^{[1]}_\res]= -G(\psi^-_\ii + d_\g^- A^{[1]}_\res,\psi_\oo^Q)  - (d_\g^+A^{[1]}_\res,\psi_\ii^- + \frac12d_\g^-A^{[1]}_\res)\label{eq: Seff nonlin pol}
\end{equation}
\end{proposition}
\begin{proof}
In terms of Feynman diagrams, the source term creates a vertex of arbitrary incoming valence on the out-boundary decorated by derivatives of $G$, and a univalent (outgoing) vertex on the in-boundary decorated by 
$\psi^-_\ii$. 
The interaction term creates  univalent in- and outgoing bulk vertices decorated by $d_\g^\pm A^{[1]}_\res$ as in the proof of Proposition \ref{prop: 1d par ghosts Seff}. 
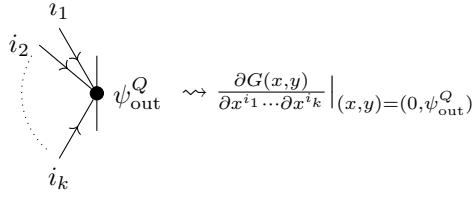
\begin{figure}
\centering
\begin{tikzpicture}
\begin{scope}[shift={(7,0)}]
\draw (0,.5) -- (0,-.5);
\node[bdry,label=right:{$\psi_\oo^Q$}] (bdry1) at (0,0) {};
\node[coordinate,label=above:{$i_1$}] at (120:1) {} edge[fermion] (bdry1);
\node[coordinate,label=left:{$i_2$}] at (140:1) {} edge[fermion] (bdry1);
\node[coordinate,label=below:{$i_k$}] at (240:1) {} edge[fermion] (bdry1);
\draw[dotted] (150:1) arc (150:230:1);
\node[coordinate, label=right:{$\leadsto \frac{\dd G(x,y)}{\dd x^{i_1} \cdots \dd x^{i_k}}\big|_{(x,y)= (0,\psi_\oo^Q)}$}] at (1,0){};
\end{scope}
\end{tikzpicture}

\caption{Additional vertex in 1D AKSZ theory with a general polarization on $X_\oo^{[0]}$.} \label{fig: vertex nonlin pol}
\end{figure}
The connected Feynman diagrams contributing to the effective action are: 
\begin{enumerate}
\item Diagrams involving the $G$-vertex on the out-boundary. The outgoing half-edges can connect either to the bulk vertex involving $\psi^+_\fl$ or the vertex on the in-boundary (see Figure \ref{fig:FeynmanDiagramsNonlin}). Summing over all valences, we obtain the Taylor series  in $x$ of $G(x,y)$ in the first argument at $(0,\psi_\oo^Q)$
evaluated  on $\psi^-_\ii + d_\g^-A^{[1]}_\res$. Hence, by analyticity of $G$ 
 those vertices sum up to
$$
- G(\psi_\ii^- + d_\g^-A^{[1]}_\res, \psi_\oo^Q).$$
\item Diagrams involving the univalent incoming bulk vertex. Here the outgoing half-edges connect to either the vertex on the in-boundary or 
an outgoing bulk vertex,
giving 
$$-(d_\g^+A^{[1]}_\res,\psi^-_\ii + \frac12 d_\g^-A^{[1]}_\res).$$
Those diagrams are the same as in the linear case (Figure \ref{fig:FeynmanGraphsLin}).
\end{enumerate}
\begin{figure}
\begin{tikzpicture}
\def\in{-1}
\def\out{2}
\draw (\in,2) -- (\in,-1);
\node[bdry,label=left:{$\gamma_\ii$}] (bdry1) at (\in,1.6) {};
\node[bdry,label=left:{$\gamma_\ii$}] (bdry2) at (\in,1.2) {}; 
\draw[dotted] (\in-0.1,1.2) -- (\in-0.1,0.6);
\node[bdry,label=left:{$\gamma_\ii$}] (bdry3) at (\in,0.5) {};
\draw (\out,2) -- (\out,-1);
\node[bdry, label=right:{$\psi_\oo^Q$ }] (bdry4) at (\out,0.5) {};
\node[b2] at (0,0.4) {} edge[fermion] (bdry4);
\node[b2] at (0,0.1) {} edge[fermion] (bdry4);
\draw[dotted] (0,0) -- (0,-0.5);
\node[b2] at (0,-0.5) {} edge[fermion] (bdry4);
\draw[fermion] (bdry1) .. controls (0.75,1) .. (bdry4);
\draw[fermion] (bdry2) .. controls (0.5,1) .. (bdry4);
\draw[fermion] (bdry3) .. controls (0,1) .. (bdry4);
\end{tikzpicture}
%
\caption{Additional Feynman diagrams in 1D AKSZ with general polarization on $X^{[0]}_\oo$.} \label{fig:FeynmanDiagramsNonlin}
\end{figure}
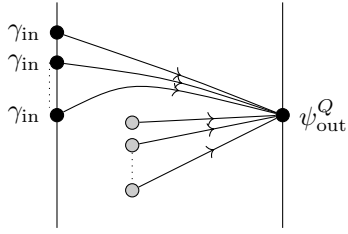
In total, we obtain the effective action (\ref{eq: Seff nonlin pol}).
\end{proof}
\begin{remark}
In the main case of interest for this paper, the target $\g = \Omega^\bullet(M)$ is infinite-dimensional and the propagator contains a delta form as the ``inverse'' of the pairing $(\tau_1,\tau_2) = \int_M \tau_1 \wedge \tau_2$ (cf. Remark \ref{rem: axial gauge}). 
However, our computations here are still valid. 
Indeed, even though Feynman diagrams contain products of delta functions, since all these diagrams are actually trees,  no problematic terms like $\delta(0)$ arise when computing the integrals. 
\end{remark}
\begin{proposition}
The lagrangian generated by 
(\ref{eq: Seff nonlin pol})
is the evolution relation in $X_\ii^{[0]} \times X_\oo^{[0]}$. 
\end{proposition}
\begin{proof} We know that in the $\psi^\pm$ variables, the evolution relation is given by $\psi^-_\oo = \psi^-_\ii + d^-_\g A^{[1]}_\res$, $\psi^+_\oo = \psi^+_\ii + d_\g^+A^{[1]}_\res$. The lagrangian generated by  $S_\eff^\ph[\psi_\ii^-,\psi_\oo^Q,A^{[1]}_\res]$ is 
\begin{align*}
-\psi^+_\ii = \frac{\dd S_\eff^\ph}{\dd \psi^-_\ii} &= -F(\psi_\ii^-+d_\g^- A^{[1]}_\res,\psi^Q_\oo)  + d_\g^+A^{[1]}_\res \\
& = -F(\psi^-_\oo,\psi^Q_\oo) + d_\g^+A^{[1]}_\res = -\psi^+_\oo+ d_\g^+A^{[1]}_\res, \\
\psi^P_\oo=\frac{\dd S_\eff^\ph}{\dd \psi^Q_\oo} &= -\frac{\dd G(\psi^-_\oo,\psi_\oo^Q)}{\dd \psi^Q_\oo},  \\ 
0 = \frac{\dd S_\ph}{\dd A^{[1]}_\res}& =d_\g^+\psi_\ii^- + d_\g^-\psi^+_\oo + d_\g^+d_\g^- A^{[1]}_\res.
\end{align*}
This lagrangian coincides with the evolution relation. 
\end{proof}
\subsubsection{Modified quantum master equation}\label{sec:mQME nonlin}
Let us also comment on the mQME. Again, we can compute the BV bracket (we ignore higher ghosts for simplicity)
$$\frac{1}{2}\{S_\eff,S_\eff\}_\res = \{S_\eff^\gh,S_\eff^\ph\} =   -(\psi^{[1]}_\oo -\psi^{[1]}_\ii,d_\g^-F(\psi^-_\ii + d_\g^-A^{[1]}_\res,\psi_\oo^Q) + d_\g^+\psi^-_\ii + d_\g^+d_\g^- A^{[1]}_\res).$$
As before, we have $\Omega_\ii = -(d_\g^+\psi_\ii^{[1]},\psi^-_\ii) - i\hbar (d_\g^-\psi_\ii^{[1]},\frac{\delta}{\delta \psi^-_\ii})$ and 
\begin{align*}Z^{-1}\Omega_\ii Z &= -(d_\g^+\psi_\ii^{[1]},\psi^-_\ii)Z - (d_\g^-\psi^{[1]}_\ii,F(\psi^-_\ii + d_\g^-A^{[1]}_\res,\psi_\oo^Q)) + (d_\g^-\psi^{[1]}_\ii,d_\g^+ A^{[1]}_\res) \\
&= - (\psi^{[1]}_\ii, d_\g^-F(\psi^-_\ii + d_\g^-A^{[1]}_\res,\psi_\oo^Q) + d_\g^+\psi^-_\ii + d_\g^+d_\g^- A^{[1]}_\res).
\end{align*}
Thus, the mQME is equivalent to 
\begin{equation}Z^{-1}\Omega_\oo Z = (\psi^{[1]}_\oo, d_\g^-F(\psi^-_\ii + d_\g^-A^{[1]}_\res,\psi_\oo^Q) + d_\g^+\psi^-_\ii + d_\g^+d_\g^- A^{[1]}_\res) = (\psi^{[1]}_\oo,d_\g\psi^{[0]}_\oo).\label{eq:mqme nonlin I}\end{equation}
The operator $\Omega_\oo$ acting on ghost number 0 fields should be obtained as a quantization of $\Theta(\psi)$ in the $\psi^P_\oo,\psi^Q_\oo$ variables, 
\begin{equation}
\Theta(\psi^P_\oo,\psi^Q_\oo) = (d_\g^+\psi^{[1]}_\oo,\psi^-(\psi^P_\oo,\psi^Q_\oo)) + (d_\g^-\psi^{[1]}_\oo,\psi^+(\psi^P_\oo,\psi^Q_\oo)).\label{eq:mqme nonlin II}
\end{equation}
The standard quantization $\Omega_\oo^{\mr{std}}$ of \eqref{eq:mqme nonlin II} --- i.e., replacing all $\psi^P_\oo$ variables with $-i\hbar \delta/\delta \psi^Q_\oo$ and moving all derivatives to the right - satisfies \eqref{eq:mqme nonlin I} to $0$-th order in $\hbar$, but there are terms of higher order in $\hbar$ corresponding to higher derivatives in $\psi^Q_\oo$ acting on $G$. To prove the mQME to all orders, one would have to find  quantum corrections to $\Omega^{\mr{std}}_\oo$ such that these terms are cancelled and the deformed operator still squares to 0. 
\begin{remark}\label{rem:mQMEnonlin}
A particularly simple case occurs when $\psi(\psi^P,\psi^Q) = \psi^P + \psi^Q$.  A rather trivial example of this case is $\psi^Q = \psi^+,\psi^P = \psi^-$. A nontrivial example will be considered in Section \ref{s:7dCSKS}.  In this case, we may define 
\begin{equation*}
\Omega_\oo = \left(d_\g\psi^{[1]}_\oo,\psi^Q_\oo - i\hbar\frac{\delta}{\delta \psi^Q_\oo}\right).
\end{equation*}
Then, from $\partial G/\partial \psi^Q = \psi^P$ we immediately get $Z^{-1}\Omega_\oo Z =  (d_\g\psi^{[1]}_\oo,\psi^Q_\oo + \psi^P_\oo) = (\psi^{[1]}_\oo,d_\g\psi^{[0]}_\oo)$, i.e., Equation \eqref{eq:mqme nonlin I}, and hence the mQME, are satisfied. 
In general, we have the mQME whenever the constraints are linear both in the original and the new momenta, see \cite[Section 12]{HJ}.
\end{remark}

\begin{remark}
Let $Z_{\nl,\perp}$ be the partition function with transversal ghost polarization and a general polarization in ghost number 0, to be precise, we are choosing the $(\psi^-,\psi^{[<0]})$  on the $\ii$-boundary and the $(\psi^Q,\psi^{[>0]})$ on the $\oo$-boundary. 
In this case, there are no residual fields, and following a computation similar to the above, one finds\footnote{There are no bulk vertices in this polarization. The two contributing terms come from multivalent $\oo$-boundary vertices in ghost number 0 and univalent boundary vertices in the ghost sector.} 
\begin{equation*}
Z_{\nl,\perp} = \exp\left(-\frac{i}{\hbar}G(\psi^-_\ii,\psi^Q_\oo) + \frac{i}{\hbar}\psi^{[<0]}_\ii\psi^{[>0]}_\oo\right).
\end{equation*}
The mQME for this partition function is just $(\Omega_\ii + \Omega_\oo)Z = 0$, since there are no residual fields. We can observe that the only obstruction for the mQME to hold  is the existence of a suitable $\Omega_\oo$ in a general polarization. 
Then, one can obtain the partition function $Z_{\nl,\parallel}$ -with parallel ghost polarization and $(\psi^-,\psi^Q)$-representation in ghost number 0, as given by \eqref{eq: Seff nonlin pol} by composition of the partition function $Z_{\mr{l},\parallel}$ with parallel ghost polarization and linear polarization in ghost number 0, given by \eqref{eq: S eff 1d CS par} with the partition function $Z_{\nl,\perp}$: $Z_{\nl,\parallel} = Z_{\nl,\perp} \circ Z_{\mr{l},\parallel}$. Since we know that $Z_{\mr{l},\parallel}$ satisfies the mQME, $Z_{\nl,\parallel}$ will satisfy it if $Z_{\nl,\perp}$ does. 
\end{remark}

\subsection{3D nonabelian Chern--Simons with parallel ghost polarization and antiholomorphic-to-holomorphic polarization in ghost degree zero 
} \label{ss: non-ab CS parall ghost}
Next, we return to the example of 3D Chern--Simons with parallel ghost polarization. In this context, it is convenient to use the   
traditional notation for the components of the superfield $\mathcal{A}$: 
$$ \mathcal{A} = c + A + A^* + c^* ,$$ 
where $\phi^*$ denotes the BV antifield of the field $\phi$. 

In this section we will use some special notations for field components (as compared to Section \ref{s: 3dCS}): ${a^{1,0}=\As_\fl^{1,0}}$, $a^{0,1}=\As_\fl^{0,1}$, $c=\As^0$, $\As^* = \As^2$, $\sigma=\As^0_\Ires$.
\subsubsection{Abelian case}\label{sec: 3d par ghost ab CS}
The action with polarization terms is:
$$ S^\pol = \int_{ I\times\Sigma} \frac12 \A d \A + \int_{\{1\}\times\Sigma } \frac12 \left(\As^{1,0} \As^{0,1}+c \As^* \right) - \int_{\{0\}\times\Sigma } \frac12 \left(\As^{0,1} \As^{1,0}+c \As^* \right) . $$
The space of fields is:
$$\FF=\Omega^\bt(I,\Omega^{1,0}\oplus \Omega^{0,1}\oplus \Omega^0[1] \oplus \Omega^2[-1])$$ 
--- here $\Omega^p$ in the coefficients stands for $\Omega^p(\Sigma)$.
It is fibered over $$\BB=(\Omega^{0,1}\oplus \Omega^0[1])\bigoplus (\Omega^{1,0}\oplus \Omega^0[-1])\;\; \ni ((\As^{0,1}_\ii,c_\ii),\; (\As^{1,0}_\oo,c_\oo))$$ with fiber
$$ \YY=\Omega^\bt(I,\{0\};\Omega^{0,1})\oplus \Omega^\bt(I,\{1\};\Omega^{1,0})\oplus \Omega^\bt(I,\{0,1\};\Omega^0[1])\oplus \Omega^\bt(I;\Omega^2[-1]) . $$
The space of residual fields is given by the (relative) cohomology in $I$-direction:
$$\VV=H^\bt(I,\{0,1\};\Omega^0[1])\oplus H^\bt(I;\Omega^2[-1])\quad \ni (dt\cdot\sigma, \As^*_\res) .$$
The gauge-fixing lagrangian $\LL$ in the fiber of $\YY\ra \VV$ is given by setting to zero the (relatively) exact 1-form components of fields along $I$.

Thus, on $\LL$ we have 
\begin{align*}
\gh=0:\quad & A^{(1)}=\til{\As}^{1,0}_\oo+\til{\As}^{0,1}_\ii+a^{1,0}+a^{0,1}+dt\cdot \sigma, \\
\gh=1:\quad & A^{(0)}=\til{c}_\oo+\til{c}_\ii+c_\fl ,\\
\gh=-1: \quad & A^{(2)}=\As^*_\res+\As^*_\fl ,\\
\gh=-2:\quad & A^{(3)}=0 ,
\end{align*}
with tilde denoting the discontinuous extension by zero from $t=1$ or $t=0$, respectively. Fluctuations are understood to satisfy
$$ a^{1,0}|_{t=1}=0,\quad a^{0,1}|_{t=0}=0 , \quad c_\fl|_{t=0}=c_\fl|_{t=1}=0,\quad \int_0^1 dt\; \As^*_\fl =0 .$$
The gauge-fixed polarized action is:
%
\begin{multline*}
S^\pol|_\LL=\\
\int_{ I\times\Sigma} a^{1,0} \dt a^{0,1}  
+ \int_{ I\times\Sigma} dt\, (a^{1,0}+a^{0,1}) d_\Sigma \sigma+
\int_\Sigma \As^{1,0}_\oo a^{0,1}\big|_{t=1}  - \int_\Sigma \As^{0,1}_\ii a^{1,0}\big|_{t=0}\\
+\int_{ I\times\Sigma} \As^*_\fl \dt c_\fl
-\int_\Sigma (\As^*_\res+\As^*_\fl\big|_{t=1}) c_\oo +
\int_\Sigma (\As^*_\res+\As^*_\fl\big|_{t=0}) c_\ii .
\end{multline*}

The propagators are given by:
\begin{align}
\langle a^{0,1}(t,z) a^{1,0}(t',z') \rangle  &=-i\hbar\, \theta(t-t')\, \delta^{(2)}(z-z')\frac{i}{2}d\bar{z}\,dz' ,  \label{non-ab phys prop} \\
\langle c_\fl(t,z) \As^*_\fl(t',z') \rangle  &=-i\hbar\, (\theta(t-t')-t)\, \delta^{(2)}(z-z')\frac{i}{2}dz'\,d\bar{z}' . \label{non-ab ghost prop}
\end{align}

The corresponding effective action is:
\begin{equation}
S^\mr{eff}=\int_\Sigma \As^{1,0}_\oo \As^{0,1}_\ii +
\As^{1,0}_\oo \bar\dd\sigma + \As^{0,1}_\ii \dd\sigma-\frac12 \dd\sigma\bar\dd\sigma - \As^*_\res (c_\oo-c_\ii).\label{eq: Seff par ab 3d CS}
\end{equation}
\begin{remark}[Hamilton--Jacobi property, mQME]
Notice that \eqref{eq: Seff par ab 3d CS} coincides with \eqref{eq: S eff 1d CS par} above upon specializing $\g = \Omega^\bullet(\Sigma)$, $X^+ = \Omega^{1,0}(\Sigma)$, $X^-=\Omega^{0,1}(\Sigma)$. Thus \eqref{eq: Seff par ab 3d CS} satisfies the modified quantum master equation, and the $\gh=0$ part of \eqref{eq: Seff par ab 3d CS} generates the evolution relation of abelian Chern--Simons theory. 
\end{remark}
\begin{remark}[
Integrating out residual fields]\label{rem:pushforward ab 3d CS}
As in Section \ref{sec: full int ahol-hol}, we can integrate out the residual fields $\sigma,\As_\res^*$ by choosing a Riemannian metric compatible with the complex structure and decomposing fields as $\sigma = \sigma_c + \underline{\sigma}, \As^*_\res = \mu\cdot\As^*_{\res,c}+\underline{\As^*_\res}$. As expected, the result differs from \eqref{ab CS ahol-hol Z fully integrated} only in the ghost sector:
\begin{equation*}
Z_*
[c_\oo,c_\ii,\As^{1,0}_\oo,\As^{0,1}_\ii] = \delta(c_{\oo,c} - c_{\ii,c}) \big({\det}'_{\Omega^0(\Sigma)}\Delta_g\big)^{-\frac12}\cdot e^{\frac{i}{\hbar} \mathbb{I}(\As^{1,0}_\oo,\As^{0,1}_\ii) }.
\end{equation*} 
Here $\mathbb{I}$ is given by \eqref{eq:def I}. 
\end{remark}

\subsubsection{Nonabelian case}
In the nonabelian Chern--Simons theory with coefficients in a semisimple Lie algebra $\ggg$ (corresponding to a compact group\footnote{To simplify the notations, and to be able to write expressions like $g^{-1}\dd g$ below, 
we will assume that $G$ is a matrix group. Otherwise, we should use left/right translations in $G$.} $G$), the superfield is $\A\in \Omega^\bt( I\times\Sigma, \ggg[1])$ and all the splittings are as before, just with components understood as $\ggg$-valued forms, paired in the quadratic part of the action via the Killing form $\langle,\rangle$ on $\ggg$. The interaction term of the nonabelian theory,  when restricted to the gauge-fixing lagrangian, yields
$$ S_\mr{int}=\frac16 \int \langle \A,[\A,\A] \rangle = -\int_\Sigma\int_I dt \langle a^{1,0},\mr{ad}_\sigma a^{0,1} \rangle   -\int_\Sigma\int_I dt \langle c_\fl,\mr{ad}_\sigma (\As^*_\res+\As^*_\fl) \rangle .
$$
This adds two new bivalent vertices and a univalent vertex to the Feynman rules.

Let us introduce the following notations:
\begin{gather}
F_+(x)=\frac{x}{1-e^{-x}}=
\sum_{n\geq 0}(-1)^n\frac{B_{n}}{n!}x^n
,\quad 
F_-(x)=-\frac{x}{e^x-1}=
-\sum_{n\geq 0} \frac{B_n}{n!} x^n,\\  \label{j}
j(\sigma)=\sum_{n=2}^\infty \frac{B_n}{n\cdot n!}\tr_\ggg(\ad_\sigma)^n = \mr{tr}_\ggg \log \frac{\sinh \frac{\ad_\sigma}{2}}{\frac{\ad_\sigma}{2}},
\end{gather}
with $B_n$ the Bernoulli numbers, $B_0=1,\; B_1=-\frac12,\; B_2=\frac16,\; B_3=0,\; B_4=-\frac{1}{30},\ldots$

In the following lemma we assume that $\sigma$ is in a sufficiently small neighborhood of zero in $\ggg$, see Remark \ref{rem: B_0} below for details.

\begin{lemma}
 \label{lemma: S_eff non-ab via sigma}
The partition function of the nonabelian Chern--Simons theory on a cylinder is $Z=e^{\frac{i}{\hbar}S^\mr{eff}}$ with the following effective action:
\begin{multline}\label{S eff non-abelian}
S^\mr{eff}=\\ =\int_\Sigma \langle \As^{1,0}_\oo, e^{-\ad_\sigma}\circ\As^{0,1}_\ii \rangle+
\langle \As^{1,0}_\oo,  
\frac{1-e^{-\ad_\sigma}}{\ad_\sigma}\circ \bar\dd\sigma \rangle+
\langle \As^{0,1}_\ii, 
\frac{e^{\ad_\sigma}-1}{\ad_\sigma}
\circ\dd\sigma \rangle \\
-\langle \dd\sigma ,\frac{e^{-\ad_\sigma}+\ad_\sigma-1}{(\ad_\sigma)^2} \circ \bar\dd\sigma \rangle-
\langle \As^*_\res,F_+(\ad_\sigma)\circ c_\oo + F_-(\ad_\sigma)\circ c_\ii \rangle -i\hbar \mathbb{W}(\sigma).
\end{multline}

In (\ref{S eff non-abelian}), the $1$-loop correction $\mathbb{W}$ stands for the contribution of ``ghost wheels'' --- cycles of $n\geq 1$ ghost-antifield propagators (at the vertices, they interact with the residual field 
$\sigma$). These graphs are ill-defined in the chosen axial gauge; their formal evaluation yields the expression
\begin{equation}\label{W non-ab}
 \mathbb{W}(\sigma)= \sum_{n\geq 1} \frac{B_n}{n\cdot n!}\, \mr{tr}_{C^\infty(\Sigma,\ggg)} (\ad_\sigma)^n  = \mr{tr}_{C^\infty(\Sigma)}  j(\sigma)\cdot 
\end{equation}
This expression 
heuristically stands for the ``sum over points $z$ of $\Sigma$'' of $j(\sigma(z))$.

\end{lemma}

We refer the reader to  \cite[Section 11.3]{HJ} for a one-dimensional toy model of this statement.

\begin{proof}
One has the following classes of Feynman diagrams contributing to the effective action:
\begin{figure}[H]
\includegraphics[scale=0.7]{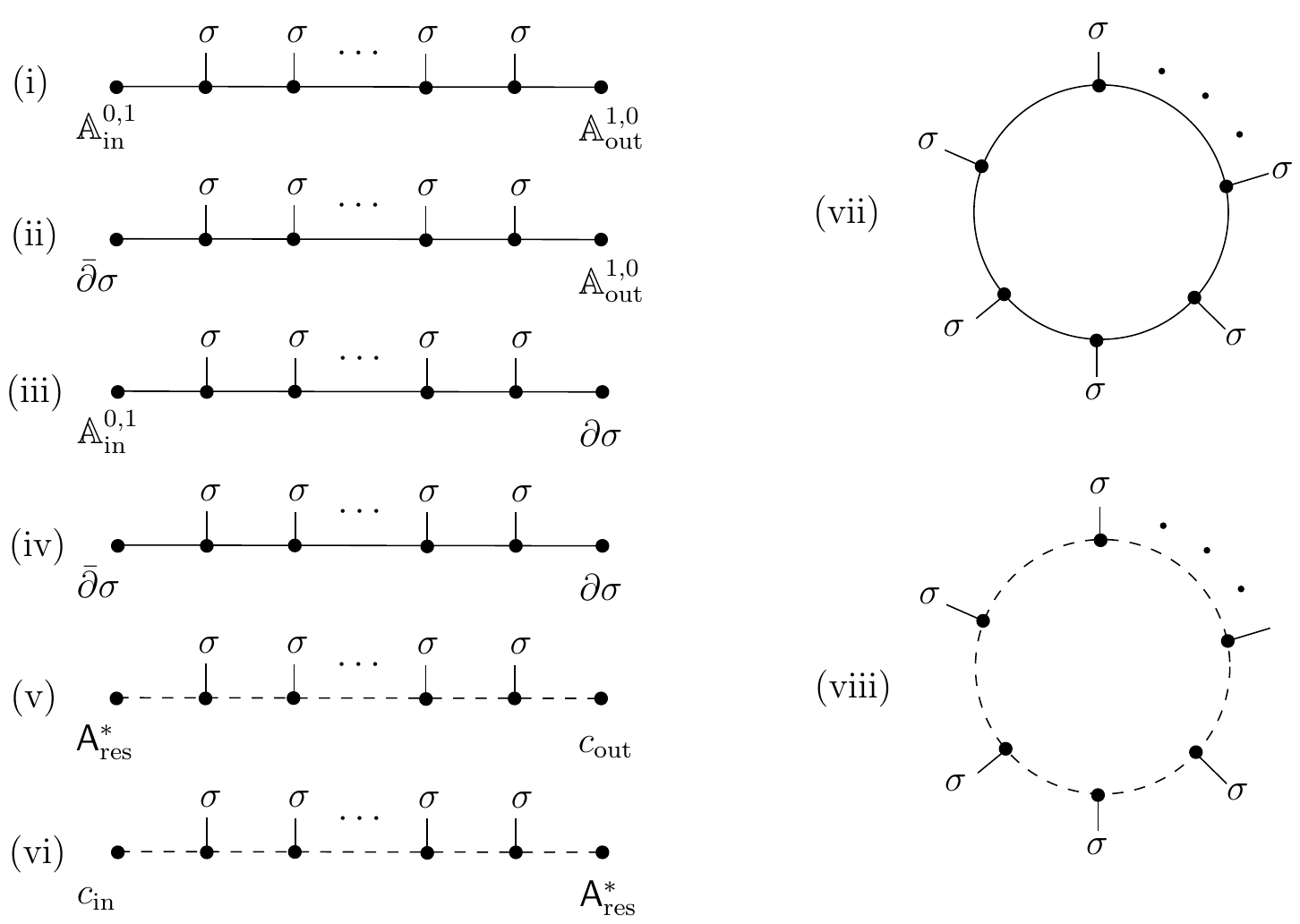}
\caption{Feynman diagrams in nonabelian theory on a cylinder with ``parallel ghost'' polarization.}
\label{S_eff non-ab diagrams}
\end{figure}
Here the solid lines represent the ``physical propagator'' (\ref{non-ab phys prop}) and the dashed lines represent the ``ghost propagator'' (\ref{non-ab ghost prop}). 

These diagrams are calculated easiest by introducing the propagators dressed with $\sigma$-insertions:
\begin{align*}
\!\!\!\!\!\!\!\!
\langle a^{0,1}(t,z)\otimes a^{1,0}(t',z') \rangle_\mr{dressed} & =-i\hbar\, \theta(t-t')\,  \delta^{(2)}(z-z')\frac{i}{2}d\bar{z}\,dz' \;\sum_{k=0}^\infty \int_{t'<t_1<\cdots<t_k<t}dt_1\cdots dt_k (-\ad_\sigma)^k\\
 &= -i\hbar\, \theta(t-t')\,e^{-(t-t')\ad_\sigma}\,  \delta^{(2)}(z-z')\frac{i}{2}d\bar{z}\,dz',
\end{align*}
\begin{align*}
\langle c_\fl(t,z) & \otimes  \As^*_\fl(t',z') \rangle_\mr{dressed}   =-i\hbar\,  \delta^{(2)}(z-z')\frac{i}{2}dz'\,d\bar{z}'\cdot\\
\cdot &\sum_{k=0}^\infty \underbrace{\int_{t_1,\ldots,t_k\in [0,1]} dt_1\cdots dt_k \; (\theta(t-t_1)-t)\,(\theta(t_1-t_2)-t_1)\cdots (\theta(t_k-t')-t_k)}_{
\begin{cases}\frac{B_{k+1}(1-t')-B_{k+1}(t-t')}{(k+1)!},\;\; t>t' \\
(-1)^k\frac{B_{k+1}(t'-t)-B_{k+1}(t')}{(k+1)!},\;\; t<t'
\end{cases}
}\; (-\ad_\sigma)^k\\
&= -i\hbar\,  \delta^{(2)}(z-z')\frac{i}{2}dz'\,d\bar{z}'\cdot 
\frac{e^{(t'-t+\theta(t-t'))\ad_\sigma}-e^{t'\ad_\sigma}}{e^{\ad_\sigma}-1}.
\end{align*}
Here $B_k(t)$ are the Bernoulli polynomials.

Computing the tree Feynman diagrams (i)--(vi) in Figure \ref{S_eff non-ab diagrams}, we have the following.
\begin{enumerate}[(i)]
\item $\displaystyle \wick{\int_\Sigma \langle\As^{1,0}_\oo,\c1{a^{0,1}}\big|_{t=1}\rangle \int_\Sigma\langle \c1{a^{1,0}}\big|_{t=0},\As^{0,1}_\ii \rangle} = \int_\Sigma \langle \As^{1,0}_\oo, e^{-\ad_\sigma}\circ \As^{0,1}_\ii \rangle$. Here the contraction is the dressed propagator.
\item \label{FeynDiag computation (ii)}  $\displaystyle  \wick{\int_\Sigma \langle \As^{1,0}_\oo, \c1{a^{0,1}}\big|_{t=1} \rangle  \int_{ I\times\Sigma} dt\,  \langle \c1{a^{1,0}},\bar\dd\sigma\rangle} = \int_\Sigma \langle \As^{1,0}_\oo, \int_0^1 dt\, e^{-(1-t)\ad_\sigma}\circ \bar\dd\sigma \rangle \\= \int_\Sigma\langle \As^{1,0}_\oo , \frac{1-e^{-\ad_\sigma}}{\ad_\sigma}\circ \bar\dd\sigma \rangle $.
\item Similarly to (\ref{FeynDiag computation (ii)}), $\displaystyle \wick{  -\int_{ I\times\Sigma} dt\, \langle \dd\sigma,\c1{a^{0,1}} \rangle  \int_\Sigma \langle \c1{a^{1,0}}\big|_{t=0},\As^{0,1}_\ii \rangle }= -\int_\Sigma \langle \dd\sigma,\int_0^1  dt\, e^{-t\ad_\sigma} \As^{0,1}_\ii \rangle \\= - \langle \dd\sigma,  \frac{1-e^{-\ad_\sigma}}{\ad_\sigma}\circ \As^{0,1}_\ii\rangle= \langle \frac{e^{\ad_\sigma}-1}{\ad_\sigma}\circ\As^{0,1}_\ii,\dd\sigma \rangle$.
\item $\displaystyle \wick{  -\int_{ I\times\Sigma} dt\, \langle \dd\sigma,\c1{a^{0,1}} \rangle   \int_{ I\times\Sigma} dt'\,  \langle \c1{a^{1,0}},\bar\dd\sigma\rangle   } = -\int_\Sigma \langle \dd\sigma  ,\int_0^1 dt \int_0^1 dt'\, e^{-(t-t')\ad_\sigma}\circ  \bar\dd\sigma \rangle \\
= - \int_\Sigma \langle \dd\sigma  , \frac{e^{-\ad_\sigma}+\ad_\sigma-1}{(\ad_\sigma)^2}\circ\bar\dd\sigma \rangle $.
\item \label{FeynDiag computation (v)}  $\displaystyle  \wick{ 
-\int_{ I\times\Sigma} dt\, \langle \ad_\sigma \As^*_\res,\c1{c_\fl} 
\rangle  \int_\Sigma \langle \c1{\As_\fl^*}\big|_{t=1}, c_\oo\rangle  
} -\int_\Sigma \langle A^*_\res,c_\oo \rangle \\
= 
-\int_\Sigma \langle \ad_\sigma \As^*_\res , \int_0^1 dt\, \frac{e^{(1-t)\ad_\sigma}-e^{\ad_\sigma}}{e^{\ad_\sigma}-1}\circ  c_\oo \rangle-\int_\Sigma \langle A^*_\res,c_\oo \rangle 
\\= -\int_\Sigma \langle \As_\res^*, 
\frac{\ad_\sigma}{1-e^{-\ad_\sigma}}
\circ c_\oo \rangle
$.
\item Similarly to (\ref{FeynDiag computation (v)}), $\displaystyle  \wick{ 
\int_{ I\times\Sigma} dt\, \langle \ad_\sigma \As^*_\res,\c1{c_\fl} 
\rangle  \int_\Sigma \langle \c1{\As_\fl^*}\big|_{t=0}, c_\ii\rangle  
} +\int_\Sigma \langle A^*_\res,c_\ii \rangle \\
=\int_\Sigma\langle \ad_\sigma \As^*_\res ,\int_0^1dt\, \frac{e^{(1-t)\ad_\sigma}-1}{e^{\ad_\sigma}-1} \circ c_\ii \rangle +\int_\Sigma \langle A^*_\res,c_\ii \rangle \\
=\int_\Sigma \langle \As^*_\res, \frac{\ad_\sigma}{e^{\ad_\sigma}-1}\circ c_\ii \rangle $.
\end{enumerate}
Thus, the Feynman diagrams (i)-(vi) in Figure \ref{S_eff non-ab diagrams} yield the $O(\hbar^0)$ part of the answer (\ref{S eff non-abelian}).

Next, consider the one-loop graphs in Figure \ref{S_eff non-ab diagrams}. 
The ``physical wheels'' --- diagrams (vii)  --- vanish due to the form of the propagator (\ref{non-ab phys prop}): they are proportional to
$$\int_{t_1,\ldots,t_k\in [0,1]} dt_1\cdots dt_k \, \theta(t_1-t_2)\theta(t_2-t_3)\cdots \theta(t_{k-1}-t_k) \theta(t_k-t_1) =0 . $$

Finally, consider the ``ghost wheels'' --- diagrams (viii). The propagator (\ref{non-ab ghost prop}) is the integral kernel of an operator $K^\mr{gh}=\mr{id}\otimes K^\mr{gh}_I$ acting on $ C^\infty(\Sigma)\otimes \Omega^\bt(I)$ with $K^\mr{gh}_I: f(t)+dt\, g(t) \mapsto \int_0^1dt'\,(\theta(t-t')-t)\,g(t')  $. As a regularization, let us replace   $C^\infty(\Sigma)$ 
with $C^\infty(X)$, with $X$  a  finite set of points --- the set of vertices of some triangulation of the surface $\Sigma$. In particular,  $C^\infty(X)$ is a finite-dimensional vector space.
Then, the regularized value of the ghost wheel diagram (viii) with $k$ $\sigma$-insertions is the supertrace:
$$ -i\hbar\; \mr{str}_{C^\infty(X)\otimes \Omega^\bt(I,\ggg)} (- K^\mr{gh}\, dt\, \ad_\sigma)^k  = -i\hbar\,\tr_{C^\infty(X)}\, \mr{str}_{\Omega^\bt(I,\ggg)} (- K^\mr{gh}_I\, dt\, \ad_\sigma)^k. $$
For the supertrace over the interval, we have (see, e.g., \cite{discrBF}):
\begin{multline*}
\mr{str}_{\Omega^\bt(I,\ggg)} (- K^\mr{gh}_I\,dt\, \ad_\sigma)^k\\
= \tr_\ggg\int_{t_1,\ldots,t_k\in [0,1]} (\theta(t_1-t_2)-t_1)dt_2 \ad_\sigma \cdots (\theta(t_{k-1}-t_k)-t_{k-1}) dt_k \ad_\sigma (\theta(t_k-t_1)-t_1) dt_1 \ad_\sigma
\\
= \frac{B_k}{ k!}\tr_\ggg(\ad_\sigma)^k .
\end{multline*}
Summing over the values of $k\geq 1$ and taking into account the symmetric factor $1/k$ (due to the automorphisms of the wheel graph), we obtain 
$$ \sum_{k\geq 1}\frac{1}{k}\mr{str}_{\Omega^\bt(I,\ggg)} (- K^\mr{gh}_I\, dt\, \ad_\sigma )^k = j(\sigma),  $$
with $j$ as in (\ref{j}). Thus, finally, the total contribution of graphs (viii) to the effective action is
$$ -i\hbar\,\mathbb{W}= -i\hbar\, \tr_{C^\infty(X)}j(\sigma)=-i\hbar\sum_{z\in X} j(\sigma(z)). $$
Trying to pass to a limit of dense triangulation $X$ obviously leads to an ill-defined result here.

Put another way, the regularized computation of a ghost wheel diagram is:
\begin{multline*}
-i\hbar\wick{
\sum_{z_1\in X}\int_I dt_1 \langle \c4{\As^*_\fl},-\ad_\sigma \c1{c_\fl} \rangle 
\sum_{z_2\in X}\int_I dt_2 \langle \c1{\As^*_\fl},-\ad_\sigma \c2{c_\fl} \rangle 
\;\;\c2{}\cdots\c3{}\;\;
\sum_{z_k\in X}\int_I dt_k \langle \c3{\As^*_\fl},-\ad_\sigma \c4{c_\fl} \rangle 
}\\
=i\hbar\sum_{z\in X} \int_{I}dt_1\cdots \int_I dt_k (\theta(t_1-t_2)-t_1)\cdots (\theta(t_{k-1}-t_k)-t_{k-1}) (\theta(t_k-t_1)-t_k) \tr_\ggg (-\ad_\sigma)^k\\
= -i\hbar\,\frac{B_k}{k!} \sum_{z\in X} \tr_\ggg (\ad_{\sigma(z)})^k .
\end{multline*}
where the contractions are the nondressed propagators (\ref{non-ab ghost prop}) with the delta form in $z$ replaced with Kronecker symbol $\delta_{z z'}$. In this regularized setup we understand the fields $c_\fl,\As^*_\fl,\sigma$ as supported at the vertices of $X$; fields $c_\fl,\As^*_\fl$ also depend on $t\in [0,1]$.
\end{proof}

\begin{remark}  \label{rem: B_0}
In Lemma  \ref{lemma: S_eff non-ab via sigma} we assumed that the residual field 
$\sigma$ takes values in a sufficiently small neighborhood of zero in $\ggg$, so that the sums 
 of Feynman diagrams in Figure \ref{S_eff non-ab diagrams} converge.\footnote{Curiously, the issue of convergence arises only in diagrams (v), (vi), (viii) --- the diagrams involving ghosts.} In fact, they converge if{f} $\sigma$ is valued in $B_0\subset \ggg$ where $B_0$ is 
the connected component of the origin in 
$$
\left\{\sigma\in\ggg\,\big|\, {\det}_\ggg \frac{\sinh \frac{\ad_\sigma}{2}}{\frac{\ad_\sigma}{2}}\neq 0\right\}\qquad \subset\quad  \ggg . $$
In other words, $B_0$ is 
the subset of $\ggg$ where all eigenvalues of $\ad_\sigma$ lie in the open interval $(-2\pi i, 2\pi i)\subset i\mathbb{R}$.  Thus, we are assuming that $\sigma$ takes values in $B_0\subset \ggg$ (cf.\ the discussion of the Gribov region in the context of 2D Yang-Mills in \cite[Section 2.4.1]{2dYM}). Furthermore, note that the exponential map $\exp:\ggg\ra G$ is a diffeomorphism from $B_0$ onto its image $\exp(B_0)$. Moreover, $\exp(B_0)$ is an open dense subset of $G$.
\end{remark}

\subsubsection{Group-valued parametrization of the residual field}\label{sss: group-valued parametrization}
Let us parametrize the residual field $\sigma$ by a group-valued map $g=e^{-\sigma}:\Sigma\ra G$. 
\begin{lemma}\label{lemma: S_eff non-ab rewriting via g}
The effective action (\ref{S eff non-abelian}) can be rewritten as 
\begin{multline}\label{S eff non-abelian via g}
S^\mr{eff}=
\int_\Sigma \Big( \langle \As^{1,0}_\oo, g\, \As^{0,1}_\ii g^{-1} \rangle -
\langle \As^{1,0}_\oo,\bar\dd g\cdot g^{-1} \rangle 
- 
\langle  \As^{0,1}_\ii, g^{-1}\,  \dd g  \rangle 
\\
+ \langle \As^*_\res,  F_-(\ad_{\log g})\circ c_\oo+F_+(\ad_{\log g})\circ c_\ii \rangle \Big)
+\mr{WZW}(g) - i\hbar \mathbb{W}.
\end{multline}
Here 
\begin{equation}\label{WZW def}
\mr{WZW}(g)=- \frac12 \int_\Sigma  \langle \dd g\cdot g^{-1}, \bar\dd g\cdot  g^{-1} \rangle -\frac{1}{12}  \int_{I\times \Sigma}\langle d\til{g}\cdot \til{g}^{-1},[d\til{g}\cdot \til{g}^{-1},d\til{g}\cdot \til{g}^{-1}] \rangle .
\end{equation}
is the Wess--Zumino--Witten action, where $\til{g}=e^{(t-1)\sigma}$ is the extension of $g$ to a mapping $ I\times\Sigma\ra G$, 
interpolating between $\til{g}=g$ at $t=0$ and $\til{g}=1$ at $t=1$.\footnote{By the standard result on Wess--Zumino terms, $\mr{WZW}(g) \bmod 4\pi^2 \mathbb{Z}$ is independent of the choice of $\til{g}$ interpolating between $g$ on one end of the cylinder and $1$ on the other.}
\end{lemma}

\begin{remark} Under the convergence assumption that $\sigma$ is valued in $B_0$ (see Remark \ref{rem: B_0}), or equivalently that $g$ is valued in $\exp(B_0)$ --- a contractible open dense subset of $G$,  $\mr{WZW}(g)$ is a single-valued function of $g$, and hence $S^\mr{eff}$ is also a single-valued expression. If $g$ is allowed to roam the entire group $G$, $\mr{WZW}(g)$ (and thus $S^\mr{eff}$) becomes multi-valued, defined only $\bmod\, 4\pi^2\mathbb{Z}$. In the latter case, for $e^{\frac{i}{\hbar}S^\mr{eff}}$ to be a single-valued expression, one needs $\hbar=\frac{2\pi}{k}$ with $k\in \mathbb{Z}$ an integer level.   The fact that quantization of $\hbar$ is necessary in one case but not in the other can be traced to the fact that the Cartan $3$-form (whose pullback by $\til{g}$ is the integrand in the second term in the r.h.s. of (\ref{WZW def})) represents a nontrivial cohomology class on $G$ but is exact 
when restricted to $\exp(B_0)$.
\end{remark}

\begin{proof}[Proof of Lemma \ref{lemma: S_eff non-ab rewriting via g}]
First terms in (\ref{S eff non-abelian}) and (\ref{S eff non-abelian via g}) obviously match. 
We have 
\begin{equation*}
\begin{aligned}
g^{-1}\dd g = e^{\sigma}\int_0^1 d\tau\, e^{-\tau\sigma}(-\dd\sigma) e^{-(1-\tau)\sigma} &= \int_0^1 d\tau\, e^{\tau \ad_\sigma}(-\dd\sigma)\\ 
&=\frac{e^{\ad_\sigma}-1}{\ad_\sigma}(-\dd\sigma) ,\\
 \bar\dd g\cdot g^{-1} = \int_0^1 d\tau\, e^{-\tau\sigma}(-\bar\dd\sigma)e^{-(1-\tau)\sigma}e^\sigma &= 
\int_0^1 d\tau \, e^{-\tau\ad_\sigma}(-\bar\dd\sigma)  \\&= \frac{1-e^{-\ad_\sigma}}{\ad_\sigma}(-\bar\dd\sigma).
\end{aligned}
\end{equation*}
Thus, second and third terms in (\ref{S eff non-abelian}) and (\ref{S eff non-abelian via g}) also match. Next, evaluating the Wess-Zumino term on our preferred extension $\til{g}=e^{(t-1)\sigma}$, we have
\begin{multline}\label{WZ computation}
-\frac{1}{12}\int_{ I\times\Sigma} \langle d\til{g}\cdot \til{g}^{-1},[d\til{g}\cdot \til{g}^{-1},d\til{g}\cdot \til{g}^{-1}] \rangle \\ =
-\frac14 \int_\Sigma \int_0^1 d t \Big\langle \sigma,\int_0^{1-t}d\tau \int_0^{1-t}d\tau' \Big[ e^{-\tau\sigma}(-d\sigma) e^{-(1-t-\tau)\sigma}e^{(1-t)\sigma},e^{-\tau'\sigma}(-d\sigma) e^{-(1-t-\tau')\sigma}e^{(1-t)\sigma}\Big] \Big\rangle \\
=\frac14 \int_\Sigma \int_0^1 dt \int_0^{1-t}d\tau \int_0^{1-t} d\tau' \Big\langle d\sigma , \Big[\sigma, e^{(\tau'-\tau)\ad_{\sigma}} d\sigma \Big] \Big\rangle 
= \frac12 \int_\Sigma   \Big\langle d\sigma, \Big[\sigma, \Big(\frac{\sinh\ad_\sigma-\ad_\sigma}{(\ad_\sigma)^3}\Big) \, d\sigma \Big] \Big\rangle \\
=\frac12 \int_\Sigma   \Big\langle d\sigma, \Big(\frac{\sinh\ad_\sigma-\ad_\sigma}{(\ad_\sigma)^2}\Big) \, d\sigma  \Big\rangle.
\end{multline}
 The WZW kinetic term is:
\begin{multline}\label{WZW kinetic computation}
- \frac12 \int_\Sigma  \langle \dd g\cdot g^{-1}, \bar\dd g\cdot  g^{-1} \rangle \\
=
-\frac12 \int_\Sigma \int_0^1 d\tau \int_0^1 d\tau' \Big\langle  e^{-\tau\sigma} (-\dd\sigma) e^{-(1-\tau)\sigma}e^\sigma ,
 e^{-\tau'\sigma} (-\bar\dd\sigma) e^{-(1-\tau')\sigma}e^\sigma \Big\rangle\\
 = -\frac12\int_\Sigma \int_0^1 d\tau \int_0^1 d\tau' \Big\langle \dd\sigma, e^{(\tau-\tau')\ad_\sigma} \bar\dd\sigma \Big\rangle = 
 - \int_\Sigma\Big\langle \dd\sigma,  \frac{\cosh\ad_\sigma-1}{(\ad_\sigma)^2} \bar\dd\sigma \Big\rangle .
\end{multline}
Putting the kinetic term (\ref{WZW kinetic computation}) and the Wess-Zumino term (\ref{WZ computation}) together, we obtain
\begin{multline*}
\mr{WZW}(g)= \int_\Sigma -\Big\langle \dd\sigma,  \frac{\cosh\ad_\sigma-1}{(\ad_\sigma)^2} \bar\dd\sigma \Big\rangle  +  \Big\langle \dd\sigma, \frac{\sinh\ad_\sigma-\ad_\sigma}{(\ad_\sigma)^2} \, \bar\dd\sigma  \Big\rangle\\
=-\int_\Sigma \Big\langle  \dd\sigma , \frac{ e^{-\ad_\sigma} +\ad_\sigma-1}{(\ad_\sigma)^2}\, \bar\dd\sigma   \Big\rangle . 
\end{multline*}
Thus, finally, $\mr{WZW}$ term in (\ref{S eff non-abelian via g}) coincides with the fourth term in (\ref{S eff non-abelian}).

Ghost terms and the 1-loop contributions in (\ref{S eff non-abelian}) and (\ref{S eff non-abelian via g}) are identified directly.

\end{proof}

\subsubsection{A comment on ghost wheels}\label{sss: ghost wheels}
To understand the role of the term $\mathbb{W}$ (ghost wheels) in (\ref{S eff non-abelian via g}), recall that, for $\mu_G$ the Haar measure on the group $G$ and $\mu_\ggg$ the Lebesgue measure on the Lie algebra, one has $\exp^*\mu_G=e^j\cdot \mu_\ggg$, with $j$ the function on $\ggg$ defined by the formula (\ref{j}). Therefore, the half-density on the space of residual fields
 associated to the effective action (\ref{S eff non-abelian via g}) is, heuristically, the following:\footnote{
Recall, see \cite{Severa}, that on an $(n|n)$-dimensional odd symplectic supermanifold $(\mc{M},\omega)$, half-densities can be understood as cohomology classes of the differential $\omega\wedge $ acting on differential forms $\Omega^\bt(\mc{M})$. Moreover, if $(x^i,\xi_i)$ are Darboux coordinates, $i=1,\ldots,n$, with $x^i$ the even coordinates, then each cohomology class has a unique representative of the form $\rho(x,\xi) dx^1\cdots dx^n \in \Omega^n(\mc{M})$ corresponding to the half-density $\rho(x,\xi)\prod_i d^{\frac12} x D^{\frac12}\xi$, with $\rho$ some function. In (\ref{half-density manipulation}), ``$\sim$" refers this choice of preferred representative for a half-density.
}
\begin{multline}\label{half-density manipulation}
 e^{\frac{i}{\hbar}S^\mr{eff}}\D^{\frac12} \sigma\; \D^{\frac12} \As^*_\res \sim 
e^{\frac{i}{\hbar}S^\mr{eff}}\D \sigma \\
= e^{\frac{i}{\hbar}S^{\mr{eff}\,(0)}}\;\; ``\prod_{z\in\Sigma} \underbrace{e^{j(\sigma(z))}\mu_\ggg(\sigma(z))}_{\mu_G(g(z))} "= 
e^{\frac{i}{\hbar}S^{\mr{eff}\,(0)}}\; \D g .
\end{multline}
Here $S^{\mr{eff}\,(0)}$ stands for (\ref{S eff non-abelian via g}) without the $\mathbb{W}$ term\footnote{
The superscript $(0)$ means the ``$0$-loop part,'' corresponding to the expansion in powers of $\hbar$:  $S^\mr{eff}=\sum_{k=0}^\infty (-i\hbar)^k S^{\mr{eff}\,(k)}$. In the present case, we have only $k=0,1$ terms.
} --- the latter was used in transforming the functional measure from the  pointwise product of Lebesgue measures for $\sigma$ to the product of Haar measures for $g$. The equivalence (\ref{half-density manipulation}) of a half-densities is an extension of a rigorous result presented in \cite{HJ} for a finite-dimensional system.

The odd symplectic form on residual fields is
\begin{equation}\label{non-ab omega BV}
\begin{aligned}
 \omega_\res& =\int_\Sigma \langle \delta  \As^*_\res , \delta \sigma \rangle = \delta  \int_\Sigma \langle  \As^*_\res , \delta \sigma \rangle  \\ 
 & = \delta  \int_\Sigma \langle  \As^*_\res ,  -\frac{\ad_\sigma}{1-e^{-\ad_\sigma}}\circ(\delta g\cdot g^{-1}) \rangle = 
 \delta \int_\Sigma \langle   -\frac{\ad_\sigma}{e^{\ad_\sigma}-1}\circ\As^*_\res , \delta g\cdot g^{-1} \rangle\\
& =\int_\Sigma \langle \delta g^*, \delta g \rangle ,
\end{aligned}
\end{equation}
where we introduced the notation
\begin{equation}
g^*=-g^{-1}\cdot  \big(F_+(\ad_{\log g})\circ \As^*_\res \big)  = \big(F_-(\ad_{\log g})\circ \As^*_\res \big)\cdot g^{-1}
\end{equation}
--- a reparametrization of the residual field $\As^*_\res$ such that $(g,g^*)$ form Darboux coordinates on $\VV$.

Rewritten in terms of the parametrization $(g,g^*)$ for residual fields, the half-density (\ref{half-density manipulation}) becomes
\begin{equation}
e^{\frac{i}{\hbar}S^\mr{eff}} \D^{\frac12} \sigma\; \D^{\frac12} \As^*_\res  =
e^{\frac{i}{\hbar}S^{\mr{eff}\,(0)}} \D^{\frac12} g\; \D^{\frac12} g^* .
\end{equation}
I.e., in the $(g,g^*)$-parametrization, the ghost loops go away and the effective action has no quantum corrections.

\begin{remark}\label{rem: log-half-density shift}
In the context of BV formalism, it is natural to think of $S^\mr{eff}$ as a ``log-half-density'' (see, e.g., \cite[section 2.6]{discrBF}) on the space of residual fields, rather than a function, i.e., behaving under a change of Darboux coordinates as  
$$S^\mr{eff}_{[x,\xi]}(x,\xi)=S^\mr{eff}_{[x',\xi']}(x',\xi')-i\hbar \log \mr{sdet} \frac{\dd(x,\xi)}{\dd(x',\xi')} ,  $$ 
so that one has  $e^{\frac{i}{\hbar}S^\mr{eff}_{[x,\xi]}(x,\xi)}d^{\frac12}x D^{\frac12} \xi=e^{\frac{i}{\hbar}S^\mr{eff}_{[x',\xi']}(x',\xi')}d^{\frac12}x' D^{\frac12} \xi'$. Here the superdeterminant (Berezinian) $\mr{sdet}\cdots$ is the Jacobian of the transformation. With that in mind, the effective action (\ref{S eff non-abelian}) is $S^\mr{eff}_{[\sigma,\As^*_\res]}$ --- relative to the coordinate system $(\sigma,\As^*_\res)$ on $\VV$. On the other hand, $S^{\mr{eff}\,(0)}$ given by (\ref{S eff non-abelian via g}) without the $-i\hbar \mathbb{W}$ term is $S^\mr{eff}_{[g,g^*]}$ --- relative to the coordinate system $(g,g^*)$.
\end{remark}

\subsubsection{Effective action vs. Hamilton--Jacobi}\label{sss: Seff vs HJ nonab CS}

Denote
\begin{multline}\label{I non-ab}
\mathbb{I}(\As^{0,1}_\ii , \As^{1,0}_\oo;g)=\\
=\int_\Sigma \Big( \langle \As^{1,0}_\oo, g \As^{0,1}_\ii g^{-1} \rangle -
\langle \As^{1,0}_\oo,\bar\dd g\cdot g^{-1} \rangle 
- \langle  \As^{0,1}_\ii, g^{-1} 
 \dd g  \rangle \Big)
+\mr{WZW}(g) 
\end{multline}
--- the effective action (\ref{S eff non-abelian via g}) restricted to fields of ghost number zero and without the $O(\hbar)$ term.

Function (\ref{I non-ab}) produces, as a 
generalized generating function 
(see \cite[Appendix A]{HJ} and Section \ref{sec:HJ}),
with $g$ an auxiliary parameter, the following lagrangian $L\subset \overline{\PhiPM_\mr{in}}\times \PhiPM_\mr{out}$ in the phase space for the boundary of the cylinder:
\begin{equation}
\begin{aligned}\label{non-ab HJ L}
L
=\Big\{&\As|_{t=1}=\As_\oo^{1,0}+\underbrace{g \As^{0,1}_\ii g^{-1} -\bar\dd g\cdot g^{-1}}_{\As^{0,1}\big|_{t=1}=\frac{\delta \mathbb{I}}{\delta \As^{1,0}_\oo}}\;,\\
&\As|_{t=0}=\As_\ii^{0,1}+\underbrace{g^{-1} \As^{1,0}_\oo g + g^{-1}\cdot\dd g}_{\As^{1,0}\big|_{t=0}=-\frac{\delta \mathbb{I}}{\delta \As^{0,1}_\ii}} \;\; 
\Big|\;\; Y=0
 \Big\} ,
\end{aligned}
\end{equation}
where we denoted\footnote{
The notation $Y=-\frac{\delta\mathbb{I}}{\delta g\cdot g^{-1}}$ means that variation of $\mathbb{I}$ under the variation of $g$ is $\delta_g \mathbb{I}=-\int_\Sigma \langle \delta g\cdot g^{-1}, Y \rangle$.
} 
\begin{multline}\label{Y constraint}
Y=-\frac{\delta\mathbb{I}}{\delta g\cdot g^{-1}}= \\
= [\As_\oo^{1,0},g \As^{0,1}_\ii g^{-1} ]+(\bar\dd-\ad_{\bar\dd g\cdot g^{-1}}) \As_\oo^{1,0} +\dd (g \As^{0,1}_\ii g^{-1}) -\dd (\bar\dd g\cdot g^{-1}) .
\end{multline}

\begin{lemma}
The lagrangian (\ref{non-ab HJ L}) generated by the functional  (\ref{I non-ab}) --- the tree part of the effective action, restricted to $\gh=0$ fields --- coincides with the 
evolution relation
in $\overline{\PhiPM_\mr{in}}\times \PhiPM_\mr{out}$ 
for
Chern--Simons theory on the cylinder $ I\times\Sigma$.
\end{lemma}

\begin{proof}
We are restricting our attention only the to $\gh=0$ connection field $\As+dt\cdot a$ with $\As$ a $t$-dependent $1$-form on $\Sigma$ and $a$ a $t$-dependent $0$-form on $\Sigma$ (both are $\ggg$-valued). The equation of motion --- zero-curvature condition --- $F_{\As+dt\, a}=0$ splits into
\begin{align}
d_\Sigma \As + \frac12 [\As,\As] = 0,  \label{non-ab HJ eq1} \\
\dd_t \As 
= (d_\Sigma + [\As,-]) a . \label{non-ab HJ eq2}
\end{align}
Equation (\ref{non-ab HJ eq2}) says that $\As$ changes by a continuous gauge transformation on $\Sigma$ as $t$ changes, with $a$ the infinitesimal generator. Thus,
\begin{equation}\label{non-ab HJ fin gt}
\As|_{t=1}=g\,\As|_{t=0}\, g^{-1} + g d_\Sigma g^{-1},\quad \mr{with}\;\; g=P\overleftarrow{\exp} \left(-\int_0^1 dt\, a\right) .
\end{equation}
This implies that we can recover the $(1,0)$ component of $\As$ at $t=0$ from its known value at $t=1$ and can recover the $(0,1)$ component at $t=1$ from its known value at $t=0$. Thus,
\begin{align}
\As|_{t=1} &= \As^{1,0}_\oo+\As^{0,1}|_{t=1}=\As^{1,0}_\oo+ g \As^{0,1}_\ii g^{-1} + g \bar\dd g^{-1}  ,
\\
\As|_{t=0} &= \As^{0,1}_\ii+\As^{1,0}|_{t=0}=\As^{0,1}_\ii+ g^{-1} \As^{1,0}_\oo g + g^{-1} \dd g  .
\end{align}
Note that these two equations coincide with 
the first two equations in (\ref{non-ab HJ L}). Next, equation (\ref{non-ab HJ eq1}) means that the curvature of $\As$ must vanish on $\Sigma\times\{t\}$ for any $t$. In fact, it suffices to verify it just for one value of $t$, because for all others it would follow from (\ref{non-ab HJ eq2}). Checking (\ref{non-ab HJ eq1}) at $t=1$, we have
\begin{equation}
F_\As \big|_{t=1}=\underbrace{\bar\dd\As^{1,0}_\oo + \dd (g \As^{0,1}_\ii g^{-1}-\bar\dd g\cdot g^{-1})}_{d_\Sigma \As|_{t=1}}  + [\As^{1,0}_\oo, g \As^{0,1}_\ii g^{-1}-\bar\dd g\cdot g^{-1}] \quad =0 .
\end{equation}
This equation coincides with the constraint $Y=0$ in (\ref{non-ab HJ L}) coming from equating to zero the variation of the generating function $\mathbb{I}$ in the auxiliary parameter $g$.

Thus, we have checked that the lagrangian in the boundary phase space induced from the equations of motion (the evolution relation) coincides with the lagrangian generated by $\mathbb{I}$. 
\end{proof}

\begin{remark}
The function $\mathbb{I}$ given by (\ref{I non-ab}) is also the Hamilton--Jacobi action (see \cite[Section 7.2]{HJ}): 
it is the evaluation of the 
Chern--Simons 
action with polarization terms, restricted to degree-zero fields,
$$ 
S^{f}_\ph
=\int_{ I\times\Sigma} \Big(\frac12 \langle A, dA \rangle+\frac16 \langle A, [A,A] \rangle \Big) + \int_{\{1\}\times\Sigma } \frac12 \langle \As^{1,0},\As^{0,1} \rangle - \int_{\{0\}\times\Sigma } \frac12 \langle \As^{0,1},\As^{1,0} \rangle ,$$
on any connection $1$-form $A$ 
solving the ``evolution equation'' $\iota_{\frac{\dd}{\dd t}} F_A=0$
subject to boundary conditions $(A|_{t=1})^{1,0}=\As^{1,0}_\oo$, $(A|_{t=0})^{0,1}=\As^{0,1}_\ii$ and with the parallel transport of $A$ along the interval $I\times \{z\}$ given by $g(z)\in G$ for any $z\in \Sigma$. One proves this by an explicit computation similar to the proof of Lemma \ref{lemma: S_eff non-ab rewriting via g}, picking a convenient gauge equivalent representative for $A=\As+a\,dt$ with $a$ constant along $I$ (but allowed to vary in $\Sigma$ direction). Here we are using gauge-invariance of Chern--Simons action $\bmod \,4\pi^2\mathbb{Z}$ with respect to gauge transformations trivial on the boundary.
\end{remark}

\subsubsection{Quantum master equation}\label{sss: mQME nonab CS}
Quantum BFV operators on in- and out-states $\Omega_\ii,\Omega_\oo$ are given by canonical quantization of the boundary BFV action
$$ S^\mr{BFV}_\Sigma= \pm \int_\Sigma \langle c,F_\As \rangle+\langle \As^* ,\frac12[c,c]\rangle $$
with $\pm$ corresponding to out-/in-boundary. Explicitly, quantum BFV operators are\footnote{The BFV operator (\ref{non-ab Omega in}), its generalization to the case of Wilson lines intersecting the boundary --- see (\ref{Omega with Wilson lines}) below --- and its cohomology in genus zero were discussed in \cite{ABM}.}
\begin{align}
\label{non-ab Omega out}
\Omega_\oo &= \int_\Sigma \left\langle c_\oo,\bar\dd \As^{1,0}_\oo -i\hbar(\dd+[\As^{1,0}_\oo,-])\frac{\delta}{\delta \As^{1,0}_\oo} \right\rangle -i\hbar \left\langle \frac12 [c_\oo,c_\oo],\frac{\delta}{\delta c_\oo} \right\rangle, \\ \label{non-ab Omega in}
\Omega_\ii &= \int_\Sigma \left\langle c_\ii,-\dd \As^{0,1}_\ii -i\hbar(\bar\dd+[\As^{0,1}_\ii,-])\frac{\delta}{\delta \As^{0,1}_\ii} \right\rangle 
-i\hbar \left\langle \frac12 [c_\ii,c_\ii],\frac{\delta}{\delta c_\ii} \right\rangle .
\end{align}

\begin{lemma}\label{lemma: mQME non-ab}
The partition function $Z=e^{\frac{i}{\hbar}S^\mr{eff}}$ with $S^\mr{eff}$ given by (\ref{S eff non-abelian}), (\ref{S eff non-abelian via g}) satisfies the modified quantum master equation
\begin{equation}\label{mQME non-ab}
(\Omega_\oo+\Omega_\ii-\hbar^2\Delta_\res)Z=0
\end{equation}
with $\Delta_\res=\int_\Sigma \langle \frac{\delta}{\delta\sigma},\frac{\delta}{\delta \As^*_\res} \rangle$ the BV Laplacian on residual fields.
\end{lemma}

\begin{proof}
Given the ansatz $Z=e^{\frac{i}{\hbar}S^\mr{eff}}$, the equation (\ref{mQME non-ab}) can be written as
\begin{equation}
Z^{-1}\Omega_\ii Z + Z^{-1}\Omega_\oo Z +\frac12 \{S^\mr{eff},S^\mr{eff}\}_\res-i\hbar \Delta_\res S^\mr{eff} \stackrel{!}{=} 0
\end{equation}
with $\{,\}_\res$ the odd Poisson bracket on residual fields associated with the symplectic structure (\ref{non-ab omega BV}). Moreover, using the decomposition $S^\mr{eff}=S^{\mr{eff}\,(0)}-i\hbar\,\mathbb{W}(\sigma)$, the mQME can be further rewritten as
\begin{multline}\label{non-ab mQME expanded}
Z^{-1}\Omega_\ii Z + Z^{-1}\Omega_\oo Z +\frac12 \{S^{\mr{eff}\,(0)},S^{\mr{eff}\,(0)}\}_\res-\\
-i\hbar \Big( \{S^{\mr{eff}\,(0)},\mathbb{W}\}_\res+ \Delta_\res S^{\mr{eff}\,(0)}\Big)  \stackrel{!}{=} 0 .
\end{multline}
It is easiest to compute the term $\frac12 \{S^{\mr{eff}\,(0)},S^{\mr{eff}\,(0)}\}_\res$ using $(g,g^*)$ - parametrization of residual fields. We have 
\begin{eqnarray*}
S^{\mr{eff}\,(0)}&=& 
\mathbb{I}
+\int_\Sigma \langle g^*, c_\oo  g- g c_\ii \rangle , \\
\frac{\delta}{\delta g}S^{\mr{eff}\, (0)} 
&=& 
-g^{-1} Y+g^* c_\oo + c_\ii g^* ,
\\
\frac{\delta}{\delta g^*} S^{\mr{eff}\,(0)} &=& c_\oo  g-g c_\ii ,
\end{eqnarray*}
with $\mathbb{I}$ as in (\ref{I non-ab}) and $Y$ as in (\ref{Y constraint}). Thus,
\begin{multline*}
\frac12 \{S^{\mr{eff}\,(0)},S^{\mr{eff}\,(0)}\}_\res = S^{\mr{eff}\,(0)}\left( \int_\Sigma \langle \frac{\overleftarrow\delta}{\delta g}, \frac{\overrightarrow\delta}{\delta g^*}\rangle \right) S^{\mr{eff}\,(0)} \\
= \int_\Sigma \langle c_\oo-g c_\ii g^{-1},-Y+gg^* c_\oo+g c_\ii g^* \rangle .
\end{multline*}
Acting on the partition function with the boundary BFV operators yields
\begin{eqnarray*}
Z^{-1}\Omega_\oo Z&=& \int_\Sigma \Big\langle c_\oo,\bar\dd\As^{1,0}_\oo + (\dd+[\As^{1,0}_\oo,-])(g \As^{0,1}_\ii g^{-1}-\bar\dd g \; g^{-1}) \Big\rangle -\frac12 \langle [c_\oo,c_\oo],gg^* \rangle , \\
 Z^{-1}\Omega_\ii Z&=& \int_\Sigma \Big\langle c_\ii,-\dd\As^{0,1}_\ii - (\bar\dd+[\As^{0,1}_\ii,-])(g^{-1} \As^{1,0}_\oo g+ g^{-1}\dd g ) \Big\rangle +\frac12 \langle [c_\ii,c_\ii],g^* g \rangle .
\end{eqnarray*}
Putting together these computations, we find that
\begin{equation}\label{non-ab mQME in zeroth order}
Z^{-1}\Omega_\ii Z + Z^{-1}\Omega_\oo Z +\frac12 \{S^{\mr{eff}\,(0)},S^{\mr{eff}\,(0)}\}_\res = 0
\end{equation}
--- all the terms in this combination cancel out. This gives us the mQME in the leading order $O(\hbar^0)$.

For the two remaining terms, $\{S^{\mr{eff}\,(0)},\mathbb{W}\}_\res$ and $\Delta_\res S^{\mr{eff}\,(0)}$, we will use the $(\sigma,\As^*_\res)$-parametrization for residual fields. The variation of $j(\sigma)$ (see (\ref{j})) in $\sigma$ is
$$ \delta_\sigma j(\sigma) =\mr{tr}_\ggg P(\ad_\sigma)\ad_{\delta\sigma},\qquad\mr{where}\quad P(x)=\frac12 \coth \frac{x}{2}-\frac{1}{x} . $$
Therefore, using (\ref{W non-ab}), we have
\begin{equation}\label{non-ab (S0,W)}
\begin{aligned}
\{S^{\mr{eff}\,(0)},\mathbb{W}\}_\res&= \sum_{z\in\Sigma} \tr_\ggg P(\ad_\sigma)[-F_+(\ad_\sigma)\circ c_\oo- F_-(\ad_\sigma)\circ c_\ii,\;\bt] \\
&= \sum_{z\in\Sigma} \tr_\ggg P(\ad_\sigma)[- c_\oo+ c_\ii,\;\bt] .
\end{aligned}
\end{equation}
Here the last simplification relies on the identity 
\begin{equation}\label{trace identity}
 \tr_\ggg [(\ad_x)^a\circ y,(\ad_x)^b\bt]=0\qquad \mr{for}\;\; a\geq 1,\; b\geq 0 ,
\end{equation}
which follows from the cyclic property of the trace and Jacobi identity. We also have
\begin{multline}\label{non-ab Delta S0}
\Delta_\res S^{\mr{eff}\,(0)} =\sum_{z\in \Sigma} -\tr_\ggg \sum_{r,s\geq 0} \frac{B_{r+s+1}}{(r+s+1)!} (\ad_\sigma)^r \ad_\bt \Big((-1)^{r+s+1}(\ad_\sigma)^s c_\oo-(\ad_\sigma)^s c_\ii\Big) \\
= \sum_{z\in \Sigma} \tr_\ggg \sum_{r,s\geq 0} \frac{B_{r+s+1}}{(r+s+1)!} \Big[ (-1)^{r+s+1}(\ad_\sigma)^s c_\oo-(\ad_\sigma)^s c_\ii  , (\ad_\sigma)^r\bt\Big] \\
\underset{(\ref{trace identity})}{=}\sum_{z\in\Sigma}\tr_\ggg \sum_{r\geq 0}\frac{B_{r+1}}{(r+1)!} [c_\oo-c_\ii,(\ad_\sigma)^r\bt] 
= \sum_{z\in \Sigma} \tr_\ggg\; \ad_{c_\oo-c_\ii}\circ P(\ad_\sigma) .
\end{multline}
Comparing (\ref{non-ab (S0,W)}) and (\ref{non-ab Delta S0}), we see that they exactly cancel each other pointwise on $\Sigma$. Thus,
$$ \{S^{\mr{eff}\,(0)},\mathbb{W}\}_\res+ \Delta_\res S^{\mr{eff}\,(0)}=0 . $$
Together with (\ref{non-ab mQME in zeroth order}), this finishes the proof of mQME (\ref{non-ab mQME expanded}).
\end{proof}

\begin{remark} The check of the mQME above clearly breaks into two parts: 
\begin{enumerate}[(a)]
\item The classical part 
\begin{equation} Z^{-1}\Omega_\ii Z + Z^{-1}\Omega_\oo Z +\frac12 \{S^{\mr{eff}\,(0)},S^{\mr{eff}\,(0)}\}=0, \label{eq:Class_mQME}
\end{equation}
which is unambiguous and requires no regularization.
\item The quantum part
$$  \{S^{\mr{eff}\,(0)},\mathbb{W}\}_\res+ \Delta_\res S^{\mr{eff}\,(0)}=0 ,  $$
which makes sense with the same regularization as the one used in the proof of Lemma \ref{lemma: S_eff non-ab via sigma}: replacing $\Sigma$ with the set of vertices of a triangulation.
\end{enumerate}
\end{remark}

\textbf{Aside: mQME and Polyakov-Wiegmann formula.} 
The classical part of the mQME, equation \eqref{eq:Class_mQME}, itself splits into two parts: terms involving the antifield $g^*$ and others. The terms involving $g^*$ cancel due to invariance of the inner product. The cancellation of the remaining terms can be understood in terms of the WZW model as follows. The part of the effective action $\mathbb{I}(\As^{0,1}_\ii , \As^{1,0}_\oo;g)$ defined in \eqref{I non-ab} can be identified as the WZW action coupled to two external chiral gauge fields $\As^{0,1}_\ii, \As^{1,0}_\oo$, see, e.g., eq.\ (4.5) in \cite[\S 4.2]{Gawedzki CFT 2}. This coupling is sometimes called ``gauging the $G_L \times G_R$ symmetry'', for instance in \cite{Witten2}. For us it is more natural to call it the  $G_\ii \times G_\oo$-action. Explicitly, the action of $(h_\ii,h_\oo) \in G \times G$ on $(\As^{0,1}_\ii,\As^{1,0}_\oo;g)$ is 
\begin{equation}
(h_\ii,h_\oo)\cdot (\As^{0,1}_\ii,\As^{1,0}_\oo;g) = \left({}^{h_\ii}(\As^{0,1}_\ii),{}^{h_\oo}(\As^{1,0}_\oo);h_\oo g h_\ii^{-1}\right) . \label{eq:GLGRaction}
\end{equation}
It is well known that under the transformation \eqref{eq:GLGRaction} $\mathbb{I}$ is not invariant, but transforms according to the Polyakov--Wiegmann \cite{PW} formula:
\begin{equation}
\mathbb{I}\left({}^{h_\ii}(\As^{0,1}_\ii),{}^{h_\oo}(\As^{1,0}_\oo);h_\oo gh_\ii^{-1}\right) = \mathbb{I}(\As^{0,1}_\ii , \As^{1,0}_\oo;g) - \mathbb{I}(\As^{0,1}_\ii,0;h_\ii) + \mathbb{I}(0, \As^{1,0}_\oo;h_\oo). \label{eq:PW}
\end{equation}
We claim that this equation is just the finite version of the classical part of the mQME, eq.\ \eqref{eq:Class_mQME} above. To see this, consider a path $(h_\ii(t),h_\oo(t))$ of gauge transformations starting at the identity and compute the derivative of \eqref{eq:PW} at $t=0$. The computation is quite 
straightforward and we just sketch it: using eq.\ \eqref{Y constraint}, we get 
$$\left.\frac{d}{dt}\right|_{t=0}\mathbb{I}\left(\As^{0,1}_\ii,\As^{1,0}_\oo,h_\oo(t)gh_\ii(t)^{-1}\right) = \int_\Sigma \langle \dot{h}_\oo - g \dot{h}_\ii g^{-1},Y\rangle.$$
Upon identifying $c_\ii = \dot{h}_\ii, c_\oo = \dot{h}_\oo$ this gives the piece of $\frac12 \{S^{\mr{eff}\,(0)},S^{\mr{eff}\,(0)}\}$ of \eqref{eq:Class_mQME} not involving $g^*$. Then, we find that 
\begin{multline*}\left.\frac{d}{dt}\right|_{t=0}\mathbb{I}\left({}^{h_\ii(t)}(\As^{0,1}_\ii),\As^{1,0}_\oo,g\right) + \mathbb{I}(\As^{0,1}_\ii,0,h_\ii(t))  \\
= - \bar\dd_{\As^{0,1}_\ii}\dot{h}_\ii(g^{-1}A^{1,0}_\oo g + g^{-1}\dd g) - \int \langle\dot{h}_\ii, Y\big|_{g=1,A^{0,1}_\oo = 0}\rangle = Z^{-1}\Omega_\ii Z \bigg|_{g^*=0,c_\ii = \dot{h}_\ii, c_\oo = \dot{h}_\oo}
\end{multline*}
and similarly for the action of $h_\oo$.
Overall, we find 
\begin{multline*}
\left.\frac{d}{dt}\right|_{t=0}\left[\mathbb{I}\left({}^{h_\ii}(\As^{0,1}_\ii),{}^{h_\oo}(\As^{1,0}_\oo);h_\oo gh_\ii^{-1}\right) + \mathbb{I}(\As^{0,1}_\ii, 0, h_\ii(t)) - \mathbb{I}(0,\As^{1,0}_\oo,h_\oo(t)\right] 
\\ =\left[ Z^{-1}\Omega_\ii Z + Z^{-1}\Omega_\oo Z +\frac12 \{S^{\mr{eff}\,(0)},S^{\mr{eff}\,(0)}\}\right]_{g^* = 0, c_\ii = \dot{h}_\ii, c_\oo = \dot{h}_\oo} ,
\end{multline*}
which proves the claim that (a part of) the $\hbar = 0$ part of the mQME is equivalent to the infinitesimal Polyakov--Wiegmann formula. 
We will comment further on the relationship between  Chern--Simons theory on $\Sigma \times I$ and  WZW theory on $\Sigma$ in Section \ref{sec:CSWZW} below. 

\begin{remark}
In the mQME (\ref{mQME non-ab}) and the proof above we were using the $(g,\As^*_\res)$-parametrization of residual fields for the BV Laplacian. The corresponding statement for the BV Laplacian in $(g,g^*)$-parametrization,
$$\Delta_{[g,g^*]}=\int_\Sigma \langle\frac{\delta}{\delta g} , \frac{\delta}{\delta g^*}\rangle =
\int_\Sigma \langle\frac{\delta}{\delta g\; g^{-1}} , \frac{\delta}{g\, \delta g^*}\rangle  , $$
is:
\begin{equation}\label{mQME in g,g^*}
 (\Omega_\oo+\Omega_\ii-\hbar^2\Delta_{[g,g^*]})\;e^{\frac{i}{\hbar}S^{\mr{eff}\,(0)}}=0 .
\end{equation}
Note that here we should not be including the $-i\hbar\mathbb{W}$ term in the effective action, cf.\ Remark \ref{rem: log-half-density shift}. The proof of (\ref{mQME in g,g^*}) is exactly as before in the order $O(\hbar^0)$. In the order $O(\hbar^1)$, we have
$$\Delta_{[g,g^*]} S^{\mr{eff}\, (0)}=\sum_{z\in\Sigma} \frac12 \mr{div}_{T^*[-1]G}\big\{\langle g^*,c_\oo g-g c_\ii \rangle,\bt\big\} . $$
The hamiltonian vector field generated by the ghost term in the effective action is the cotangent lift to $T^*[-1]G$ of the vector field 
$$X=\langle c_\oo,\frac{\dd}{\dd g\; g^{-1}} \rangle-\langle c_\ii ,\frac{\dd}{g^{-1}\dd g} \rangle.$$
--- This is a sum of a right-invariant and a left-invariant vector field on $G$. Since the Haar measure is bi-invariant, $X$ has divergence zero.
Therefore,
$$\Delta_{[g,g^*]} S^{\mr{eff}\, (0)}=\sum_{z\in\Sigma} \mr{div}_G X = 0 . $$
\end{remark}

Ultimately, to avoid the ambiguity as to whether we should be including the term $-i\hbar\mathbb{W}$ into the partition function or not, we can use the invariant formulation where the partition function is a \emph{half-density} (rather than a function) on residual fields and the BV Laplacian is the canonical BV Laplacian on half-densities. Then the mQME is
$$ (\Omega_\oo+\Omega_\ii-\hbar^2 \Delta_\res^\mr{can})\,Z^\mr{can}=0 , $$
where $Z^\mr{can}=e^{\frac{i}{\hbar}S^\mr{eff}_{[x,\xi]}}d^{\frac12}x D^{\frac12}\xi$. Here $(x,\xi)$ can be any Darboux coordinate system on $\VV$, e.g., $(\sigma,\As^*_\res)$ or $(g,g^*)$.

\textbf{Summary.} 
Summarizing the main results of Section \ref{ss: non-ab CS parall ghost}, we have the following:
\begin{itemize}
\item The canonical partition function of the nonabelian theory on the cylinder $[0,1]\times \Sigma $ with the ``parallel ghost'' polarization is:  
$Z^\mr{can}=e^{\frac{i}{\hbar}S^\mr{eff}_{[g,g^*]}} \D^{\frac12} g\, \D^{\frac12} g^*$ where the effective action relative to the coordinate system $(g,g^*)$ on $\VV$ is given explicitly by
\begin{multline}
S^\mr{eff}_{[g,g^*]}=\\
=\int_\Sigma \Big( \langle \As^{1,0}_\oo, g \As^{0,1}_\ii g^{-1} \rangle -
\langle \As^{1,0}_\oo,\bar\dd g\cdot g^{-1} \rangle 
- \langle  \As^{0,1}_\ii, g^{-1} 
 \dd g \rangle \Big)
+\mr{WZW}(g)  \\+
\int_\Sigma -\langle c_\oo,g\, g^* \rangle + \langle c_\ii, g^* g \rangle . \label{eq:CSWZW}
\end{multline} 
In particular, there are no quantum corrections in $S^\mr{eff}_{[g,g^*]}$.
\item $Z^\mr{can}$ satisfies the modified quantum master equation.
\item The restriction of $S^\mr{eff}_{[g,g^*]}$ to ghost number zero fields is the Hamilton--Jacobi action, i.e., is the generalized generating function for the evolution relation of the classical theory obtained by evaluating the classical action on a solution of the evolution equations, see Section \ref{sec:HJ}.
\end{itemize}

\begin{remark}\label{rem: comaprison with Blau-Thompson} The relation between 3D nonabelian Chern--Simons theory and the (gauged) WZW model was studied from different angles in the literature. The closest discussion to ours, perhaps, was in \cite{BT}: $G/G$ WZW theory was recovered from Chern--Simons on a cylinder, using essentially the same gauge fixing and polarization as the ones we employ. But there are crucial differences in the two approaches. We have an explicit Feynman diagram computation of the partition function and 
prove the QME and the gauge invariance property at the quantum level. In \cite{BT}, on the other hand, quantum gauge invariance was assumed and was used to evaluate the Chern--Simons partition function.
\end{remark}

\subsubsection{``Vertical'' Wilson lines}\label{sec:WilsonLines}
One can enrich Chern--Simons theory with Wilson line observables given classicaly by the parallel transport of the connection field $A$ along a curve $\gamma$ ending on the boundary; the parallel transport is evaluated in some linear representation $\rho$ of $G$ on a vector space $R$.

Let us consider a very simple case: several ``vertical'' Wilson lines with $\gamma_j=I\times \{z_j\}$ connecting the in- and out-boundaries of the cylinder $ I\times\Sigma$; here $z_j$ are a collection of points on $\Sigma$, $j=1,\ldots,n$. We are fixing representations $\rho_j$ for the Wilson lines, with $R_j$ the respective representation spaces.

\begin{figure}[H]
\includegraphics[scale=0.7]{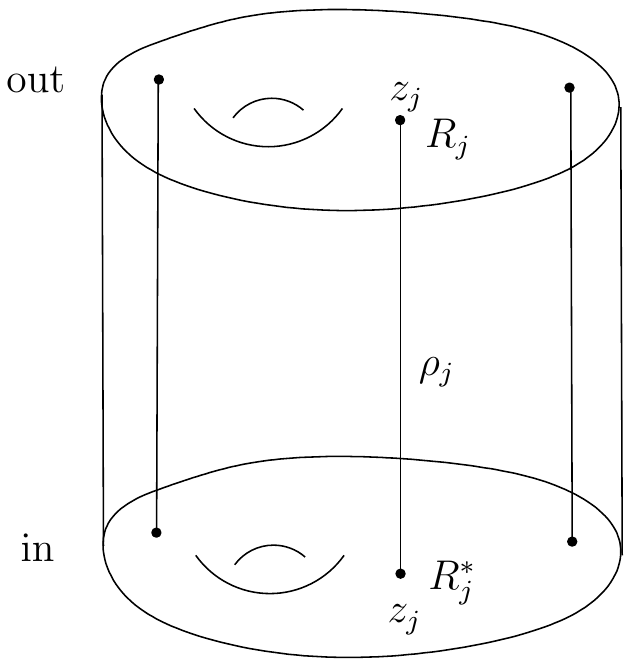}
\caption{Vertical Wilson lines.}
\end{figure}

Note that for our choice of gauge fixing, we have for the Wilson lines 
$$W_j=\rho_j(P\overleftarrow{\exp}(-\int_{\gamma_j} A)) = \rho_j(e^{-\sigma(z_j)})= \rho_j(g(z_j)) \qquad \in \mr{End}(R_j) .$$
I.e., vertical Wilson lines depend only on the residual fields. 

Thus, the partition function of the theory enriched with vertical Wilson lines is:
\begin{equation}\label{non-ab Z with vertical Wilson lines}
Z_{ I\times\Sigma, \{\gamma_j\}}
= 
Z_{ I\times\Sigma} \cdot \bigotimes_j \rho_j(g(z_j)) 
\end{equation}
with $Z_{ I\times\Sigma}=e^{\frac{i}{\hbar}S^\mr{eff}}$ the partition function without the Wilson lines.

The space of out-states is given by functionals of $\As^{1,0}_\oo, c_\oo$ with values in $\bigotimes_j  R_j$, while the space of in-states is given by functionals of $\As^{0,1}_\ii, c_\ii$ also with values in $\bigotimes_j  R_j^*$.\footnote{Our convention is that the partition function is not a  homomorphism from in-states to out-states (depending on residual fields), but an element of $\mathcal{H}_\oo\widehat\otimes \mathcal{H}_\ii$, i.e., without dualization of the in-factor. 
In this description, gluing of two cylinders involves a pairing between the out-states of the first cylinder $\mathcal{H}_\oo^I$ and the in-states $\mathcal{H}_\ii^{II}$ of the second.
} The BFV operators are:\footnote{
See \cite{ABM} for a construction of these BFV operators from the presentation of Wilson lines via a path integral over auxiliary fields supported on $\gamma_j$ (the Alekseev--Faddeev--Shatashvili formula).
}
\begin{equation} \label{Omega with Wilson lines}
\Omega_\oo^{\Sigma,\{z_j\}} = \Omega_\oo^\Sigma + i\hbar \sum_j \rho_j(c_\oo(z_j)) \;\;, \quad
\Omega_\ii^{\Sigma,\{z_j\}} = \Omega_\ii^\Sigma + i\hbar \sum_j  \rho_j^*(c_\ii(z_j))  ,
\end{equation}
where $\Omega_\oo^\Sigma$, $\Omega_\ii^\Sigma$ are the BFV operators for the theory without the Wilson lines, given by (\ref{non-ab Omega out}), (\ref{non-ab Omega in}); $\rho_j^*$ is the dual representation to $\rho_j$ with representation space $R_j^*$. 

As a direct consequence of  Lemma \ref{lemma: mQME non-ab}, one has that the partition function with Wilson lines (\ref{non-ab Z with vertical Wilson lines}) satisfies the modified quantum master equation:
$$ (\Omega_\oo^{\Sigma,\{z_j\}}+\Omega_\ii^{\Sigma,\{z_j\}}-\hbar^2 \Delta_\res )\, Z_{ I\times\Sigma, \{\gamma_j\}} =0 . $$
Here we understand that the $\rho_j^*$ term in $\Omega_\ii$ acts on the second factor in $\rho_j(g(z_j))\in R_j\otimes R_j^*$, while the $\rho_j$ term in $\Omega_\oo$ acts on the first factor.

\subsubsection{The CS-WZW correspondence: WZW theory as an effective theory of Chern--Simons}\label{sec:CSWZW}
Equation \eqref{eq:CSWZW} is evidence of a strong relationship between the Chern--Simons theory on a manifold with boundary $\Sigma$ and the WZW theory on the Riemann surface $\Sigma$. This relationship has, of course, already been subject to a lot of scrutiny after Witten's seminal article \cite{Witten}. In the approach of this paper, this relationship stems from the fact that the gauged WZW action emerges as an effective action of the Chern--Simons theory, as is clear from eq.\ \eqref{eq:CSWZW}. To be precise, the following two theories are equivalent: 
\begin{enumerate}[i)]
\item The BV-BFV effective theory of Chern--Simons on a $I\times \Sigma$, restricted to the gauge-fixing lagrangian $\LL = \{g^* = 0\} \subset T^*[-1]G$. 
\item The WZW theory with gauged ``$G_\ii \times G_\oo$''-symmetry.
\end{enumerate} This is a very strong statement of equivalence: It means that essentially all quantities computed from the action functional in gauged WZW theory have an expression in Chern--Simons theory. We summarize this relationship in Table \ref{tab:CSWZW} below.
\begin{table}[h]
\begin{center}
\begin{tabular}{c|c}
Object in CS on $I\times \Sigma $ & Object in gauged WZW on $\Sigma$ \\
\hline \\

Effective action $S^\eff_{[g,g*]}$ & Gauged WZW action $\mathbb{I}(\As^{0,1}_\ii,\As^{1,0}_\oo;g)$ \\
&\\
mQME $(\Omega - \Delta_\res )Z = 0$ &\makecell{ Polyakov--Wiegmann formula \eqref{eq:PW} \\ (group 1-cocycle property) }\\
& \\
\makecell{Expectation value $W$ of Wilson line \\ $\gamma = I\times \{z\}$ in rep.  $\rho$ } & Field insertion $\rho(g(z))$ \\ 
& 
\end{tabular}
\end{center}
\caption{The CS-WZW correspondence}\label{tab:CSWZW}
\end{table}

\begin{remark}
One might wonder why in Table \ref{tab:CSWZW} on the left hand side we have objects defined in the quantization on the Chern--Simons theory, while on the right-hand side we have entirely classical objects in the WZW model. This apparent puzzle is resolved by the observation that on the left-hand side we are seeing only the
semiclassical limit of the quantum Chern--Simons theory (which in this case happens to be exact, since there are no loop contributions).\footnote{Up to the subtleties concerning ghost loops discussed above.} 
\end{remark}
\begin{remark}[Nonequivalent gauge-fixing lagrangians]
Instead of  $\LL = \{g^* =  0\}$, one could restrict the effective Chern--Simons action \eqref{eq:CSWZW} also to another lagrangian $\LL'\subset T^*[-1]G$ given by $g=1$. In that way, one obtains 
\begin{equation}
S^\eff_{[g,g^*]}\bigg|_{g=1} = \int_\Sigma \langle \As^{1,0}_\oo,\As^{0,1}_\ii\rangle + \int_\Sigma \langle g^*, c_\ii - c_\oo\rangle .
\end{equation}
Upon integrating $Z$ over $\mathcal{L}$ or $\mathcal{L'}$, one obtains two $(\Omega_\ii + \Omega_\oo)$-cocycles $Z_1$, $Z_2$. $Z_1$ is concentrated in ghost degree 0 (we will discuss it in more detail in the next subsection) while 
$$Z_2 = \exp\left(\frac{i}{\hbar}\int_\Sigma \langle \As^{1,0}_\oo,\As^{0,1}_\ii\rangle \right)\delta( c_\ii - c_\oo) $$
has nonzero ghost number (formally, it is infinite, $\gh=\dim \Omega^0(\Sigma,\g)$). 
Therefore, $\LL$ and $\LL'$ provide an example of nonequivalent gauge-fixing lagrangians. \end{remark}

In the $(g,g^*)$-coordinates we can define a particularly simple gauge-fixing lagrangian $\LL$ given by $g^* = 0$ (for the general remarks in this section, we will allow ourselves to ignore issues arising from possible zero modes).
\begin{equation}
S^\eff\big|_\LL = \mr{WZW}(g) + \langle \As^{1,0}_\oo, g \As^{0,1}_\ii g^{-1} \rangle -
\langle \As^{1,0}_\oo,\bar\dd g\cdot g^{-1} \rangle 
- \langle  \As^{0,1}_\ii, g^{-1} 
 \dd g \rangle \equiv \mathbb{I}[\As^{1,0}_\oo,\As^{0,1}_\ii;g].
\end{equation} 
Here $ \mathbb{I}[\As^{1,0}_\oo,\As^{0,1}_\ii;g]$ is the standard way of gauging the WZW action, see, e.g., eq.\ (4.5) in \cite[\S 4.2]{Gawedzki CFT 2}. We can then express the Chern--Simons partition function on $I\times \Sigma$ as 
\begin{multline}
Z_{I\times \Sigma }[\As^{1,0}_\oo,\As^{0,1}_\ii, c_\ii,c_\oo]  \equiv Z_{I\times \Sigma }[\As^{1,0}_\oo,\As^{0,1}_\ii]  = \int_g \exp\frac{i}{\hbar} \mathbb{I}[\As^{1,0}_\oo,\As^{0,1}_\ii;g]\DD g\label{eq:CSWZW2}
\end{multline}
(notice the partition function does not depend on $c_\ii,c_\oo$). 
This is the definition of the partition function $Z_\As^\mr{WZW}$ of gauged WZW, see, e.g., eq.\ (4.7) in \cite[\S 4.2]{Gawedzki CFT 2}. Here we abbreviate $\As = (\As^{1,0}_\oo, \As^{0,1}_\ii)$.
Similarly, we see that a correlator in the gauged WZW theory can be expressed as the partition function of Chern--Simons theory enriched with Wilson lines: 
\begin{multline} 
\langle \rho_1(g(z_1))\otimes\ldots \otimes \rho_n (g(z_n))\rangle_\As \\
= \int_g \rho_1(g(z_1))\otimes\ldots \otimes \rho_n(g(z_n)) e^{\frac{i}{\hbar}\mathbb{I}[g,\As]}\DD g = \int_g Z_{I\times\Sigma,\{\gamma_j\}} \DD g. \label{eq:CSWZW3}
\end{multline}

For the purposes of this subsection, we will treat the path integral expressions on the right hand side of \eqref{eq:CSWZW2} and \eqref{eq:CSWZW3} heuristically. In the literature, these objects are typically defined via representation theory. In this paper, we are typically interpreting path integral expressions as defined via Feynman graphs and rules, but 
for WZW the absence of a natural linear structure on the target (the group $G$) obstructs the treatment of the path integral as a perturbed Gaussian integral.
We will therefore simply assume that the partition function exists and defines an element in $\Omega$-cohomology. In \cite{ABM} it was shown that in genus 0, the $\Omega$-cohomology with $n$ Wilson lines ending on the boundary can be identified with the $n$-point space of conformal blocks. We expect this to hold for all genera, and assume it for the purpose of the next section. 
We summarize the content of the CS-WZW correspondence after integrating over $\mathcal{L}$ in Table \ref{tab:CSWZWII} below. 
\begin{table}[h]
\begin{center}
\begin{tabular}{c|c}
Object in CS on $I\times \Sigma $ & Object in gauged WZW on $\Sigma$ \\
\hline \\

CS partition function $ Z_{I\times \Sigma }[\As^{1,0}_\oo,\As^{0,1}_\ii] $ & Gauged WZW partition function $Z_\As^\mr{WZW}$  \\
&\\

\makecell{Expectation value $W$ of $n$ Wilson lines \\ $\gamma = I\times \{z_i\}$ in rep.  $\rho_1,\ldots \rho_n$ } & 
\makecell{Gauged WZW correlator \\
$\langle \rho_1(g(z_1))\cdots \rho_n (g(z_n))\rangle_\mathsf{A}$ }
\\ 
& 
\\
$\Omega$-cohomology with $n$ Wilson lines & $n$-point space of conformal blocks \\
&
\end{tabular}
\end{center}
\caption{The CS-WZW correspondence after integrating over the gauge-fixing lagrangian $\mathcal{L}$. }\label{tab:CSWZWII}
\end{table}

\subsubsection{An application: Holomorphic factorization of the WZW theory}\label{sec: hol fact WZW}
We will now discuss an application of the correspondence observed above. The arguments in this section will be more heuristic in nature. 

Suppose we fix the boundary condition on one side, e.g., fix the antiholomorphic boundary condition at the $\ii$-boundary by setting $\As^{0,1}_\ii = 0$ (remember that in our treatment of boundary conditions, setting $\As^{0,1}_\ii = 0$ means that $\As\big|_{\Sigma_\ii} \in \Omega^{1,0}(\Sigma,\g)$
), and take the $\oo$-boundary in the $\As^{1,0}_\oo$-representation. We will call such a cylinder a ``chiral cylinder.'' 

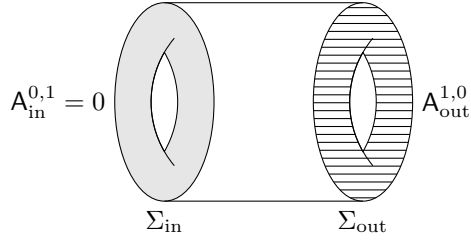
\begin{figure}[h]
\begin{tikzpicture}[scale=0.66]

\draw[fill=gray!20] (0,0) ellipse (1cm and 2cm);
\draw[fill=white] (0,1) arc (150:210:2cm) arc (330:390:2cm);
\draw (0,1) arc (150:140:2cm);
\draw (0,1) arc (150:220:2cm);
\draw (0,2) -- (4,2);
\draw (0,-2) -- (4,-2);
\node[coordinate,label=below:{$\Sigma_\ii$}] at (0,-2) {};
\node[coordinate,label=left:{$\As^{0,1}_\ii = 0$}] at (-1,0){};
\begin{scope}[shift={(4,0)}]
\draw[pattern=horizontal lines] (0,0) ellipse (1cm and 2cm);
\draw[fill=white] (0,1) arc (150:210:2cm) arc (330:390:2cm);
\draw (0,1) arc (150:140:2cm);
\draw (0,1) arc (150:220:2cm);
\node[coordinate,label=below:{$\Sigma_\oo$}] at (0,-2) {};
\node[coordinate,label=right:{$\As^{1,0}_\oo$}] at (1,0){};
\end{scope}
\end{tikzpicture}
\caption{A chiral cylinder: antiholomorphic boundary conditions on $\Sigma_\ii$, $\As^{1,0}_\oo$-representation on $\Sigma_\oo$. Gray indicates that we fix a boundary condition on this boundary, while hatching indicates we fix only the polarization.} 
\end{figure}

 After integrating out $g$, we obtain the partition function  $\psi(\As_{\oo}^{1,0})$ of a ``chiral gauged WZW theory,'' i.e.,  a WZW theory coupled to a chiral gauge field, see, e.g., \cite{Witten2}.\footnote{Where Witten suggests 
 that ``this in fact can be regarded as the essential relation between the
WZW model and Chern--Simons theory.''} This partition function is not a number, but rather a holomorphic gauge invariant section of a line bundle over the space of connections on $\Sigma$.\footnote{In our approach, holomorphicity simply follows from the fact that it depends only on $\As^{1,0}_\oo$, while gauge invariance is the statement that $\Omega_\oo \psi = 0$.} We can glue the chiral cylinder to another ``antichiral'' cylinder with opposite boundary conditions (see Figure \ref{fig:CScyls}). 
\begin{figure}
\begin{tikzpicture}[scale=0.66]

\draw[fill=gray!20] (0,0) ellipse (1cm and 2cm);
\draw[fill=white] (0,1) arc (150:210:2cm) arc (330:390:2cm);
\draw (0,1) arc (150:140:2cm);
\draw (0,1) arc (150:220:2cm);
\draw (0,2) -- (4,2);
\draw (0,-2) -- (4,-2);
\begin{scope}[shift={(4,0)}]
\draw[pattern=horizontal lines] (0,0) ellipse (1cm and 2cm);
\draw[fill=white] (0,1) arc (150:210:2cm) arc (330:390:2cm);
\draw (0,1) arc (150:140:2cm);
\draw (0,1) arc (150:220:2cm);
\end{scope}
\node[coordinate,label=below:{\huge{*}}] at (5.25,0) {};
\begin{scope}[shift={(6.5,0)}]
\draw[pattern=vertical lines] (0,0) ellipse (1cm and 2cm);
\draw[fill=white] (0,1) arc (150:210:2cm) arc (330:390:2cm);
\draw (0,1) arc (150:140:2cm);
\draw (0,1) arc (150:220:2cm);
\draw (0,2) -- (4,2);
\draw (0,-2) -- (4,-2);
\begin{scope}[shift={(4,0)}]
\draw[fill=gray!60] (0,0) ellipse (1cm and 2cm);
\draw[fill=white] (0,1) arc (150:210:2cm) arc (330:390:2cm);
\draw (0,1) arc (150:140:2cm);
\draw (0,1) arc (150:220:2cm);
\end{scope}
\end{scope}
\node[coordinate,label=below:{\huge{=}}] at (12,0) {};
\begin{scope}[shift={(13.5,0)}]
\draw[fill=gray!20] (0,0) ellipse (1cm and 2cm);
\draw[fill=white] (0,1) arc (150:210:2cm) arc (330:390:2cm);
\draw (0,1) arc (150:140:2cm);
\draw (0,1) arc (150:220:2cm);
\draw (0,2) -- (4,2);
\draw (0,-2) -- (4,-2);
\begin{scope}[shift={(4,0)}]
\draw[fill=gray!60] (0,0) ellipse (1cm and 2cm);
\draw[fill=white] (0,1) arc (150:210:2cm) arc (330:390:2cm);
\draw (0,1) arc (150:140:2cm);
\draw (0,1) arc (150:220:2cm);
\end{scope}
\end{scope}
\end{tikzpicture}
\caption{Gluing a chiral and an antichiral cylinder into a cylinder with opposite chiral boundary conditions. Gray indicates a fixed boundary condition, hatching indicates a polarized boundary.}
\label{fig:CScyls}
\end{figure}
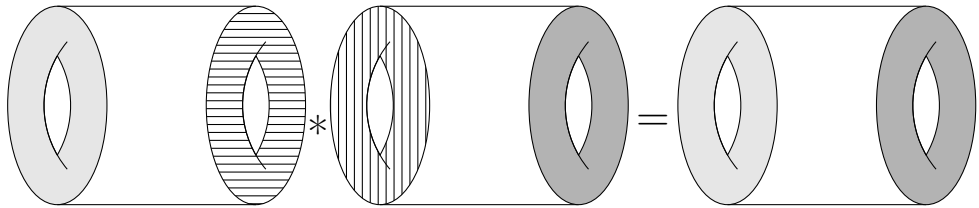In this way, we obtain  -- as explained in \cite{Witten2} -- the square of the norm of $\psi$.\footnote{A missing factor of $\int_\Sigma A^{1,0}A^{0,1}$ in comparison with \cite{Witten2} comes from the gluing conventions in BV-BFV, see Appendix \ref{app:SB}.} Here, the ``norm square'' should be taken with respect to a well-defined inner product on the $\Omega$-cohomology, i.e., the finite-dimensional moduli space of gauge-invariant holomorphic sections.\footnote{See, for instance, \cite{Gawedzki_CS_Zero} for a more detailed discussion in genus 0, or \cite{FG_CS_One} for a discussion in genus 1. The authors do not know of an explicit construction of this inner product in higher genera.} On the other hand, from the general principles of the BV-BFV formalism, we will then obtain the partition function of Chern--Simons theory with opposite chiral boundary conditions, which is given by specializing to $\As^{1,0}_\oo = \As^{1,0}_\ii = 0$  in \eqref{eq:CSWZW2}: 
\begin{equation}
|\psi|^2 = Z_{I\times\Sigma}^{\mr{CS}} = Z_{\Sigma}^\mr{WZW}.\label{eq:HolomorphicFactorization}
\end{equation}
Here on the left-hand side we have the norm-square of the partition function of chiral WZW, in the middle we have Chern--Simons partition function on the cylinder with opposite chiral boundary conditions, while on the right-hand side we have the definition of the WZW partition function. Equation \eqref{eq:HolomorphicFactorization} is sometimes called ``holomorphic factorization of the WZW model'', because one sees that the partition function of the full WZW model --- which does not vary holomorphically on the moduli space of conformal structures --- splits 
into a sum of products of holomorphic and  antiholomorphic factors, which do depend (anti)holomorphically on the complex structure.  Thus, holomorphic factorization of the WZW model follows from the self-similarity of the Chern--Simons partition function on cylinders. \\
Using the results of Section \ref{sec:WilsonLines}, in particular  equation \eqref{non-ab Z with vertical Wilson lines}, these results for the partition function generalize to correlators in chiral and full WZW. Namely, suppose we are given $n$ Wilson lines colored by representations $R_1,\ldots,R_n$, and let $V = \otimes_j R_j$. Then, the Chern--Simons partition function with Wilson lines on a chiral cylinder $\psi_{\{\gamma_j\}}$ is naturally a degree zero element of $V \otimes \mathcal{H}_\oo$ with $\mathcal{H}_\oo$ the space of functionals  of $\As^{1,0}_\oo,c_\oo$ with values in $V^*$. Gluing with an antichiral cylinder, we obtain the Chern--Simons partition function with Wilson lines and opposite boundary conditions - the correlator 
$\langle \rho_1(g(z_1)) \cdots \rho_n(g(z_n)) \rangle \in V \otimes V^*$ in a full WZW model. On the other hand, explicitly computing the BV-BFV gluing we obtain\footnote{
The gluing is defined as a formal functional integral, but its restriction to the finite-dimensional $\Omega$-cohomology gives rise to a well-defined inner product.} $(\psi,\overline{\psi}) \in V \otimes V^*$. Here $(\cdot,\cdot)$ is the inner product on the space on $n$-point conformal blocks. Thus, we obtain the generalization of \eqref{eq:HolomorphicFactorization} to the case with Wilson lines: 
\begin{equation}
(\psi_{\{\gamma_j\}},\overline{\psi}_{\{\gamma_j\}}) = Z^{CS}_{I\times\Sigma,\{\gamma_j\}} = \langle \rho_1(g(z_1))\cdots \rho_n(g(z_n))\rangle.
\end{equation}  \\

\subsection{3D nonabelian Chern--Simons theory in holomorphic-to-holomorphic polarization}
\label{ss: non-ab hol-hol}
Next, consider the nonabelian Chern--Simons theory on $\Sigma\times [0,1]$ with polarizations as in Section \ref{ss: hol-hol}. The residual fields 
$$\As^+_\Ires=\As^{0}_{\Ires}+\As^{1,0}_{\Ires},\qquad \As^-_\res= \As^{0,1}_\res+ \As^2_\res $$ 
and the gauge fixing are as in Section \ref{ss: hol-hol} (but now all forms are $\ggg$-valued). We will use the notations $\sigma=\As^{0}_\Ires$, $\lambda=\As^{0,1}_\res$ for $\gh=0$ residual fields, as in Section \ref{sss: hol-hol HJ}.

In this subsection we will present only the results; the computations 
are similar to those of Section \ref{ss: non-ab CS parall ghost}.

The Feynman diagrams for the partition function $Z=e^{\frac{i}{\hbar}S^\mr{eff}}$ are: 
\begin{figure}[H]
\includegraphics[scale=0.7]{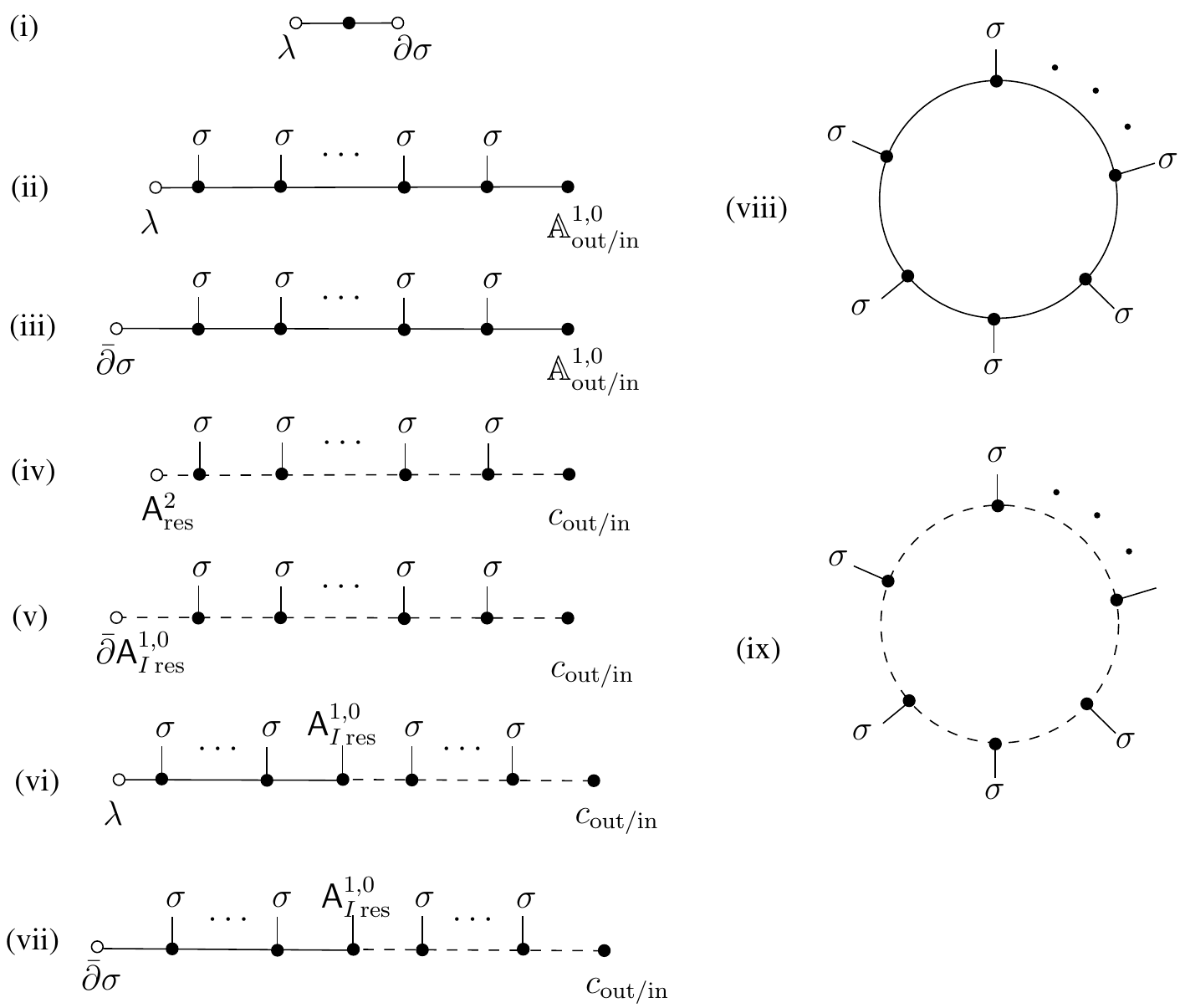}
\caption{Feynman diagrams in nonabelian theory on a cylinder in holomorphic-to-holomorphic polarization.}
\label{hol-hol non-ab diagrams}
\end{figure}
Here the ``physical wheels'' (viii) and the ``ghost wheels'' (ix) cancel each other, due to the form of propagators in the chosen polarization.

Calculating the Feynman diagrams, one finds the following expression for the effective action:
\begin{equation}\label{hol-hol non-ab S_eff}
S^\mr{eff}= S^{\mr{eff}}_\ph+S^{\mr{eff}}_\gh ,
\end{equation}
where the part depending only on ``physical'' ($\gh=0$) fields (the contribution of diagrams (i), (ii), (iii)) is
\begin{multline}\label{hol-hol non-ab S_eff phys part}
S^{\mr{eff}}_\ph=\int_\Sigma \langle  \lambda,\dd\sigma\rangle +
\left\langle \As^{1,0}_\oo, \frac{\ad_\sigma}{e^{\ad_\sigma}-1}\circ\lambda+ \left(-\frac{1}{e^{\ad_\sigma}-1}+\frac{1}{\ad_\sigma} \right)\circ\bar\dd\sigma \right\rangle\\
+\left\langle \As^{1,0}_\ii, -\frac{\ad_\sigma}{1-e^{-\ad_\sigma}}\circ\lambda+ \left(\frac{1}{1-e^{-\ad_\sigma}}-\frac{1}{\ad_\sigma} \right)\circ\bar\dd\sigma \right\rangle
\end{multline}
and the ghost-dependent part (the contribution of diagrams (iv)--(vii)) is
\begin{multline}\label{hol-hol non-ab S_eff ghost part}
S^{\mr{eff}}_\gh=\int_\Sigma\left\langle c_\oo, \sum_{k\geq 0}\frac{B_k^-}{k!}(\ad_\sigma)^k \As^2_\res  - 
\sum_{k\geq 0} \frac{B^-_{k+1}}{(k+1)!}(\ad_\sigma)^k\bar\dd\As^{1,0}_\Ires \right. \\
\left.
+\sum_{k,l\geq 0}\frac{B^-_{k+l+1}}{(k+l+1)!}(\ad_\sigma)^k\ad_{\As^{1,0}_\Ires}(\ad_\sigma)^l \lambda 
- \sum_{k,l\geq 0} \frac{B^-_{k+l+2}}{(k+l+2)!} (\ad_\sigma)^k\ad_{\As^{1,0}_\Ires}(\ad_\sigma)^l \bar\dd\sigma
  \right\rangle\\
+
\left\langle c_\ii, -\sum_{k\geq 0}\frac{B_k^+}{k!}(\ad_\sigma)^k \As^2_\res 
+\sum_{k\geq 0} \frac{B^+_{k+1}}{(k+1)!}(\ad_\sigma)^k \bar\dd\As^{1,0}_\Ires \right. \\
\left.
-\sum_{k,l\geq 0} \frac{B^+_{k+l+1}}{(k+l+1)!} (\ad_\sigma)^k\ad_{\As^{1,0}_\Ires}(\ad_\sigma)^l \lambda
+\sum_{k,l\geq 0} \frac{B^+_{k+l+2}}{(k+l+2)!} (\ad_\sigma)^k\ad_{\As^{1,0}_\Ires}(\ad_\sigma)^l \bar\dd\sigma
 \right\rangle .
\end{multline}
Here $B_n^\pm$ are the Bernoulli numbers with $B^\pm_1=\pm\frac12$ and with $B^+_n=B^-_n$ the usual Bernoulli numbers for $n\neq 1$  (thus, $B^-_n=B_n$ are the standard Bernoulli numbers for $n=0,1,2,\ldots$):

\begin{center}
\begin{tabular}{c|cccccccc}
$n$ & $0$ & $1$ & $2$ & $3$ & $4$& $5$ & $6$ & $\cdots$ \\ \hline  
$B^+_n$ & $1$ & $+\frac12$ & $\frac{1}{6}$ & $0$ & -$\frac{1}{30}$ & $0$ & $\frac{1}{42}$ & $\cdots$\\
$B^-_n$ & $1$ & $-\frac12$ & $\frac{1}{6}$ & $0$ & -$\frac{1}{30}$ & $0$ & $\frac{1}{42}$ & $\cdots$
\end{tabular}
\end{center}

\begin{remark}
Another form of the ghost-dependent part of the effective action (\ref{hol-hol non-ab S_eff ghost part}), with sums over $k,l$ below evaluated explicitly, is:
\begin{multline*}
S^{\mr{eff}}_\gh=\int_\Sigma
\Big\langle c_\oo, \frac{\ad_\sigma}{e^{\ad_\sigma}-1}\circ \As^2_\res +\Big(-\frac{1}{e^{\ad_\sigma}-1}+\frac{1}{\ad_\sigma} \Big)\circ \bar\dd\As^{1,0}_\Ires\\ 
+\ad_{\As^{1,0}_\Ires}\frac{1}{e^{\ad_\sigma}-1}\circ\lambda-\frac{\ad_\sigma}{1-e^{-\ad_\sigma}}\ad\left(\frac{1-e^{-\ad_\sigma}}{\ad_\sigma}\circ \As^{1,0}_\Ires \right)\frac{1}{e^{\ad_\sigma}-1}\circ \lambda\\
-\frac{1}{\ad_\sigma}\ad_{\As^{1,0}_\Ires}\frac{1}{\ad_\sigma}\bar\dd\sigma + 
\frac{1}{e^{\ad_\sigma}-1} \ad\left(\frac{e^{\ad_\sigma}-1}{\ad_\sigma}\circ \As^{1,0}_\Ires\right) \frac{1}{1-e^{-\ad_\sigma}}\circ \bar\dd\sigma  \Big\rangle \\
+
 \Big\langle c_\ii, \Big(\sigma \ra -\sigma, \As^2_\res \ra -\As^2_\res \Big) 
\Big\rangle .
\end{multline*}
Here the coefficient of $c_\ii$ is obtained from the coefficient of $c_\oo$ by replacing $\sigma$ with $-\sigma$ and replacing $\As^2_\res$ with $-\As^2_\res$.
\end{remark}

Next, one can introduce a new parametrization of the space of residual fields by a group-valued map $g:\Sigma\ra G$, a  $(0,1)$-form $\Lambda$, a   $(1,0)$-form $\Lambda^*$ and a $2$-form $g^*$:\footnote{$\Lambda,\Lambda^*$ are $\ggg$-valued forms while $g^*$ is a $\ggg$-valued form  translated by $g^{-1}$. The ghost number
 is $0$ for $g,\Lambda$ and is $-1$ for $g^*,\Lambda^*$.}
\begin{align*}
g&=e^{-\sigma},\\
\Lambda&= \mbox{coefficient of $\As^{1,0}_\oo$ in (\ref{hol-hol non-ab S_eff phys part})} \\
& =
\frac{\ad_\sigma}{e^{\ad_\sigma}-1}\circ\lambda+ \left(-\frac{1}{e^{\ad_\sigma}-1}+\frac{1}{\ad_\sigma} \right)\circ\bar\dd\sigma, \\
\Lambda^*&=\frac{1-e^{-\ad_\sigma}}{\ad_\sigma}\circ \As^{1,0}_\Ires,\\
g^*&= (\mbox{coefficient of $c_\ii$ in (\ref{hol-hol non-ab S_eff ghost part}) })\cdot g^{-1} .
\end{align*}
This change of parametrization has Jacobian $1$ and changes one Darboux coordinate system with respect to the BV symplectic form on $\mathcal{V}$ into another one:
$$ \omega_\res=\int_\Sigma \langle \delta \sigma,\delta \As^2_\res \rangle + \langle \delta \lambda, \delta \As^{1,0}_\Ires \rangle =  \int_\Sigma \langle \delta g,\delta g^* \rangle + \langle \delta \Lambda,\delta \Lambda^* \rangle . $$
In terms of this new parametrization, the effective action (\ref{hol-hol non-ab S_eff}) can be written 
more concisely:
\begin{multline}\label{hol-hol non-ab S_eff via g}
S^\mr{eff}=- \mr{WZW}(g^{-1}) -\int_\Sigma \langle \Lambda,\dd g\cdot g^{-1} \rangle 
\\
+\int_\Sigma
\langle \As^{1,0}_\oo,\Lambda \rangle + \langle\As^{1,0}_\ii,g^{-1}\Lambda g+g^{-1}\bar\dd g\rangle +\langle c_\oo, -g\, g^*+\bar\dd\Lambda^*+[\Lambda,\Lambda^*] \rangle +
\langle c_\ii,g^*g  \rangle  ,
\end{multline}
where $\mr{WZW}$ is the Wess--Zumino--Witten action defined as in (\ref{WZW def}).

The effective action (\ref{hol-hol non-ab S_eff via g})
satisfies the following properties:
\begin{itemize}
\item Its restriction to $\gh=0$ fields satisfies the Hamilton--Jacobi property, i.e., it is the generalized generating function for the evolution relation of nonabelian Chern--Simons theory. From this identification one can see that, on-shell, $\Lambda$ can be interpreted as the $(0,1)$-component of the connection field at $t=1$.\footnote{There is, of course, a similar change of variables where instead we take the $(0,1)$-component at $t=0$ to be a coordinate on $\mathcal{V}$. It leads to an expression for $S^\mr{eff}$ where the symmetry between in/out boundaries is broken in the opposite way to (\ref{hol-hol non-ab S_eff via g}): the coefficients of $\As^{1,0}_\ii$, $c_\oo$ are simple and the coefficients of $\As^{1,0}_\oo$, $c_\ii$ are more complicated.}
\item One has the modified quantum master equation
$$ (\Omega_\oo+\Omega_\ii -\hbar^2\Delta_\res) e^{\frac{i}{\hbar} S^\mr{eff}}=0 $$
with the boundary BFV operators
$$ 
\Omega_{\oo/\ii}= \int_\Sigma \big\langle c,\pm \bar\dd \As^{1,0} -i\hbar(\dd+[\As^{1,0},-])\frac{\delta}{\delta \As^{1,0}} \big\rangle -i\hbar \big\langle \frac12 [c,c],\frac{\delta}{\delta c} \big\rangle .
$$
Here the sign $\pm$ is $+$ for out-boundary and $-$ for in-boundary;
we also suppressed the $\oo/\ii$ subscript in the boundary fields
$\As^{1,0}$   and $c$.
\end{itemize}

\section{BV-BFV approach to higher-dimensional Chern--Simons theories}\label{sec:higherdimCS}
 The observations on abelian Chern--Simons theory in Sections~\ref{ss: 1d ab CS} and~\ref{s: 3dCS}  generalize readily to  cylinders $I\times M$ of other dimensions $d$. Observe that $d$ must be odd because we want the field $\mc{A}$ to belong to the superspace $\Omega^\bullet(I\times M)$ or $\Pi\Omega^\bullet(I\times M)$ and, in either case, the BV action $S=\int_{ I\times M} \frac12 \mc{A}\wedge d\mc{A}$ is even if{f} $d$ is odd.

In the following, we will actually focus on the graded case where the field $\mc{A}$ belongs to the graded space $\Omega^\bullet(I\times M)[k]$ for some integer $k$ and the BV action has degree zero. This forces $d=2k+1$.
If $k$ were even, we would have $\mc{A}\wedge d\mc{A}=\frac12d\mc{A}^2$, so the BV action would have no bulk contribution. Therefore, we will have to assume that $k$ is odd. To summarize:\footnote{There are other ways to get at this results. For example, we may consider the classical theory with action $S_\text{cl}=\int_{ I\times M} \frac12 AdA$ with $A$ a $k$-form for some $k$. This immediately forces $d=2k+1$. For $k$ even, we have $AdA=\frac12dA^2$, so the classical action has no bulk term. We then have to assume $k$ odd. Another option is to consider the AKSZ construction with target $\RR[k]$, for some $k$, endowed with a symplectic form. If we denote by $x$ the coordinate, the general $2$-form is $\omega=f(x)dxdx$ with $f$ a function. If $k$ is even, $dxdx=0$ and $\omega$ is degenerate, so we have to assume $k$ odd. In this case $f(x)=a+bx$ for some real numbers $a$ and $b$. Now $\omega$ is closed if{f} $b=0$ and nondegenerate if{f} $a\not=0$ (we may, e.g., take $a=\frac12$). We then have that $\omega$ has degree $2k$. Since we want to produce a BV form (i.e., degree equal to $-1$) on
$ I\times M$ by the AKSZ construction, we need $d=2k+1$.}
\[
d=\dim(I\times M)=2k+1,\quad k=2l+1.
\]
The case $k=1$ has been considered in Section~\ref{s: 3dCS}. We will briefly describe the general case before turning to the next example of interest, $k=3$. 

Next we assume that the $2k$-dimensional manifold
$M$ is closed and oriented. Again, we can construct a BV-BFV theory by the AKSZ construction as 
$$\FF= \mr{Map}(T[1](I \times M), \RR[k]) =\Omega^\bt(I \times M)[k]  $$
and rewrite this space of fields in the form 
$$\FF = \Omega^\bt(I,\Omega^\bt(M)[k]),$$
exhibiting the theory as a 1-dimensional Chern--Simons theory with coefficients in $\g = \Omega^\bt(M)$. 
The BV action is then, mimicking \eqref{ab CS S}, 
\begin{equation*}
S = \int_{I \times M} \frac12 \mc{A} \wedge d\mc{A} = \int_I \frac12 (\mc{A},d_I\mc{A}) + \frac{1}{2}(\mc{A},d_M\mc{A}),
\end{equation*}
where $d = d_I + d_M$ and $(a,b) = \int_M a \wedge b$. Again, the field $\mc{A}$ can be split as $\mc{A} = \As + dt \cdot \AsI$ and the boundary phase space is $\PhiPM_M = \Omega^\bt(M)[k]$ with 
Noether 1-form 
$$\alpha =\frac12 \int_{ \{1\} \times M}\As \wedge \delta \As -
\frac12 \int_{ \{0\}\times M}\As \wedge \delta \As.$$ 

Next, assume that $M$ carries a complex structure. Then we can split the space of complexified $k$-forms as 
$$\Omega^k_\CC(M) = \bigoplus_{j_1 + j_2 = k}\Omega^{j_1,j_2}(M).$$
Given that $k$ is odd, the splitting 
\begin{equation} \label{Omega^k splitting into + and -}
\Omega^k_\CC(M) = \underbrace{\bigoplus_{j = 0}^{l}\Omega^{k-j,j}(M)}_{\Omega^k_+(M)} \oplus \underbrace{\bigoplus_{j = l+1}^{k}\Omega^{k-j,j}(M)}_{\Omega^k_-(M)}
\end{equation}
provides a splitting into lagrangian subspaces of $\Omega^k_\CC(M)$ (which is the degree 0 subspace of $\PhiPM_M$). The splitting \eqref{ab CS pol} generalizes to 
\begin{equation*}
\underbrace{\Omega^\bt_\CC(M)}_{\g_\CC} = \underbrace{\bigoplus_{j=0}^{k-1}\Omega^j_\CC(M) \oplus \Omega^k_+(M)}_{\g^+_\CC} \oplus \underbrace{\Omega^k_-(M)\oplus \bigoplus_{j=k+1}^{d-1}\Omega_\CC^j(M)}_{\g^-_\CC} .
\end{equation*}
Correspondingly, we split the fields $\As$ into its components $\As = \As^+ + \As^-$, and similarly for $\AsI$. The de Rham differential restricted to the subcomplex $
\Omega^{k-1}(M)\to \Omega^k_+(M) \oplus \Omega^k_-(M) \to \Omega^{k+1}(M)$ splits as $d_M = d_M^+ + d_M^-$ as in Section \ref{sec: 1d par ghosts complex},\footnote{ In more detail, for a $(p,q)$-form $\alpha$ with $p+q = k-1$ we have $$d_M^+\alpha =\begin{cases} d_M\alpha, & q<l, \\  \partial\alpha, & q=l, \\ 0, & q > l. 
\end{cases} $$ and vice versa for $d_M^-$.} see also Figure \ref{fig: pol 7d CS} below. 
\begin{figure}[H]
\includegraphics[scale=1]{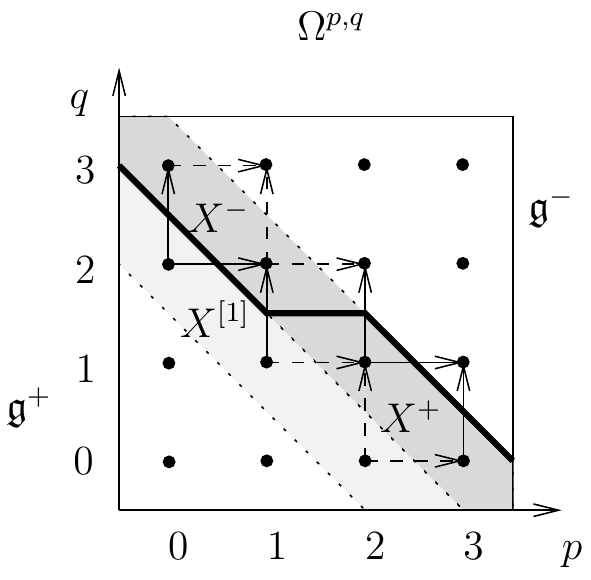}
\caption{Splitting of $\Omega^{p,q}$ into $\g^+$ (below the thick line) and $\g^-$ (above the thick line), in the case $k=3$. Solid arrows are components of $d^-$, dashed arrows are components of $d^+$. Horizontal arrows (dashed or not) are $\dd$, vertical ones are $\bar\dd$.}
\label{fig: pol 7d CS}\end{figure}

In particular, we have $(\As^+,d_M\As^+) = (\As^+,d_M^-\As^+)$. 
Before turning to the particular example of $k=3$, let us briefly have a look at the general form of the partition function in the two polarizations considered already in the last section. 
\subsection{Partition functions in $4l+3$-dimensional CS theories}\label{ss: higher dim CS transversal}
First, let us take $\g^+$ as the base of the polarization on both ends of the cylinder. This means that 
$$\BB = \g^+[k] \oplus \g^+[k] \ni (\As^+_{\rm{in}},\As^+_{\rm{out}})$$ with fiber 
$$\YY = \Omega^\bt(I,\partial I; \g^+[k]) \oplus \Omega^\bt(I,\g^-[k]).$$ Again, we can gauge fix the polarized theory by choosing 
$$ \mathcal{V} = dt \cdot \g^+[k] \oplus 1\cdot\g^-[k] \ni (dt\cdot \As^+_{I,\rm{res}}, \As^-_{\rm{res}})  $$
and using Hodge decomposition \eqref{1d ab CS hol-hol Hodge} with chain contraction \eqref{1d ab CS hol-hol K} and corresponding propagator \eqref{1d ab CS hol-hol eta}. Thus, we  obtain the splitting 
$$\mc{A} = \widetilde{\As}_{\rm{in}}^+ + \widetilde{\As}_{\rm{out}}^+ + dt\cdot \As^+_{I,\rm{res}} + \As^-_{\rm{res}} + \As^+_{\rm{fl}} + \As^-_{\rm{fl}}, $$ where 
$\widetilde{\As}_{\rm{in}}^+,\widetilde{\As}_{\rm{out}}^+$ are the discontinuous extensions of $\As^+_{\rm{in}},\As^+_{\rm{out}}$ into the bulk. In terms of this splitting, we can rewrite the ``perturbation'' $\frac12\int_{I\times M}\mc{A}\wedge d_M\mc{A}$ as 
$$
\frac12 \int_{I\times M}\mc{A}\wedge d_M\mc{A} = \int_I dt \int_M \As^+_{I,\res} d_M (\As^-_{\rm{res}} + \As^+_{\rm{fl}} + \As^-_{\rm{fl}}), $$
since $dt\cdot\As^+_{I,\res}$ is the only term containing a $dt$. By definition (see \eqref{1d ab CS hol-hol Hodge}) the $\As^-_\fl$ fluctuations have vanishing integral over $I$. Applying 
the fact that $d_M = d_M^+ + d_M^-$, we obtain 
$$
\frac12 \int_{ I \times M }\mc{A}\wedge d_M\mc{A} = \int_I dt\left( \int_M \As^+_{I,\res} d_M \As^-_{\rm{res}} + \int_M \As^+_{I,\res} d_M^- \As^+_\fl\right).$$
 The BV-BFV partition function - as in \eqref{ab CS hol-hol Z} - is then given by 
\begin{multline}
Z(\As^+_{\rm{in}},\As^-_{\rm{out}},\As^+_{I,\rm{res}}, \As^-_{\rm{res}}) =\\
=
\int_{\YY'_{K-ex}\subset \YY'}\mc{D} \As^+_\fl \;\mc{D} \As^-_\fl\;e^{\frac{i}{\hbar}S^\pol\left(
\til{\As^+_\ii}+\til{\As^+_\oo}+ \As^+_\fl+\As^-_\res+\As^-_\fl+dt\cdot\As^+_{I\,\res}
\right)}\\
=\int \mc{D}\As^+_\fl \;\mc{D} \As^-_\fl
\;e^{\frac{i}{\hbar}\big(
\int_{I\times M} \As^-_\fl  \dt \As^+_\fl+ \int_{\{1\}\times M} \As^+_\oo \As^- - \int_{\{0\}\times M} \As^+_\ii \As^- + \int_{I\times M} \frac12\A\, d_M \A \big)}
\\
=\int \mc{D} \As^+_\fl \;\mc{D} \As^-_\fl\;
\exp\frac{i}{\hbar} \Big(\int_{I\times M} \As^-_\fl \, \dt \As^+_\fl+ \int_{M} 
\As^+_\oo \, (\As^-_\res+\As^-_\fl\big|_{t=1}) -\\ -
 \int_{M} \As^+_\ii\, (\As^-_\res+\As^-_\fl\big|_{t=0})
 +
\int_M \As^+_{I\,\res} d_M \As^-_\res  + \int_{I\times M} dt\;
 \As^+_{I\,\res} \, d_M^- 
 \As^+_\fl
\Big).
\end{multline}
This is structurally the same formula as in Section \ref{s: 3dCS} before, 
with the difference
that the pairing on residual fields in slightly more complicated in this case. The Feynman diagrams defining this functional integral are the same as in \eqref{ab CS hol-hol Z answer} and yield 
\begin{equation*}
Z=\exp\frac{i}{\hbar}\int_M\Big(
(\As^+_\oo - \As^+_\ii)\; \As^-_\res+\As^+_{I\,\res}\,d_M \As^-_\res+\frac12 (\As^+_\oo+\As^+_\ii)\;d_M^- \As^+_{I\,\res}
\Big).
\end{equation*}
Similarly, the partition function in the holomorphic-to-antiholomorphic polarization, with space of boundary conditions 
$$\BB =\g^-[k]\oplus \g^+[k]\ni(\As^-_\ii,\As^+_\oo) $$ 
and space of residual fields 
$$\VV  = dt\cdot \g_{\CC}[k-1]\oplus (1-t)\cdot \g^+[k] \oplus t\cdot \g^-[k] \;\; \ni \;\; dt\cdot \As_\Ires + (1-t)\cdot \As^+_\res + t\cdot \As^-_\res ,$$ 
is 
\begin{align*}
Z(\As^-_\ii,\As^+_\oo&;\As_\Ires, \As^+_\res,\As^-_\res) \\
&=\exp\frac{i}{\hbar}\Bigg(\int_M \ -\As^+_\oo\As^-_\ii+\As^+_\oo \As^-_\res -\As^-_\ii\As^+_\res +\frac12 \As^-_\res \As^+_\res +\\
&+\frac12 \int_M 
\As_\Ires^+ \; (d_M^- \As^+_\res+d_M\As^-_\res) + \frac12\int_M\As_\Ires^-d_M\As_\res^+
\Bigg).\end{align*} 
\subsection{Parallel ghost polarization}\label{ss: higher dim CS parallel}
We can also choose again the ``parallel ghost'' polarization discussed in Section \ref{sec: 1d par ghosts complex}. To be more explicit, and in preparation for the next section, let us fix $k=3$. Then we have $\g = \Omega^\bullet(M,\CC)$ where $M$ is a 6-dimensional manifold with a complex structure, and $X = \g[3]$.\footnote{We use freely the notation from Section \ref{sec: 1d par ghosts complex}.} For later purposes, let us suppose that $M$ is endowed with a K\"ahler metric $g$. We use the complex structure to define the polarization of the $\gh = 0$ component of $X$: 
\begin{equation*}
X^{[0]} = \Omega^3_\CC(M) = \underbrace{\Omega^{3,0}(M) \oplus \Omega^{2,1}(M)}_{X^+} \oplus \underbrace{\Omega^{1,2}(M) \oplus \Omega^{0,3}(M)}_{X^-}  .
\end{equation*} 
Correspondingly, we split the field $\As = \As^{\leq 2} + \As^{3,+} + \As^{3,-}  + \As^{\geq 4}$ and similarly for $\As_I$. We denote by $\PP^{[<0],-}$ the polarization given by 
\begin{equation*}
\PP^{[<0],-} = \left\lbrace\frac{\delta}{\delta \As^{\geq 4}},\frac{\delta}{\delta \As^{3,-}}\right\rbrace ,
\end{equation*}
whose base is parametrized by $(\As^{\leq 2}, \As^{3,+})$, and by $\PP^{[<0],+}$ the similar polarization with $-$ and $+$ exchanged. The maps $d_\g^+,d_\g^-\colon X^{[1]} \to X^{[0]}$ defined in Section \ref{sec: 1d par ghosts nonlinear} are given by projecting the de Rham differential $$d_M \colon \Omega^2_\CC(M) \to \Omega^3_\CC(M)$$ to $X^\pm$. Explicitly, they are given by 
\begin{align*}
d^+_\g\big|_{\Omega^{2,0}} &= d_M\big|_{\Omega^{2,0}}, \qquad &d^+_\g\big|_{\Omega^{1,1}} = \partial_M\big|_{\Omega^{1,1}},\qquad &d^+_\g\big|_{\Omega^{0,2}} =0 , \\
d^-_\g\big|_{\Omega^{2,0}} &=0, \qquad &d^-_\g\big|_{\Omega^{1,1}} = \bar\partial_M\big|_{\Omega^{1,1}},\qquad &d^-_\g\big|_{\Omega^{0,2}} =d_M\big|_{\Omega^{0,2}}  ,
\end{align*}
see Figure \ref{fig: pol 7d CS}.
Now, we consider the cylinder $I\times M$ with  with $\PP^{[<0],+}$ on the in-boundary and  $\PP^{[<0],-}$ on the out-boundary. 
We then have the following fields in the effective action (referring to notation from Section \ref{sec: 1d par ghosts complex}): 
\begin{itemize}
\item $\psi_\oo^+ = \As^{3,+}_\oo= \As^{3,0}_\oo + \As^{2,1}_\oo$ --- physical boundary field on out-boundary,
\item $\psi^-_\ii =\As^{3,-}_\ii= \As^{0,3}_\ii + \As^{1,2}_\ii$  --- physical boundary field on in-boundary,
\item $(\psi^{[>0]}_\ii,\psi^{[>0]}_\oo) = (\As_\ii^{\leq  2},\As_\oo^{\leq 2})$ --- boundary fields in higher ghost number (collected in a superfield),
\item $A^{[1]}_\res = \As^2_\Ires = \As_\Ires^{2,0}+\As_\Ires^{1,1}+\As_\Ires^{0,2}$ --- 2-form,  residual field in ghost number 0,
\item $(A^{[>1]}_\res,\psi_\res)=(\As^{<2}_\Ires,\As_\res^{>3})$ --- residual fields of higher ghost number (form degree $<2$) or negative ghost number (form degree $>3$).
\end{itemize}
The effective action \eqref{eq: S eff 1d CS par} then reads 
\[ S_\eff[\As_\oo,\As_\ii,\As_\Ires,\As_{\res}] = S_\ph + S_\gh , \]
where
 \begin{align*}
S_\ph &= \int_M \As^{3,0}_\oo\As^{0,3}_\ii + \As^{2,1}_\oo\As^{1,2}_\ii
 + \As^{3,0}_\oo\bar\dd\As_\Ires^{0,2} + \As^{0,3}_\ii\dd\As_\Ires^{2,0}\\
 &+\int_M\As^{1,2}_\ii\partial\As_\Ires^{1,1}  +\As^{1,2}_\ii\bar\dd\As_\Ires^{2,0}+\As^{2,1}_\oo\dd\As_\Ires^{0,2} + \As^{2,1}_\oo\bar\dd \As_\Ires^{1,1} \\
 &+\frac12\int_M\bar{\partial}\As_\Ires^{1,1}\partial\As_\Ires^{1,1} , \\
 S_\gh &= \int_M(\As^{\leq 2}_\oo - \As^{\leq 2}_\ii)\As^{>3}_\res +\As_\res^{>3}d_M\As_\Ires^{<2}. 
\end{align*}
A similar formula holds for higher-dimensional Chern--Simons theories. Some comments: 
\begin{enumerate}[i)]
\item  As a consequence of Propositions \ref{prop: 1d par ghosts mQME} and \ref{prop: 1d par ghosts HJ}, the effective action satisfies the modified quantum master equation, and the $\gh=0$ part $S_\ph$ satisfies the generalized Hamilton--Jacobi equations. In particular, $S_\ph$ can be identified with the HJ action. 
\item One can rewrite $S_\ph$ as 
\begin{equation*}
S_\ph = \int_M (\As^{3,+}_\oo - d_\g^+\As^2_\Ires)(A^{3,-}_\ii + d_\g^-\As_\Ires^2) -\frac12 \bar\dd A^{1,1}_\Ires\dd A^{1,1}_\Ires
\end{equation*}
--- a higher-dimensional version of an abelian gauged WZW model, see also footnote \ref{fn:abWZW}. 
\end{enumerate}
\subsubsection{Pushforward over residual fields}
The space of residual fields is 
\begin{equation*}
\VV =\{ dt\cdot \As^{\leq 2}_\Ires + \As_\res^{\geq 4}\} \;\;  =
 \Omega^{\leq 2}_\CC(M)[2] \oplus \Omega_\CC^{\geq 4}(M)[3].
\end{equation*}
In particular, the components and their ghost numbers are 
\begin{center}
\begin{tabular}{c|c|c|c|c|c|c}
field & $\As^0_\Ires$ & $\As^1_\Ires $ & $ \As^2_\Ires $ & $\As^4_\res$ & $\As^5
_\res $ & $\As^6_\res$ \\
\hline
ghost number & 2 & 1 & 0 & -1 & -2 & -3 \end{tabular}
\end{center}
A gauge-fixing lagrangian can be constructed by using the Hodge decomposition for $\bar\partial$: Namely, using the K\"ahler metric $g$, we decompose 
\begin{equation*}
\Omega^{p,q}(M) = H^{p,q}(M) \oplus \Omega^{p,q}_{\bar\dd-\mr{ex}} \oplus \Omega^{p,q}_{\bar\dd-\mr{coex}}
\end{equation*}
where the middle and rightmost terms denote the spaces of $\bar\dd$-exact and $\bar\dd^*$-exact forms, respectively. The gauge-fixing lagrangian $\mathcal{L}\subset \VV$  is then defined as 
\begin{equation}
\LL =\bigoplus_{p+q \leq 2} \Omega^{p,q}_{\bar\dd-\mr{coex}} \oplus \bigoplus_{p+q \geq 4}\left( H^{p,q}(M) \oplus \Omega^{p,q}_{\bar\dd-\mr{coex}} \right).  \label{eq:GFlag7d} 
\end{equation}
Restricted to this gauge-fixing lagrangian, the effective action is nondegenerate in residual fields and the integral gives
\begin{multline*}
Z_*\propto 
\delta(\bar\dd \As^{3,0}_\oo+\dd \As^{2,1}_{\oo,\bar\dd\mr{-ex}})
\,
\delta(\bar\dd \As^{1,2}_\ii+\dd \As^{0,3}_{\ii,\bar\dd\mr{-ex}})
\, \delta(\As^{\leq 2}_{\oo,\mr{harm}}-\As^{\leq 2}_{\ii,\mr{harm}})\cdot\\
\cdot\exp \frac{i}{\hbar}\bigg( \int_M  \big(\As^{3,0}_\oo \As^{0,3}_\ii + \As^{2,1}_\oo (\mr{id}-P_{\bar\dd\mr{-ex}}) \As^{1,2}_\ii \big)-\\
-\frac12 \int_{M\times M}\bar\dd \As^{2,1}_\oo(x)K(x,x')\bar\dd \As^{2,1}_\oo(x') -\frac12 
\int_{M\times M}\dd \As^{1,2}_\ii(x)K(x,x')\dd \As^{1,2}_\ii(x')
\bigg),
\end{multline*}
where $K(x,x')$ is the integral kernel of the inverse of the operator $(\dd\bar\dd)$ restricted to $\Omega^{1,1}_{\bar\dd-\mr{coex}}$.
\subsection{7D Chern--Simons and Kodaira--Spencer action functional}\label{s:7dCSKS} 
We now turn our attention to 7-dimensional Chern--Simons theory on a cylinder with a particular polarization on the out-boundary. This polarization was first discovered by Hitchin \cite{Hitchin}. It was used in \cite{GS} to argue that the semi-classical approximation of the Chern--Simons wave function can be expressed in terms of the Kodaira--Spencer action functional introduced in \cite{BCOV} whose classical solutions are deformations of  complex structures on a K\"ahler manifold (see Appendix \ref{app:KS} for a brief review of the Kodaira-Spencer theory). Here we argue that this semiclassical approximation is in fact exact in the axial gauge. From the general arguments of the BV-BFV formalism, it follows that a change of gauge fixing will result in an $\Omega$-exact change of the partition function, hence its $\Omega$-cohomology class is well-defined and given by the Kodaira--Spencer partition function. A caveat is that in this section we do not take care of determinants arising in Gaussian path integrals. Those might lead to anomalies similar to the discussion of Remark \ref{rem: Liouville}, and would have to be treated separately. 

\subsubsection{General polarizations in $4l+3$-dimensional Chern--Simons theory}

Using the results of Section \ref{sec: 1d par ghosts nonlinear}, one can consider also more general polarizations in higher-dimensional Chern--Simons theories, in dimension $d = 2k+1 = 4l+3$. 

Suppose that $\PP^{[0]}$ is any polarization on $X^{[0]} = \Omega^k_\CC(M)$ such that we have local coordinates $\As^Q$ on the base and $\As^P$ on the fibers, and let $G = G(\As^{-},\As^Q)$ be the corresponding generating function. From Section \ref{sec: 1d par ghosts nonlinear} we know that the partition function of abelian Chern--Simons theory with $\PP^{[<0],+}$-polarization on the in-boundary and $\PP^{[<0],P}$-polarization on the out-boundary is $Z = \exp(\frac{i}{\hbar} S_\eff)$ with $S_\eff = S_\ph + S_\gh$ and 
\begin{align}
S_\ph[\As^-_\ii,\As^Q_\oo,\As_\Ires] &=  \frac{1}{2}\int_M \partial_M \As^{l,l}_\Ires\bar\dd_M \As^{l,l}_\Ires + \int_M \As^{k,-}_\ii d_\g^+\As^{k-1}_\Ires \notag
\\&\qquad - G(\As^{k,-}_\ii + d_\g^-\As^{k-1}_\Ires,\As^Q_\oo ) , \notag\\
S_\gh[\As^{[>0]}_\ii,\As^{[>0]}_\oo,\As_\Ires,\As_\res] &= \int_M(\As^{[>0]}_\oo - \As^{[>0]}_\ii)\As^{[<0]}_\res + \int_M\As_\res^{[<0]}d_M\As_\Ires^{[>0]}.\label{eq:S eff gh 7d}
\end{align}
See also the toy model in considered \cite[Section 12]{HJ}.
\subsubsection{Hitchin polarization on 6-dimensional manifolds and effective action}\label{sss: Hitichin}
In 7-dimensional Chern--Simons theory, there is an interesting --- nonlinear --- polarization on the boundary phase space, coming from the special geometry of three-forms in six dimensions first described by Hitchin in \cite{Hitchin}.
The idea is as follows. A complex 3-form $A$ on a six-dimensional manifold with a complex structure decomposes as  
\begin{equation} A = A^{+,\nl} + A^{-,\nl} \label{eq:defHitchinPol}
\end{equation}
where $A^{+,\nl}$ and $A^{-,\nl}$ are decomposable complex 3-forms, i.e.,  triple wedge products:
$$A^{\pm,\nl} = \theta_1^\pm \wedge \theta_2^\pm \wedge \theta_3^\pm, \qquad \theta_i^\pm \in \Omega^1(M,\CC).$$ The $3$-form $A$ is called \emph{nondegenerate} if $A^{+,\nl} \wedge A^{-,\nl}$ is everywhere nonvanishing (which is equivalent to the fact that the form $A$ is not decomposable). In this case $A^{+,\nl}$ and $A^{-,\nl}$ are uniquely determined by $A$ and define a polarization of $\Omega^3(M)_{\mr{nd}}$, the subset of nondegenerate forms. For more details on this polarization, we refer to \cite{GS}.
The effective action on the cylinder with $\PP^{[<0],-,\mr{l}}$-polarization on the in-boundary\footnote{For consistency with \cite{GS} and \cite{BCOV}, we switch here to the polarization with base parametrized by $\As^{3,+,\mr{l}}$, as opposed to the rest of the paper.} and $\PP^{[<0],+,\nl}$-polarization on the out-boundary thus reads $S_\eff = S_\ph + S_\gh$ with $S_\gh$ given by \eqref{eq:S eff gh 7d} and  
\begin{multline}
S_\ph[\As^{+,\mr{l}}_\ii,\As^{-,\nl}_\oo,\As^{2,0}_\Ires,\As^{1,1}_\Ires,\As^{0,2}_\Ires] \\ = \frac12 \int_M \dd \As^{1,1}_\Ires \bar\dd \As^{1,1}_\Ires  + \int_M \As^{+,\mr{l}}_\ii d\As^{0,2}_\Ires + \As^{2,1}_\ii\bar\dd \As^{1,1}_\Ires  
- G(\As^{+,\l}_\ii +  d\As^{2,0}_\Ires + \dd\As^{1,1}_\Ires;\As^{-,\nl}_\oo).   \label{eq:SeffCS7d}
\end{multline}
\begin{proposition}
The partition function is  given by 
\begin{equation*}
Z = \exp\frac{i}{\hbar}\left(S_{\ph} + S_{\gh}\right),
\end{equation*}
with $S_\ph$ given by \eqref{eq:SeffCS7d} and $S_\gh$ given by \eqref{eq:S eff gh 7d}. It satisfies the modified quantum master equation 
\begin{equation*}
( \Omega_\ii + \Omega_\oo - \hbar^2\Delta_\res)Z = 0, 
\end{equation*}
with $\Omega_\ii$, $\Omega_\oo$ given by the standard quantization of the BFV boundary action $\frac12\int_M \As d\As$: 
\begin{align*}
\Omega_\ii &= \int_M d\As_\ii^{2}\left(\As^{+,\mr{l}}_\ii - i\hbar\frac{\delta}{\delta \As_\ii^{+,\mr{l}}}\right) -i\hbar \int_M d\As_\ii^{1}\frac{\delta}{\delta\As_\ii^{4}}-i\hbar \int_M d\As_\ii^{0}\frac{\delta}{\delta\As_\ii^{5}},  \\
\Omega_\oo &= \int_M d\As_\oo^{2}\left(\As^{-,\mr{nl}}_\oo - i\hbar\frac{\delta}{\delta \As^{-,\mr{nl}}_\oo}\right) -i\hbar \int_M d\As_\oo^{1}\frac{\delta}{\delta\As_\oo^{4}}-i\hbar \int_M d\As_\oo^{0}\frac{\delta}{\delta\As_\oo^{5}}. 
\end{align*}
\begin{proof}
 We are in the situation of Remark \ref{rem:mQMEnonlin} here because of the fact that the boundary action is linear in the canonical variables defining the Hitchin polarization, see equation \eqref{eq:defHitchinPol}. It follows that the mQME is satisfied for the standard quantization of the boundary action. 
\end{proof}
\end{proposition}
\begin{remark}\label{rem: 7d CS exact} We want to stress that in this gauge, the semiclassical approximation to this partition function that was used as an ansatz in \cite{GS} (the integral kernel of the generalized Segal--Bargmann transform\footnote{See Appendix \ref{app:SB} for a motivation for the comparison of BV-BFV partition function with standard Segal--Bargmann transform.} from the linear to the nonlinear polarization) is found to be \emph{exact}: there are no quantum corrections.  Adapting the proof of \cite[Appendix B]{HJ} to the infinite-dimensional setting one can show that changing the gauge fixing by changing the propagator on the interval
results in a change of the partition function by a $(\hbar^2\Delta_\res - \Omega)$-cocycle.
\end{remark}
\begin{remark}
To be precise in comparison with \cite{GS}, one should identify the residual fields in  \eqref{eq:SeffCS7d} with the Lagrange multipliers enforcing the constraint $d_M\As=0$ in \cite{GS}. This, of course, is precisely their role when interpreting \eqref{eq:SeffCS7d} as the Hamilton--Jacobi action for 7D Chern--Simons theory in the chosen polarizations. 
\end{remark}

\subsubsection{Comparison with Kodaira--Spencer gravity} \label{sss: CS vs KS}
Following \cite{GS}, we want to compare the Chern--Simons effective action \eqref{eq:SeffCS7d} with the Kodaira--Spencer action functional \eqref{eq:KSaction2}. 
Let us fix a reference holomorphic 3-form $\omega_0\in\Omega^{3,0}(M)$. Its conjugate is an antiholomorphic 3-form $\overline{\omega}_0 \in \Omega^{0,3}(M)$. We can then parametrize $A^{+,\nl}$ and $A^{-,\nl}$ as 
\begin{align*}
A^{+,\nl} &= \rho e^\mu \omega_0, \\
A^{-,\nl} &= \overline{\rho} e^{\overline{\mu}}\overline{\omega}_0,
\end{align*}
where $\rho,\ol{\rho} \in \Omega^0_\CC(M)$, $\mu \in \Omega^{-1,1}(M)$, $\ol{\mu}\in \Omega^{1,-1}(M)$ and 
\begin{equation} \rho e^\mu\omega_0 = \rho\left(\omega_0 + \mu \omega_0 + \frac{\mu^2}{2}\omega_0 + \frac{\mu^3}{6}\omega_0 \right) ,
\label{eq:nlpoldef_maintext}
\end{equation}
where $\mu\omega_0$ should be interpreted as extension of contraction to forms with values in vector fields. 

Of course, a complex 3-form still has a decomposition $A = A^{+,\mr{l}} + A^{-,\mr{l}}$, with $A^{+,\mr{l}} \in \Omega^{3,0}(M) \oplus \Omega^{2,1}(M)$, and $A^{-,\mr{l}} \in \Omega^{1,2}(M) \oplus \Omega^{0,3}(M)$.

The following expression for $G$ is given in \cite{GS}: 
\begin{multline}
G(A^{3,0},A^{2,1},\overline{\rho},\overline{\mu}) =\\
=\int_M \overline{\rho}( A^{3,0}\overline{\omega}_0 + A^{2,1}\overline{\mu}\,\overline{\omega}_0) + \overline{\rho}^2\langle\overline{\mu}^3\rangle\omega_0\overline{\omega}_0 -
\frac{\left\langle\left( (A^{2,1} - \frac12\overline{\rho}\,\overline{\mu}^2\overline{\omega}_0)^\vee\right)^3\right\rangle}{(A^{3,0})^\vee - \overline{\rho}\langle\overline{\mu}^3\rangle}\omega_0\overline{\omega}_0 . \label{eq:genfctmain}
\end{multline}
For completeness, we present a derivation of this formula in Appendix \ref{app:genfctHitchin}.
Here for $\overline{\mu} \in \Omega^{1,-1}(M)$ the function $\langle \mu^3 \rangle$ is defined in \eqref{eq:mucube} and for any $A \in \Omega^{p,q}$ we 
define $A^\vee \in \Omega^{p-3,q}$ by $A^\vee \omega_0 =A$.

%
Consider now the state 
\begin{equation}
\psi\left(\overline{\rho}_\oo,\overline{\mu}_\oo,\overline{\As^{[>0]}_\oo}\right) = \delta(\overline{\mu}_\oo)\delta\left(\As^{[>0]}_\oo\right)\exp\frac{i}{\hbar}\int_M\overline{\rho}_\oo\omega_0\overline{\omega}_0. \label{eq:weirdstate}
\end{equation}
This is an extension of the physical state\footnote{This state $\psi_\ph$  was proposed in \cite{GS} as a way of ``fixing the string coupling constant.'' It is a quantization of the lagrangian given by $\ol\mu=0$, $p_{\ol\rho}=1$, where $p_{\ol\rho}$ denotes the canonical momentum of $\ol\rho$.} $$\psi_\ph(\ol\mu,\ol\rho) = \delta(\overline{\mu}_\oo)\exp\frac{i}{\hbar}\int_M\overline{\rho}_\oo\omega_0\overline{\omega}_0$$ to the ghost sector. Note that this state trivially satisfies the mQME because $\Omega_\oo\psi_\ph$ is linear in the ghost fields, but we are multiplying with the ghost delta function.\footnote{This is just the identity $x \delta(x) = 0$. }
Now  we compute (formally) the vector $Z\left|\psi\right\rangle$ (put differently, we are specifying a boundary condition on the $\oo$-boundary for the quantum theory).  
%
%
This means computing the following functional integral:
\begin{align*}
 Z \left|\psi\right\rangle &= \int\DD\overline{\rho}_\oo\DD\overline{\mu}_\oo\DD\As_\oo^{[>0]} Z \cdot \psi \\ 
&= \int\DD\overline{\rho}_\oo Z[\As^{+,\mr{l}}_\ii,\As^{[>0]}_\ii,\overline{\rho}_\oo,0,0,\As^{\geq 4}_\res, \As^{\leq 2}_\Ires]\exp\frac{i}{\hbar}\int_M\overline{\rho}_\oo\omega_0\overline{\omega}_0 \\
&=: Z'[\As^{+,\mr{l}}_\ii,\As^{[>0]}_\ii,\As^{\geq 4}_\res, \As^{\leq 2}_\Ires] .
\end{align*}
The partition function $Z$ depends on $\overline{\rho}_\oo$ only through $G$ and  we have 
\begin{multline}
\int\DD\overline{\rho}\;e^{-\frac{i}{\hbar}G(A^{3,0},A^{2,1},\overline{\rho},0)+\frac{i}{\hbar}\int_M\overline{\rho}\omega_0\overline{\omega}_0} =\\
=\delta(A^{3,0}-\omega_0)\, \exp\left(-\frac{i}{\hbar} \int_M \frac16 \langle A^{2,1}, A^{2,1},A^{2,1}\rangle \right) . \label{eq:Gintegral}
\end{multline}
Here for $A \in \Omega^{2,1}(M)$ we have $\frac16\langle A,A,A\rangle := \langle (A^\vee)^3\rangle \omega_0\overline{\omega}_0$. Thus, $Z'$ has the following expression: 
\begin{align}
Z' &= Z'_\ph Z'_\gh,  \quad \mbox{where} 
\notag \\
Z'_\ph 
&= \delta(\As^{3,0}_\ii + \dd \As^{2,0}_\Ires - \omega_0) \,
\exp\frac{i}{\hbar} \bigg(\frac12 \int_M \dd \As^{1,1}_\Ires \bar\dd \As^{1,1}_\Ires  + \int_M \As^{3,0}_\ii\bar\dd \As^{0,2}_\Ires \label{eq:Zprime2}\\
&+\int_M \As^{2,1}_\ii\dd \As^{0,2}_\Ires + \As^{2,1}_\ii\bar\dd \As^{1,1}_\Ires - \langle( (\As_\ii^{2,1}+ \bar\dd \As^{2,0}_\Ires + \dd \As^{1,1}_\Ires)^\vee)^3\rangle\omega_0\overline{\omega}_0\bigg),\notag\\
Z'_\gh &= \exp\frac{i}{\hbar}\left(\int_M -\As^{\leq 2}_\ii \As^{\geq 4}_\res + \int_M \As^{\geq 4}_\res d_M\As^{\leq 1}_\Ires\right).\label{eq:Zprime3}
\end{align} 
We stress that $Z'$ was obtained from $Z$ through a formal functional integral. 
However, we have the following result.
\begin{lemma}
The function $Z'[\As^{+,\mr{l}}_\ii,\As^{[>0]}_\ii,\As^{\geq 4}_\res,\As^{\leq 2}_\Ires]$ satisfies the modified quantum master equation, i.e., is an $(\Omega_\ii - \hbar^2\Delta_\res)$-cocycle, where $\Omega_\ii = \Omega^{(0)}_\ii + \Omega^{(1)}_\ii$ is the standard quantization of $-\frac12\int_M \As d_M \As$ in the $\PP^{[<0],-}$-polarization, explicitly given by 
\begin{align*}
\Omega^{(0)}_\ii &= -\int_M \As^{2,1}_\ii \bar\dd\As^{1,1}_\ii + \As^{2,1}_\ii\dd \As_\ii^{0,2} + \As_\ii^{3,0}\bar{\dd}\As_\ii^{0,2}, \\ 
\frac{i}{\hbar}\Omega^{(1)}_\ii &=  -\int_M \frac{\delta}{\delta \As^{3,0}_\ii}\dd \As^{2,0}_\ii  +  \frac{\delta}{\delta \As^{2,1}_\ii}\bar\dd\As^{2,0}_\ii + \frac{\delta}{\delta \As^{2,1}_\ii}\dd \As^{1,1}_\ii + \frac{\delta}{\delta \As^{\leq 2}_\ii}d\As^{\leq 1}_\ii .
\end{align*}
\end{lemma}
This is expected because $\psi$ satisfies $\Omega_\oo\psi = 0$ and $Z$ satisfies the mQME. See also the discussions of gluing in \cite[Sections 11.4,12.2]{HJ}.
The interpretation of this Lemma is that $Z'$ is a valid state in the  linear polarization. 
\begin{proof}
We will only prove the claim in ghost number 0, since in positive ghost number the effective action is the same as for linear polarizations. To begin, we note that any function of $\As^{3,0}_\ii + \dd \As^{2,0}_\Ires$ or $\As_\ii^{2,1} + \bar\dd\As^{2,0}_\Ires + \dd\As^{1,1}_\Ires$ is $( 
i \hbar\{S_\gh,\bullet\}_\res - \Omega^{(1)}_\ii)$-closed since $\{S_\gh,\bullet\}_\res\big|_{\gh =0} = \int_M\As^{i,j}_\ii\frac{\delta}{\delta \As^{i,j}_\Ires}$.
This implies that the delta function in \eqref{eq:Zprime2} and the last term in  \eqref{eq:Zprime3} are $\hbar^2\Delta_\res - \Omega^{(1)}$ closed. It is a straightforward check that the remaining exponential terms are $(\hbar^2\Delta_\res - \Omega^{(0)} - \Omega^{(1)})$-closed, which concludes the proof. \end{proof}
We will now argue that formally integrating out the residual fields, in ghost number 0 we obtain the Kodaira--Spencer action.  Let us restrict to the gauge-fixing lagrangian $\LL$ defined similarly to \eqref{eq:GFlag7d}, but given in ghost number 0 by $\dd^*$-exact 2-forms.  We will denote 
\begin{equation*}
Z''[\As^{3,0}_\ii,\As^{2,1}_\ii,\As^{[>0]}_\ii] = \int_\LL Z'[\As^{+,\mr{l}}_\ii,\As^{[>0]}_\ii,\As^{\geq 4}_\res,\As^{\leq 2}_\Ires] .
\end{equation*}
The modified quantum master equation implies that for a $(2,0)$-form $\chi$ one has
$$ Z'[\As^{+,\mr{l}}_\ii +d \chi,\As^{[>0]}_\ii,\As^{\geq 4}_\res,\As^{\leq 2}_\Ires]  = Z'[\As^{+,\mr{l}}_\ii 
,\As^{[>0]}_\ii,\As^{\geq 4}_\res,\As^{\leq 2}_\Ires - \chi].$$
By a change of variables, this implies 
\begin{equation*}
Z''[\As^{3,0}_\ii +\dd\chi,\As^{2,1}_\ii+\bar\dd\chi,\As^{[>0]}_\ii] = Z''[\As^{3,0}_\ii,\As^{2,1}_\ii,\As^{[>0]}_\ii] .
\end{equation*}
 We can use this property to 
 reduce the computation of $Z''$ to the case
 $\As^{3,0} = \rho_0\omega_0$, where $\rho_0$ is a constant. The $\delta$ function in $Z'_\ph$ then factorizes as $\delta(\rho_0 - 1)\delta( \dd A^{2,0}_\Ires)$. Since $\dd$ is an isomorphism on the gauge-fixing lagrangian, the integral over $\As^{2,0}_\Ires$ gives
 \begin{multline*} Z''_\ph = \int\DD A^{0,2}_\Ires A^{1,1}_\Ires \delta(\rho_0 -1) \exp\frac{i}{\hbar} \bigg(-\frac12 \int_M \dd \As^{1,1}_\Ires \bar\dd \As^{1,1}_\Ires  +\\
+\int_M \As^{2,1}_\ii\dd \As^{0,2}_\Ires + \As^{2,1}_\ii\bar\dd \As^{1,1}_\Ires  
- \langle( (\As_\ii^{2,1} + \dd \As^{1,1}_\Ires)^\vee)^3\rangle\omega_0\overline{\omega}_0\bigg) 
\end{multline*}
and the integral over $A^{0,2}_\Ires$ then gives
\begin{multline*}
Z''_\ph = \int\DD \As^{1,1}_\Ires \delta(\rho_0 -1)\delta(\dd \As_\ii^{2,1})\cdot\\
\cdot \exp\frac{i}{\hbar} \bigg( \int_M \frac12\dd \As^{1,1}_\Ires \bar\dd \As^{1,1}_\Ires  
+ \As^{2,1}_\ii\bar\dd \As^{1,1}_\Ires -\langle( (\As_\ii^{2,1} + \dd \As^{1,1}_\Ires)^\vee)^3\rangle\omega_0\overline{\omega}_0\bigg).
\end{multline*} 
Finally, writing $\As^{2,1}_\ii$ in the Hodge decomposition $\As^{2,1}_\ii = x + \dd\lambda + \dd^*\tau$, we obtain by another change of variables\footnote{We can choose $\lambda$ such that $\dd^*\lambda = 0$, so that $b \in \LL \cap \Omega^{1,1}(M)$.} $b = \As^{1,1}_\Ires + \lambda$ the expression
\begin{multline}
Z''_\ph[\rho_0,x,\lambda,\tau] = \delta(\rho_0 -1)\delta(\dd\dd^*\tau)\cdot \\
\cdot \int_{\LL \cap \Omega^{1,1}(M)}\DD b\;\exp\frac{i}{\hbar}\int_M\left(-\frac12\dd\lambda\bar\dd\lambda +\frac12\dd b \bar\dd b+ \bar\dd b \dd\lambda + \frac16\langle (x+\dd b), (x + \dd b), (x + \dd b) \rangle\right) , \label{eq:finalZprimeprime}
\end{multline}
which coincides with eq.\ (2.50) in \cite{GS}. 
Thus, we see that the Chern--Simons partition function on a cylinder, paired with the state \eqref{eq:weirdstate}, coincides with the partition function of Kodaira--Spencer theory with background $x$ and action functional given in \eqref{eq:KSaction2} for $\lambda = 0$. The latter integral can be evaluated perturbatively in terms of Feynman graphs and rules. It would be interesting to compare our results to other constructions of the BCOV theory, such as in \cite{CL}.

\begin{remark}[On gauge invariance of $Z''$]
If one uses formally the properties of the BV integral, it is immediate that the $Z''$ gives a class in $\Omega_\ii$-cohomology independent of the gauge-fixing lagrangian $\LL$.\footnote{To give a rigorous proof of this statement would require to give a strict interpretation of the functional integral in \eqref{eq:finalZprimeprime}, which is beyond the scope of the present paper. It should be noted that the restriction of $b$ to the subspace $\dd^*b = 0$ is a gauge-fixing condition for the KS theory, so one should consider also the the gauge independence of $Z_{KS}$.}  We see here that this cohomology class has a representative given in terms of the KS partition function. The partition function $Z''$ can be also interpreted as the BV-BFV partition function on the cylinder paired with the state $\psi$ at the $\oo$-boundary, with \emph{all} fields integrated out using an axial-type gauge (the components of the gauge field involving $dt$ are set to zero). Another open question is how $Z''$ behaves when we deform away from this type of gauge to a general gauge fixing on the cylinder (say, one given by a Riemannian metric). This is a subject of ongoing research. \end{remark}
%

\appendix

\section{Segal--Bargmann transform via BV-BFV}\label{app:SB}
Recall (see \cite{Hall} for details) that the Segal--Bargmann space $\mc{H}^\mr{SB}$ is the Hilbert space of holomorphic functions $\psi(z)$ on $\CC$ satisfying
$$\int_\CC \frac{i}{4\pi\hbar} dz\,d\bar{z} \; e^{-\frac{|z|^2}{2\hbar}} \;\overline{\psi(z)}\psi(z)<\infty$$
(here we assume that $\hbar$ is a fixed positive number), equipped with inner product
\begin{equation}\label{SB inner product}
\langle \psi_1,\psi_2 \rangle =\int_\CC \frac{i}{4\pi\hbar } dz\,d\bar{z} \; e^{-\frac{|z|^2}{2\hbar}} \;\overline{\psi_1(z)}\psi_2(z).
\end{equation}
The Segal--Bargmann space is isomorphic to the Hilbert space $L^2(\RR)$ of square-integrable functions on $\RR$, with the unitary isomorphism $L^2(\RR)\ra \mc{H}^\mr{SB}$ given by the Segal--Bargmann transform:
\begin{equation}\label{SB transform}
\chi(x)\mapsto 
\psi(z)=(\pi\hbar)^{-\frac14}\int_\RR dx\; e^{-\frac{1}{\hbar}\left(\frac{z^2}{4}-zx+\frac{x^2}{2}\right)}\; \chi(x).
\end{equation}

Now we would like to show how the transformation (\ref{SB transform}) can be seen as the partition function for topological quantum mechanics on an interval with appropriate boundary polarizations.

Consider topological quantum mechanics on the interval $I$ parametrized by $0\leq t\leq 1$ --- the theory with $0$-form fields $x,p\in \Omega^0(I)$ and action 
\begin{equation}\label{S=pdx}
S=\int_I p\, d x .
\end{equation}
In the BV-BFV formalism, we adjoin the anti-fields $x^*,p^*\in \Omega^1(I)$ --- $1$-form fields carrying ghost number $-1$ (while $x,p$ carry ghost number $0$), so that the odd symplectic form on BV fields is: $\int_I \delta x\wedge \delta x^* + \delta p\wedge \delta p^*$. The BFV phase space assigned to a point $\mr{pt}^\pm$ (where $\pm$ is the orientation) is: $\PhiPM=\mathbb{R}^2\; \ni (x,p)$ with the Noether $1$-form $\alpha_{\mr{pt}^\pm}=\pm p\,\delta x$ and vanishing BFV action $S_\mr{pt}=0$.

Alongside the real coordinates $x,p$ on the phase space, we will consider the complex coordinates $z=x-ip$, $\bar{z}=x+ip$.  The symplectic structure on the phase space is $\omega_{\mr{pt}^\pm}=\delta\alpha_{\mr{pt}^\pm}=\pm\delta p\wedge\delta x$. Written in complex coordinates it has the form $\mp\frac{i}{2}\delta\bar{z}\wedge \delta z$.

Consider the polarization $\mr{Span}\{\frac{\dd}{\dd p}\}$ (i.e., $x$ fixed) at $t=0$ and the polarization $\mr{Span}\{\frac{\dd}{\dd\bar{z}}\}$ (i.e., $z$ fixed) at $t=1$. The corresponding modification of the action (\ref{S=pdx}) by a boundary term is:
\begin{equation}
S^\pol=S\underbrace{-\big(\frac{i}{4}x^2+\frac12 xp +\frac{i}{4} p^2 \big)}_{f}\Big|_{t=1}
\end{equation}
--- this boundary term is chosen so that one has $-\frac{i}{2}\bar{z}\delta z= p\delta x+\delta f$. Thus, the corresponding boundary Noether $1$-form is:
\begin{equation}
\alpha_{\dd I}^\pol=-\frac{i}{2}\bar{z}\,\delta z\big|_{t=1} - p\,\delta x \big|_{t=0}
\end{equation}
--- it vanishes along the chosen polarization, as desired.

Next consider the following splitting of the (complexified) phase space
$$\PhiPM_\CC=\g^+\oplus\g^- , $$
where $\g^+$ is parametrized by $z$ and $\g^-$ is parametrized by $x$ (we are borrowing the notations from (\ref{g=g+ + g-}) 
here). The space of fields $\FF=\Omega^\bt(I,\g^+)\oplus \Omega^\bt(I,\g^-)$ is fibered over $\BB\ni (z_\oo,x_\ii)$ with the fiber
\begin{multline*}
\YY= \Omega^\bt(I,\{1\};\g^+)\oplus \Omega^\bt(I,\{0\};\g^-)=\\
=
\underbrace{\big(\Omega^0(I,\{1\};\g^+)\oplus \Omega^0(I,\{0\};\g^-)\big)}_{\YY'_{K-ex}}\bigoplus
\underbrace{\big(\Omega^1(I;\g^+[-1])\oplus \Omega^1(I;\g^-[-1])\big)}_{\YY'_{d-ex}} .
\end{multline*}
This is an acyclic complex, and thus we can choose the space of residual fields to be zero (cf.\ Remark \ref{rem: 1d ab cs ahol-hol full integral}). The corresponding propagator --- the integral kernel of the chain contraction $K$ --- is:
$$ \eta(t,t')= -\pi^+\otimes \theta(t'-t) + \pi^-\otimes \theta(t-t') . $$

The BV-BFV partition function is then given by the following path integral
\begin{multline}\label{Z SB}
Z(z_\oo,x_\ii)=\int \D z_\fl\;\D x_\fl\; e^{\frac{i}{\hbar}S^\pol(\til{z_\oo}+z_\fl,\til{x_\ii}+x_\fl)}\\
=\int \D z_\fl\;\D x_\fl\;e^{\frac{i}{\hbar} \Big(i\int_I z_\fl\,dx_\fl+\big(-ix_\fl(1)z_\oo+\frac{i}{4}z^2_\oo\big)
+\big(-iz_\fl(0)\,x_\ii +\frac{i}{2}x_\ii^2\big)
\Big) }\\
=e^{-\frac{1}{\hbar}\big(\frac{z^2_\oo}{4}-z_\oo x_\ii +\frac{x^2_\ii}{2}\big)} .
\end{multline}
Here we have a contribution from the Feynman diagram with single propagator connecting $z_\oo$ and $x_\ii$.  In (\ref{Z SB}) we recognize the integral kernel of the Segal--Bargmann transform (\ref{SB transform}). This is, of course, to be expected: the partition function for a cylinder (in this case, an interval), with polarization $\mc{P}_1$ on the in-boundary and polarization $\mc{P}_2$ on the out-boundary, maps $\mc{P}_1$-states to $\mc{P}_2$-states.

\begin{remark}
Note that the Hamilton--Jacobi action for the theory (\ref{S=pdx}) on the interval with our choice of in/out polarizations is: 
$$S_\text{HJ}=S^\pol\big(x(t)=x_\ii,z(t)=z_\oo\big)\quad = \quad 
\frac{i}{4}z_\oo^2-iz_\oo x_\ii +\frac{i}{2} x_\ii^2 .
$$
Here we again recognize the expression in the exponent in (\ref{SB transform}).
\end{remark}

Finally, the measure $e^{-\frac{|z|^2}{2\hbar}}$ in (\ref{SB inner product}) from the BV-BFV standpoint originates from the gluing of intervals (more precisely, from gluing the out-end of one interval with $z$-fixed polarization to the in-end of another interval with $\bar{z}$-fixed polarization). Indeed, consider the theory (\ref{S=pdx}) on the interval $I=[t_0,t_1]$ with some polarization $\mc{P}$ at $t_0$ and $z$-fixed polarization at $t_1$, and also the same theory on the interval $I'=[t_1,t_2]$, with some polarization $\mc{P}'$ at $t_2$ and with $\bar{z}$-fixed polarization at $t_1$. 
$$
\includegraphics[scale=0.7]{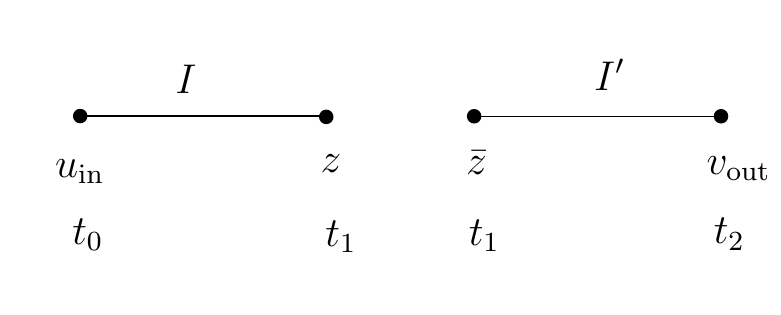}
$$
The respective actions including the boundary terms adjusting for the polarization are
$$ S^\pol_{I}=-\frac{i}{2}\int_{I} \bar{z} dz + f^\mc{P}\big|_{t_0}\quad,\quad
S^\pol_{I'}=-\frac{i}{2}\int_{I'} \bar{z} dz -\frac{i}{2} z\bar{z}\big|_{t_1} + f^\mc{P'}\big|_{t_2} ,
$$
with $f^\mc{P},f^{\mc{P}'}$ the appropriate\footnote{I.e., chosen in such a way that $\frac{i}{2}\bar{z}\delta z+\delta f^\mc{P}$ vanishes along $\mc{P}$ and $-\frac{i}{2}\bar{z}\delta z+\delta f^{\mc{P}'}$ vanishes along $\mc{P}'$.} boundary terms at $t_0$, $t_2$. It follows that the partition function for the glued interval $I\cup I'=[t_0,t_2]$ is:
\begin{multline}\label{SB measure from gluing}
Z_{I\cup I'}(v_\oo, u_\ii) = \int \D z(t)\,\D\bar{z}(t)\; e^{\frac{i}{\hbar}(S^\pol_I+\frac{i}{2}z\bar{z}\big|_{t_1}+S^\pol_{I'})} =\\
= \int_\CC \frac{i\, dz\, d\bar{z}}{4\pi\hbar }\; Z_{I'}(v_\oo,\bar{z})\; e^{-\frac{z\bar{z}}{2\hbar}}  \; Z_I (z,u_\ii) .
\end{multline}
In the first integral, the term $\frac{i}{2}z\bar{z}\big|_{t_1}$ compensates the boundary term of $S^\pol_{I'}$ at $t_1$.  The final integral is over the values of $z,\bar{z}$ at $t=t_1$. Also, we denoted $u_\ii$ a coordinate parametrizing the space of leaves of the polarization $\mc{P}$ and similarly $v_\oo$ a coordinate paramterizing the space of leaves of $\mc{P}'$. In (\ref{SB measure from gluing}), we see the Segal--Bargmann measure $\frac{i\,dz\,d\bar{z}}{4\pi\hbar}\,e^{-\frac{|z|^2}{2\hbar}}$ , cf.\ (\ref{SB inner product}), appearing.  
The normalization factor $\frac{i}{4\pi\hbar}$ is chosen in such a way that, choosing $\mc{P}'=\mr{Span}\{\frac{\dd}{\dd \bar{z}}\}$, we have 
$$Z_{I\cup I'}(z_\oo,u_\ii)=Z_I(z_\oo,u_\ii) , $$
which in turn follows from the identity
$$ \int \frac{i}{4\pi\hbar} dz\, d\bar{z}\; \underbrace{e^{\frac{z_\oo \bar{z}}{2\hbar}} }_{Z_{I'}(z_\oo,\bar{z})} e^{-\frac{|z|^2}{2\hbar}}\Psi(z) = \Psi(z_\oo) $$
true for any holomorphic function $\Psi(z)$ (for which the l.h.s. converges), applied to $\Psi(z)=Z_I(z,u_\ii)$.

\subsection{Aside: contour integration in the complexified space of fields, a lattice toy model}
\label{App: real contour}
Throughout the paper we are dealing with complexified phase spaces (so that we can impose the convenient holomorphic/antiholomorphic polarizations) and complexified spaces of fields where the path integral should be understood as an integral over a real contour. 

A toy model is provided by topological quantum mechanics $S=\int_I pdx$, as above. The phase space for a point is $\Phi=\RR^2$. We consider the model with boundary polarization $\mr{Span}\left\{ \frac{\dd}{\dd z} \right\}$ (i.e., $\bar{z}$ fixed) at $t=0$ and $\mr{Span}\left\{ \frac{\dd}{\dd \bar{z}} \right\}$ (i.e., $z$ fixed) at $t=1$; these polarizations are defined 
on the complexified phace space $\Phi_\CC=\CC^2$. The action modified by the appropriate boundary term is $S^f=\int_I (-\frac{i}{2}\bar{z} dz) -\frac{i}{2} z\bar{z}\big|_{t=0} $.

The path integral for this model can be presented by a lattice approximation (which happens to be exact):
\begin{multline}\label{lattice contour integral}
Z(z_\oo=z_N,\bar{z}_\ii=\bar{z}_0)=\int_{C\subset Y_\CC}  \prod_{k=1}^{N-1} \frac{i\, dz_k d\bar{z}_k}{4\pi\hbar}\;\;  e^{  \frac{1}{2\hbar}  
\big(\sum_{k=1}^N(z_k-z_{k-1})\bar{z}_{k-1}+z_0 \bar{z}_0\big)} \\
= \int_{C\subset Y_\CC}  \prod_{k=1}^{N-1} \frac{i\, dz_k d\bar{z}_k}{4\pi\hbar}\;\;  e^{\frac{1}{2\hbar} \big( 
z_1 \bar{z}_0 - z_1 \bar{z}_1 + z_2 \bar{z}_1 - z_2\bar{z}_2 +\cdots - z_{N-1} \bar{z}_{N-1}+ z_N \bar{z}_{N-1}
\big)} .
\end{multline}
Here: 
\begin{itemize}
\item We understand that the interval $I=[0,1]$  is partitioned   into $N\geq 1$ smaller intervals $[t_0=0,t_1]$, $[t_1,t_2]$, \ldots, $[t_{N-1},t_N=1]$.  
\item The space $Y_\CC =(\CC^2)^{N-1}$  --- the 
fiber of the complexified space of fields over boundary conditions --- is the product of complexified phase spaces corresponding to $t_1,\ldots,t_{N-1}$. In particular, we understand that $z_k$ and $\bar{z}_k$ are independent complex variables: they do not have to be complex conjugates of each other.
\item The integration is over a ``contour'' $C\subset Y_\CC$ --- a real $2(N-1)$-dimensional subspace. In particular, the integrand in (\ref{lattice contour integral}) is a holomorphic $2(N-1)$-form on $Y_\CC$ pulled back to $C$ by the inclusion.
\end{itemize}

For the contour $C$, we can consider the following two examples:
\begin{enumerate}[(i)]
\item Contour $C_1$ given by $\bar{z}_k=z_k^*$, $k=1,\ldots,N-1$, where $*$ is the complex conjugation.
\item Contour $C_2$ given by reality conditions $z_k\in \RR\subset\CC$, $\bar{z}_k\in i\RR\subset \CC$.
\end{enumerate}
For $C=C_1$, (\ref{lattice contour integral}) is an absolutely convergent Gaussian integral (for arbitrary boundary conditions) and yields
$$ Z(z_\oo, \bar{z}_\ii) = e^{   \frac{1}{2\hbar}  
\bar{z}_\ii z_\oo} . $$ 
For $C=C_2$, (\ref{lattice contour integral}) is an oscillatory Fresnel integral which is only conditionally convergent and even that only under special assumptions on the boundary conditions ($\bar z_\ii\in i\RR, z_\oo\in \RR$).\footnote{
Absolute/conditional convergence property is particularly clear in the simplest case $N=2$: here (\ref{lattice contour integral}) becomes $\int_{C\subset \CC^2} \frac{i\, dz_1 d\bar{z}_1}{4\pi\hbar}\; e^{\frac{1}{2\hbar}(z_1 \bar{z}_0 -z_1 \bar{z}_1 +z_2 \bar{z}_1)}$.
}
When the integral over $C_2$ converges, its value coincides with the result of integration over $C_1$ (which is clear, e.g., from a contour deformation argument).

In summary, we have the complexified space of fields of the lattice theory $F^I_{\CC}\ni (\bar{z}_0,z_1,\bar{z}_1,\ldots,z_{N-1},\bar{z}_{N-1},z_N)$ fibered over the complex space of boundary conditions $\mc{B}^{\dd I}_\CC \ni \{\bar{z}_\ii,z_\oo\}$ with complex fiber $Y_\CC$ (lattice fields with zero boundary conditions), and the integration in the lattice path integral (\ref{lattice contour integral}) is over a contour $C\subset Y_\CC$ --- a half-dimensional real submanifold.

\section{Kodaira--Spencer theory}\label{app:KS}
We briefly review the definition of the Kodaira--Spencer action functional that was introduced in \cite[Section 5]{BCOV}, where it was used to analyze the target space physics of the $B$-model. See also \cite[Section 2.1]{GS}.  
\subsection{Some operations on complex forms}
Let $M$ be a 6-dimensional Calabi--Yau manifold with a reference holomorphic 3-form $\omega_0$ (sometimes the pair $(M,\omega_0)$ is called a gauged Calabi--Yau manifold). We denote by $\Omega^{p,q}(M)$ complex forms of Hodge type $(p,q)$ --- sections of the bundle $\wedge^p (T_\CC^*M)^{1,0} \otimes \wedge^q (T^*_\CC M)^{0,1}$ and by $\Omega^{-p,q}(M)$ sections of the bundle $\wedge^p (T_\CC M)^{1,0} \otimes \wedge^q(T^*_\CC M)^{0,1}$, i.e., $(0,q)$-forms with values in $(p,0)$-vector fields. Contraction with the reference holomorphic 3-form provides a map 
\begin{align*}
\Omega^{-p,q}(M) &\to \Omega^{3-p,q}(M) \\ 
A &\mapsto A^\vee = A\omega_0
\end{align*}
(we omit the symbols for wedge products and contractions). For a $(p,q)$-form $A$ with $p\geq 0$, we set $A^\vee = A\omega_0^{-1} \in \Omega^{p-3,q}$, in particular we have $(A^\vee)^\vee = A$.  For $A,B,C \in \Omega^{-1,1}(M)$, we further define the operations
\begin{align}
A^\vee \circ B^\vee & = (AB)^\vee= (AB)\omega_0 &\in \Omega^{1,2}(M), \notag \\
\langle A^\vee,B^\vee,C^\vee\rangle &= A^\vee(B^\vee\circ C^\vee) = A^\vee(BC)\omega_0 & \in \Omega^{3,3}(M), \notag \\
\langle A^3\rangle &= -\frac16 \frac{\langle A^\vee,A^\vee,A^\vee\rangle}{\omega_0\overline{\omega}_0} = \frac16 (A^3\omega_0)(\ol{\omega}_0)^{-1} & \in\Omega^{0,0}(M), \label{eq:mucube}
\end{align} 
and the same operations make sense for $\bar A,\bar B, \bar C\in \Omega^{1,-1}(M)$ if we replace $\omega_0$ by $\bar\omega_0$. The minus sign in \eqref{eq:mucube} ensures that $\langle A^3\rangle\ol\omega_0 = \frac16 A^3\omega_0$. Also, in the last equation we made use of the fact that one can divide by sections of a line bundle. 
By a lemma of Tian \cite{Tian}, if $\dd A = \dd B = 0$, we have 
\begin{equation}
[A,B]^\vee = \dd(A^\vee \circ B^\vee). \label{eq:Tian}
\end{equation}
\subsection{The generating function for Hitchin polarization} \label{app:genfctHitchin}
For completeness, we include here a derivation of the generating function \eqref{eq:genfctmain} for the transformation from the linear polarization to the nonlinear polarization. 
A complex 3-form $A$ has  decompositions $A = A^{+,\mr{l}} + A^{-,\mr{l}}$, with $A^{+,\mr{l}} \in \Omega^{3,0}(M) \oplus \Omega^{2,1}(M)$, and $A^{-,\mr{l}} \in \Omega^{1,2}(M) \oplus \Omega^{0,3}(M)$, and 
$$A = A^{+,\nl} + A^{-,\nl} , $$
where $A^{+,\nl}$ and $A^{-,\nl}$ are decomposable complex 3-forms. The $3$-form $A$ is called \emph{nondegenerate} if $A^{+,\nl} \wedge A^{-,\nl}$ is everywhere nonvanishing. We parametrize
\begin{align*}
A^{+,\nl} &= \rho e^\mu \omega_0 , \\
A^{-,\nl} &= \overline{\rho} e^{\overline{\mu}}\overline{\omega}_0 ,
\end{align*}
where $\rho,\ol{\rho} \in \Omega^0_\CC(M)$, $\mu \in \Omega^{-1,1}(M)$, $\ol{\mu}\in \Omega^{1,-1}(M)$ and 
\begin{equation} \rho e^\mu\omega_0 = \rho\left(\omega_0 + \mu \omega_0 + \frac{\mu^2}{2}\omega_0 + \frac{\mu^3}{6}\omega_0 \right) .
\label{eq:nlpoldef}\end{equation}
 To write the generating function from the linear to the nonlinear polarization, we use \begin{equation}
G(A^{+,\mr{l}},A^{-,\nl}) = 
\frac12 \int_M A^{+,\mr{l}}A^{-,\mr{l}} - A^{-,\nl}A^{+,\mr{nl}} .
\label{eq:G1}
\end{equation}
\begin{lemma}
In the variables $A^{3,0},A^{2,1},\ol\rho,\ol\mu$, the generating function is given by 
\begin{multline}
G(A^{3,0},A^{2,1},\overline{\rho},\overline{\mu}) =\\
=\int_M \overline{\rho}( A^{3,0}\overline{\omega}_0 + A^{2,1}\overline{\mu}\,\overline{\omega}_0) + \overline{\rho}^2\langle\overline{\mu}^3\rangle\omega_0\overline{\omega}_0 - 
\frac{\left\langle\left( (A^{2,1} - \frac12\overline{\rho}\,\overline{\mu}^2\overline{\omega}_0)^\vee\right)^3\right\rangle}{(A^{3,0})^\vee - \overline{\rho}\langle\overline{\mu}^3\rangle}\omega_0\overline{\omega}_0 . \label{eq:genfctapp}
\end{multline}
\end{lemma}
\begin{proof}
This is a tedious but straightforward computation. One way to do it is to express $G$ in terms of $\rho,\mu,\ol\rho,\ol\mu$ first. To this end, notice that the decomposition in \eqref{eq:nlpoldef} is a decomposition into forms of definite Hodge type. Thus, we can write 
\begin{align}
A^{3,0}
 &= {\rho}\,{\omega}_0+ \ol\rho\frac{\ol\mu^3}{6} \ol\omega_0  = \left(\rho + \ol\rho\langle\ol\mu^3\rangle\right)\omega_0,\label{eq:A30} \\
A^{2,1}
 &={\rho}\,{\mu}\,{\omega}_0 +  \ol\rho\frac{\ol\mu^2}{2}\ol\omega_0,  \label{eq:A21} 
\end{align}
and similarly for $A^{0,3}$ and $A^{1,2}$.
We thus obtain 
\begin{align}
A^{3,0}A^{0,3} &= \rho\ol\rho(1+ \langle\mu^3\rangle\langle\ol{\mu}^3\rangle)\omega_0\ol\omega_0 + \left(\rho^2\langle\mu^3\rangle\ + \ol\rho^2\langle\ol\mu^3\rangle\right)\omega_0\ol\omega_0, \label{eq:A30A03}\\
A^{2,1}A^{1,2} &= \rho\ol{\rho}\left(\mu\omega_0\ol\mu\,\ol\omega_0  + \frac14\ol{\mu}^2\ol{\omega}_0\mu^2\omega_0\right) +\frac12\left(\rho^2\mu\omega_0{\mu}^2\omega_0 - \ol\rho^2\ol\mu\,\ol\omega_0\ol{\mu}^2\ol\omega_0\right). \label{eq:A+l+A-l}
\end{align}
On the other hand, we have 
\begin{equation}
A^{-,\nl}A^{+,\nl} = \rho\ol\rho\left(\ol\omega_0\omega_0 + \ol\mu\,\ol\omega_0\mu\omega_0 + \frac14\ol\mu^2\ol\omega_0\mu^2\omega_0 + \langle{\ol{\mu}^3}\rangle\langle\mu^3\rangle \omega_0\ol\omega_0\right). \label{eq:A+nlA-nl}
\end{equation}
Summing \eqref{eq:A30A03} and \eqref{eq:A+l+A-l} and subtracting \eqref{eq:A+nlA-nl},  the last two terms in \eqref{eq:A+nlA-nl} cancel and we obtain
\begin{multline}
A^{+,\mr{l}}A^{-,\mr{l}} - A^{-,\nl}A^{+,\nl} = 2\rho\ol\rho\left(\omega_0\ol\omega_0 -\ol\mu\,\ol\omega_0\mu\omega_0\right) \\+ \rho^2\left(\langle\mu^3\rangle\omega_0\ol\omega_0 + \frac12 \mu\omega_0  \mu^2\omega_0\right) 
+\ol\rho^2\left(\langle\ol\mu^3\rangle\omega_0\ol\omega_0-\frac12\ol\mu\,\ol\omega_0\ol\mu^2\ol\omega_0\right). \label{eq:difference}
\end{multline}
Recall that $\mu\omega_0 \mu^2\omega_0 = (\mu^3\omega_0)\omega_0 = 6\langle\mu^3\rangle\ol\omega_0\omega_0$, hence we can simplify this expression to 
$$A^{+,\mr{l}}A^{-,\mr{l}} - A^{-,\nl}A^{+,\nl} = 2\rho\ol\rho\left(\omega_0\ol\omega_0 -\ol\mu\,\ol\omega_0\mu\omega_0\right) -2 \left(\rho^2\langle\mu^3\rangle +\ol\rho^2\langle\ol\mu^3\rangle\right)\omega_0\ol\omega_0.$$
From equations \eqref{eq:A30}, \eqref{eq:A21} we get 
\begin{align*}
\rho\omega &= A^{3,0} - \ol\rho\langle\ol\mu^3\rangle\omega_0, \\
\rho\mu\omega_0 &= A^{2,1} - \frac12\ol\rho\, \ol\mu^2\ol\omega_0,
\end{align*}
which we use to rewrite the first term as 
$$2\rho\ol\rho\left(\omega_0\ol\omega_0 -\ol\mu\,\ol\omega_0\mu\omega_0\right)=2\ol\rho A^{3,0}\ol\omega_0 - 2\ol\rho^2\langle\ol\mu^3\rangle\omega_0 + 2\ol\rho A^{2,1}\mu\omega_0 + 6\ol\rho^2\langle\ol\mu^3\rangle\omega_0\ol\omega_0.$$
In total, \eqref{eq:G1} evaluates to 
\begin{equation}G = \int_M \ol\rho A^{3,0}\ol\omega_0 + \ol\rho A^{2,1}\ol\mu\,\ol\omega_0 + \ol\rho^2\langle\ol\mu^3\rangle\omega_0\ol\omega_0 - \frac{\langle(\rho\mu)^3\rangle}{\rho}\omega_0\ol\omega_0.\label{eq:genfctappproof}
\end{equation}
Formula \eqref{eq:genfctapp} may now be obtained by 
solving equations \eqref{eq:A30},\eqref{eq:A21} for $\rho$ and $\rho\mu$, which gives \begin{align}
\rho &= \frac{A^{3,0} - \frac16\overline{\rho}\,\overline{\mu}^3\overline{\omega}_0}{\omega_0} = (A^{3,0})^\vee -  \ol\rho\langle\ol{\mu}^3\rangle, \label{eq:rho} \\
\rho\mu &= \frac{A^{2,1} - \frac12\overline{\rho}\,\overline{\mu}^2\overline{\omega}_0}{\omega_0} = \left(A^{2,1} - \frac12\overline{\rho}\,\overline{\mu}^2\overline{\omega}_0\right)^\vee .  \label{eq:rhomu}
\end{align}
Plugging \eqref{eq:rho},\eqref{eq:rhomu} into \eqref{eq:genfctappproof} we obtain \eqref{eq:genfctapp}. 
\end{proof}
The defining property of $G$ is the following. 
\begin{lemma} 
We have $\delta G = \theta^\mr{l} - \theta^\nl$ where 
$\theta^\mr{l} = A^{-,\mr{l}}\delta A^{+,\mr{l}}$ and $\theta^\nl = A^{+,\nl}\delta A^{-,\nl}$.
\end{lemma}
\begin{proof}
This follows from Equation \eqref{eq:G1}.
But one can also check it through direct computation: we have 
$$\frac{\delta G}{\delta A^{3,0}} = \ol\rho\,\ol\omega_0 + \frac{\left\langle\left(\left(A^{2,1} -\frac12\ol\rho\,\ol\mu^2\ol\omega\right)^\vee\right)^3\right\rangle}{(A^{3,0})^\vee - \ol\rho\langle\ol\mu^3\rangle^2}(\omega_0)^{-1}\omega_0\ol\omega_0 = \ol\rho\,\ol\omega_0+ \frac{\langle(\rho\mu)^3\rangle}{\rho^2}\ol\omega_0 = A^{3,0}. $$
Notice that we have $$\frac{\delta}{\delta \mu}\langle\mu^3\rangle = \frac12 (\mu^2\omega)(\ol\omega)^{-1}.$$ It follows that 
$$\frac{\delta}{\delta A^{2,1}} \langle ((A^{2,1})^\vee)^3\rangle = \frac12 (((A^{2,1})^\vee)^2\omega_0)\ol\omega_0^{-1}\omega_0^{-1} = -\frac{\frac12(((A^{2,1})^\vee)^2\omega_0)}{\omega_0\ol\omega_0^{-1}}$$(note the sign)
and therefore
$$\frac{\delta G}{\delta A^{2,1}} = \ol\rho\,\ol\mu\,\ol\omega  + \frac12\rho\mu^2\omega_0 = A^{1,2}.$$ 
This proves that $\delta G/\delta A^{+,\mr{l}} = A^{-,\mr{l}}$. 
Computing $\delta G /\delta\ol\rho$ gives 
\begin{align*}
\frac{\delta G}{\delta\bar\rho}\delta\ol\rho &= \delta\ol\rho \left(A^{3,0}\ol\omega_0 + A^{2,1}\ol\mu\,\ol\omega_0 + 2\ol\rho\langle\ol\mu^3\rangle\omega_0\ol\omega_0 + \frac12\rho\mu^2\omega_0(\frac12\ol\mu^2\ol\omega_0) +  \rho\langle\mu^3\rangle\omega_0\langle\ol\mu^3\rangle\ol\omega_0\right) \\
&= -\rho e^\mu\omega_0e^{\ol\mu}\ol\omega_0\delta\ol\rho.
\end{align*}
Finally, computing $\delta G/\delta \ol\mu$ gives 
\begin{align*}
\frac{\delta G}{\delta\ol\mu}\delta\ol\mu &= - \ol\rho A^{2,1}(\delta\ol\mu\, \ol\omega_0) + \frac12\ol\rho^2(\ol\mu^2\ol\omega_0)(\delta\ol\mu\,\ol\omega_0) + \frac12\rho\mu^2\omega_0(\ol\rho\,\ol\mu\,\delta\ol\mu\,\ol\omega_0) + \rho\langle\mu^3\rangle(\delta\ol\mu\,\ol\omega_0)\frac12\ol\rho\,\ol\mu^2\ol\omega_0  \\
&= -\ol\rho\rho\,\mu\,\omega_0 (\delta\ol\mu\,\ol\omega_0) + \frac12\rho\mu^2\omega_0(\ol\rho\,\ol\mu\,\delta\ol\mu\,\ol\omega_0) + \rho\langle\mu^3\rangle(\delta\ol\mu\,\ol\omega_0)\frac12\ol\rho\,\ol\mu^2\ol\omega_0  \\
 & = -\rho\ol\rho (e^\mu\omega_0)\delta\ol\mu(\ol\omega_0 + \ol\mu\omega_0 + \frac12\ol\mu^2\ol\omega_0). 
\end{align*}
Using 
$$\delta A^{-,\nl} = \delta \ol\rho e^{\ol\mu}\ol\omega_0 = e^{\ol\mu}\ol\omega_0\delta\ol\rho + \ol \rho\delta\ol\mu(\ol\omega_0 + \ol\mu\omega_0 + \frac12\ol\mu^2\ol\omega_0),$$ we obtain 
$$\delta G = A^{+,\mr{l}}\delta A^{-,\mr{l}} - A^{+,\nl}\delta A^{-,\nl}.$$
\end{proof}
\subsection{Deformations of complex structures}
Let $M$ be a compact Calabi--Yau manifold supplied with a reference holomorphic 3-form $\omega_0$. A deformation of the complex structure is equivalent to a deformation of the $\bar\dd$ operator $\bar\dd \to \bar\dd_{\bar{A}} = \bar\dd + \bar{A}$, where $\bar{A} \in \Omega^{-1,1}(M) = \Gamma(T^{1,0}M \otimes (T^*)^{0,1}M)$. The integrability condition $\bar\dd_A^2 = 0$ is equivalent to (\cite{KS}) 
\begin{equation}
\bar\dd\bar{A} + \frac12 [\bar{A},\bar{A}] = 0.\label{eq:KS}
\end{equation} 
The moduli space of complex structures is thus given by solutions of \eqref{eq:KS} modulo gauge transformations 
\begin{equation}
\delta\bar{A} = \bar\dd\varepsilon + [\bar{A},\varepsilon], \label{eq:KS gauge trafo}
\end{equation}
with $\varepsilon \in \Omega^{-1,0}(M)$. The tangent space to the moduli space of complex structures is given by the linearization of \eqref{eq:KS}, i.e., it is the quotient of $\{\alpha\colon\bar\dd \bar{\alpha} = 0\}$ by linearized gauge transformations $\delta\alpha = \bar\dd\varepsilon$. \\
After Tian (\cite{Tian}), this problem can be reformulated using $\bar{A}^\vee$ as follows. Imposing the constraint $\dd\bar{A}^\vee =0$ and using \eqref{eq:Tian}, we can rewrite \eqref{eq:KS} as 
\begin{equation}
\bar\dd \bar{A}^\vee + \dd (\bar{A}^\vee \circ \bar{A}^\vee) = 0. \label{eq:KS2}
\end{equation}
\subsection{Kodaira--Spencer action}
The Kodaira--Spencer action functional as introduced in \cite{BCOV} is 
\begin{equation}
S_{KS}[\bar{A}^\vee]= \int_M\frac 12 \bar{A}^\vee \dd^{-1}\bar\dd\bar{A}^\vee + \frac16 \langle \bar{A}^\vee,\bar{A}^\vee,\bar{A}^\vee \rangle.\label{eq:KSaction1}
\end{equation}
Here the first term is well-defined due to $\dd\bar\dd$-lemma. The equation of motion of \eqref{eq:KSaction1} is  \eqref{eq:KS2}. One can resolve the nonlocality by writing $\bar{A}^\vee = x + \dd b$, where $x$ is a $\dd$-harmonic $(2,1)$-form. The action functional then becomes 
\begin{equation}
S_{KS}(x;b) = \int_M \frac12 \dd b\bar\dd b + \frac 16 \langle (x + \dd b),(x + \dd b),(x + \dd b)\rangle .\label{eq:KSaction2}
\end{equation}
This action functional has the following remarkable property. 
From eq.\ \eqref{eq:KS2}, it follows that any harmonic $(2,1)$-form $x = \bar{A}^\vee_1$ can be interpreted as a first order deformation of the complex structure. The tree level diagrams of \eqref{eq:KSaction2} then generate forms $\bar{A}^\vee_n$ with the property that $\bar{A}^\vee = \sum \varepsilon^n \bar{A}^\vee_n$ is a solution of the Kodaira--Spencer equation \eqref{eq:KS2}. We refer to \cite[Section 5.2]{BCOV} for details.

\end{document}